\documentclass[11pt]{article}
\usepackage[margin=1in]{geometry}
\newcommand{\keywords}[1]{%
  \par\noindent\textbf{Keywords:} #1
}
\usepackage{natbib}

\usepackage{comment,tabularx}
\usepackage{amsmath}%,romannum}
\usepackage[font=footnotesize,labelfont=bf]{caption}

\usepackage{longtable}
\usepackage{times}
\usepackage{bm}
\usepackage{amssymb}
\usepackage{graphicx}%http://ctan.org/pkg/graphicx
\usepackage{natbib}
\usepackage[plain,noend]{algorithm2e}
\usepackage{caption}
\usepackage{breqn}
\usepackage{multirow}
\usepackage{xcolor}
\usepackage{array}
\usepackage{tabularray}
\usepackage{diffcoeff}
\usepackage{subfigure}
\usepackage{amsthm}
\usepackage{appendix}
\usepackage{placeins}
\usepackage{hyperref}
\usepackage{url}

\newtheorem{proposition}{Proposition}
\newtheorem{lemma}{Lemma}
\newtheorem{assumption}{Assumption}
\newtheorem{remark}{Remark}
\usepackage{enumitem,color,verbatim,pgf,pgfarrows,pgfnodes}
\usepackage{tikz}
\usetikzlibrary{arrows.meta}

\def\bSig\mathbf{\Sigma}

\newcommand*{\rom}[1]{\expandafter\@slowromancap\romannumeral #1@}
\usepackage{makecell}
\usepackage{booktabs}
\usepackage{adjustbox}
\setcellgapes{2pt}

\newcommand{\R}{\mathbb{R}}
\newcommand{\s}{\mathbf{s}}
\newcommand{\X}{\mathbf{X}}
\newcommand{\Y}{\mathbf{Y}}
\newcommand{\D}{\mathbf{D}}
\newcommand{\btheta}{\boldsymbol{\theta}}
\newcommand{\bbeta}{\boldsymbol{\beta}}
\DeclareMathOperator*{\argmax}{arg\,max}
\DeclareMathOperator*{\argmin}{arg\,min}

\begin{document}

\title{Mixture-based Nonparametric Estimation of Spatial Covariance Functions with Applications to HIV Key Population Size Estimation across Sub-Saharan Africa}

\author{
Manushi Siriwardana, Hyebin Song, Le Bao, and Stephen Berg\\
Department of Statistics\\
Pennsylvania State University\\
University Park, PA, USA
}
\date{}
\maketitle

\label{firstpage}

\begin{abstract}
Consistent data on the sizes of key populations, such as female sex workers (FSWs), are often scarce, particularly at the sub-national level. Accurate size estimates are critical to effectively allocate resources and achieve HIV targets. Since FSW population sizes may be spatially correlated across areas, models that account for spatial dependence can improve estimation. An important component of such models is the covariance function, which characterizes the spatial dependence structure of the underlying process. In this work, we study spatial covariance functions to estimate FSW population sizes in Sub-Saharan Africa (SSA). Many spatial models rely on parametric covariance functions. However, parametric estimation can suffer from model mis-specification, potentially leading to inefficient or biased predictions. We therefore develop a robust non-parametric approach for estimating the covariance function of a stationary isotropic process in $\mathbb{R}^d$. We focus on a class of covariance functions that are valid in all dimensions, which includes popular kernels such as the exponential and Mat\'{e}rn kernels. Leveraging the fact that such covariance functions can be represented as infinite mixtures of scaled Gaussian kernels, we propose two estimation methods: weighted least squares and nonparametric maximum likelihood estimation to estimate the mixing measure of scaled Gaussian kernels. We also develop computationally efficient methods to solve these optimization problems using non-negative least squares and second-order descent updates. We evaluate the proposed methods through simulations and apply them to estimate the FSW population sizes at the sub-national level in SSA.
\end{abstract}

\keywords{
Key Population Size Estimation,
Spatial Covariance Estimation,
Stationary Isotropic Processes,
Gaussian Scale Mixtures,
Nonparametric Maximum Likelihood
}

%  As usual, the \maketitle command creates the title and author/affiliations
%  display 

\maketitle

\section{Introduction}
\label{sec:Int}

Key populations related to human immunodeficiency virus (HIV) infection are groups of people who are at a higher risk of being exposed to, acquiring, and transmitting HIV \citep{Xu2022}. In 2022, over half (55\%) of new HIV infections were concentrated among key populations and their sexual partners \citep{UNAIDS2024HIV}. These key populations include female sex workers (FSWs), cisgender men who have sex with men (MSM), people who inject drugs (PWID), and transgender women. Such groups of people may face barriers to accessing HIV treatments for various reasons, including discrimination, financial constraints, and criminalization \citep{Shannon2018}. Consequently, many countries in Sub-Saharan Africa (SSA) experience a high burden from the HIV epidemic. Key populations comprised 25\% of new HIV infections in SSA in 2022 \citep{UNAIDS2024HIV}.
The prevalence of HIV is particularly high among female sex workers and has not declined significantly over the years \citep{Shannon2018}. Criminalization of sex workers leads to restricted access and retention in treatments. \cite{Lyons2020} found that criminalizing and non-protective sex work laws lead to higher prevalence of HIV among female sex workers in SSA countries. Nevertheless, the HIV crisis among sex workers continues to receive inadequate attention \citep{Shannon2018}. Therefore, estimating the population sizes of FSWs in SSA is crucial to developing the necessary policies to prevent and treat HIV.

Accurate estimates of population sizes of key populations are critical for resource allocation and for achieving national and global HIV goals \citep{AbdulQuader2015}. However, there is a limited amount of reliable, consistent and timely data on the sizes of key populations, particularly at the sub-national level \citep{Abhirup2019}. Estimation of the population sizes of key populations can be demanding because these populations are difficult to reach and traditional methods of estimation are subject to bias \citep{Okal2013, Xu2022}. Therefore, data-driven  methods can be used to obtain final population size estimates (PSEs) by considering available non-sensitive data and estimates obtained from traditional methods while balancing the strengths and limitations of these methods \citep{Okal2013, Xu2022, Laga2023, Oliver2024}. Furthermore, the use of appropriate data-driven methods can be useful in deriving a single final estimate by considering multiple estimates observed at the same location and obtaining estimates for locations without any observed data \citep{Laga2023}.

Many statistical approaches used to produce PSEs employ Bayesian hierarchical models with spatial random effects, often specified through neighborhood-based areal models such as conditional autoregressive (CAR) structures~\citep{Laga2023, Oliver2024}. While these models are useful for borrowing information across spatial units, their induced spatial dependence is determined by a prespecified neighborhood or weight matrix. Consequently, the resulting estimates can be sensitive to the choice of areal units, definition of boundaries, and neighboring relationships~\citep{Earnest2007}. Moreover, CAR-type models often induce dependence between non-neighboring areas through paths in the graph, but this graph-based distance does not necessarily correspond to geographic distance. Therefore, prediction at a new location requires defining how the new location is connected to the existing neighborhood graph, which is not always straightforward.

In this paper, we take a covariance function based approach for modeling spatial dependence in PSEs. Covariance functions describe spatial dependence directly as a function of geographic distance and naturally support prediction at unobserved locations through kriging. The performance of spatial interpolation methods such as kriging can depend significantly on the selected covariance function \citep{Christakos1984, Barry1996}. Parametric methods, while widely used, can suffer from model mis-specification, leading to inefficient or biased predictions and misleading inferences \citep{Choi2013}. Hence, we develop a nonparametric estimator for stationary and isotropic covariance functions. 

Our approach leverages the fact that a stationary isotropic covariance function valid in all dimensions can be represented as an infinite mixture of scaled Gaussian kernels. We directly target estimating the mixing measure by approximating it with a finite discrete measure, and consider both weighted least squares (WLS) estimation and nonparametric maximum likelihood (NPML) estimators. Furthermore, we develop computationally efficient algorithms for the resulting constrained optimization problems,  using non-negative least squares for the WLS estimator and second-order descent updates for the NPML estimator. We use the proposed nonparametric approach to obtain the PSEs of FSWs in SSA at the sub-national level, and compare it with most frequently used parametric models, another nonparametric method, and the conditional autoregressive integrated nested Laplace approximation (CAR-INLA) model which treats spatial locations using spatially correlated random effects. The CAR-INLA model is used in related literature to predict logit proportions of key populations \citep{Oliver2024}. Therefore, in our study, we consider it to evaluate the performance of the proposed nonparametric covariance estimation method relative to a benchmark model while comparing it with other covariance-based methods.

\subsection{Related Work}

Estimating the population sizes of key populations has received increasing attention over the years \citep{Okal2013, Abhirup2019, Xu2022}. 
Existing approaches include traditional methods, such as capture-recapture, multiplier, Delphi, mapping, workbook method, respondent-driven sampling, and network scale-up methods, as well as data-driven approaches, including Bayesian estimation, stochastic simulation, and Laska, Meisner, and Siegel (LMS) estimation~\citep{Xu2022}.
Each method comes with limitations. Bayesian methods are vulnerable to subjectivity due to their dependence on prior distributions,  stochastic simulation methods can be impacted by the quality of data used in calibrating the  models, and the LMS method may demonstrate comparatively lower accuracies since it uses information only from one sample~\citep{Xu2022}. Therefore, multiple methods are often combined to overcome the limitations of traditional methods and optimize their respective strengths in producing a final PSE \citep{AbdulQuader2015, Xu2022}. For example, \cite{Okal2013} combined estimates from the multiplier, Wisdom of the Crowds, and literature-based methods with a surveillance survey to estimate the population size of MSM, FSWs, and intravenous drug users (IDUs) in Nairobi, Kenya, while \cite{Abhirup2019} used Bayesian hierarchical regression to obtain estimates of MSM population sizes in sub-national areas of Côte d’Ivoire using initial estimates obtained from several traditional methods. Despite substantial progress, the availability of population size estimates in SSA remains inconsistent and uneven across different countries and key populations \citep{Viswasam2020}. 

Among different key populations, FSWs generally represent the most studied area of key population research in SSA. Nevertheless, FSW PSEs are not available for all countries in SSA, and existing estimates are produced using heterogeneous methods \citep{Viswasam2020, Oliver2024}. \cite{Laga2023} fitted a Bayesian linear hierarchical model to integrate FSW PSEs and uncertainties in estimates in sub-national areas in SSA and Madagascar. Their spatial model used independent and identically distributed Gaussian random effects for the estimation method and country effects. Consequently, they address spatial structure solely using country-level random effects. This accounts for clustering within a country; however, it ignores spatial correlation structure across countries since random effects are assumed to be independent and identically distributed draws from a Gaussian distribution centered at zero. \cite{Oliver2024} used a Bayesian mixed-effects spatial regression to model logit-transformed PSEs of key populations at the provincial level, incorporating effects for the survey method, spatially correlated random effects between neighboring countries and provinces, and a study-level random effect that allowed capturing relationships among PSEs from the same study. Accounting for spatial structure solely based on spatially correlated random effects between neighboring entities assumes that the spatial covariance structure occurs only due to neighbors, and can be an ineffective strategy when there is insufficient data on nearby locations.

The use of spatial covariance functions allows understanding and capturing local as well as long range spatial relationships in data. Spatial correlation functions have been used in epidemiology and public health applications~\citep{Abhirup2019, Qu2023}. For example, \cite{Abhirup2019} used kriging to impute missing values in HIV prevalence data, which was one of the covariates used in their study. They modeled HIV prevalence using a Gaussian Process (GP) with a constant mean and an exponential covariance function. \cite{Qu2023} studied HIV areal count data with missing values and, under a likelihood-based approach, specified a spatiotemporal model for the complete data with several spatial correlation structures, including homogeneous, exponential, and CAR correlation functions. In public health and epidemiological research, spatial covariance functions, when used, are often estimated using parametric methods \citep{Abhirup2019, Qu2023}. However, the use of parametric methods can suffer from model mis-specification and complications in optimizations leading to biased predictions \citep{Stein1999, Warnes1987, Diggle2007, Geoga2023}. For example, the Mat\'ern smoothness parameter controls the differentiability of the spatial process and plays an important role in interpolation, although it may be difficult to estimate reliably in practice.

Nonparametric methods have been proposed in the literature for estimating spatial covariance functions to avoid various issues encountered in using parametric methods \citep{GentonGorsich2002, Choi2013}. Many approaches have been considered in the literature, including variogram-based estimators, kernel regression estimators, $Q_n$ estimators, tapered estimators, robust estimators, and estimators based on linear combinations of basis functions~\citep{Cressie1980,Guyon1982, RousseeuwCroux1992, Hall1994,Ma2000, Choi2013, Wang2023}. Among these approaches, methods that leverage spectral or mixture representations of valid covariance functions are most closely related to our work. They have received increasing attention because they ensure that the resulting covariance functions are positive definite \citep{GentonGorsich2002, Choi2013, Wang2023}. 
\cite{ShapiroBotha1991} proposed a constrained nonparametric estimator for stationary isotropic variograms by representing the corresponding covariance function as a non-negative linear combination of Bessel functions. \cite{GentonGorsich2002} developed a nonparametric method to estimate the covariance function based on optimal discretizations of the isotropic spectral representation of positive definite functions using Fourier–Bessel matrices. \cite{Huang2011} also used the spectral representation of isotropic variograms, but estimated the spectral density using spline methodology. \cite{Choi2013} considered estimation of covariance functions that are valid in all dimensions by using completely monotone functions and B-spline approximations of the mixing measure. In these methods, WLS is used to fit a valid representation to empirical variogram or covariogram estimates. 
More recently, \cite{Wang2023} proposed a sieve maximum likelihood estimator for covariance functions valid in all dimensions by approximating the density of the mixing measure with Bernstein polynomials. Our approach follows these general directions, but approximates the mixing measure directly as a finite discrete measure, a formulation commonly arising in nonparametric mixture estimation. Our approach is simpler to implement compared with the approach of \citet{Choi2013}, where the mixing measure was estimated as a monotonicity-constrained B-spline function, and the approach of \citet{Wang2023}, which used Bernoulli polynomials to represent a mixture density.

We also develop novel theoretical results for mixture covariance estimators. We show the existence of the NPML estimator over the space of measures, which has not been proven before, and prove that the NPML estimated measure can be represented as a discrete measure with finite support set. Furthermore, using results from discrete geometry, we show that when the dataset is collected on an $m\times m$ square lattice, the NPML estimator can be represented by a measure with order $m^2/\sqrt{\log m}$ points, which shows that the complexity of the mixing measure increases slightly more slowly than the sample size. Our results are flexible, holding for mixtures over a fairly general choice of covariance kernel function, and also when the mean function of the process is allowed to depend on spatially indexed covariates, as is typical in applications.

The remainder of this paper is organized as follows. Section~\ref{sec:Data_Description} introduces the motivating dataset and defines the problem. In Section~\ref{sec:Methods}, we focus on the proposed nonparametric method to estimate the spatial covariance functions. Section~\ref{sec:Simulation_Study} presents results from simulation studies conducted to evaluate the performance of the nonparametric method relative to parametric and nonparametric covariance-based methods and the CAR-INLA model. In Section~\ref{sec:Data_Analysis}, we apply the proposed nonparametric method along with other methods based on spatial covariance functions and the CAR-INLA model on the motivating dataset.

\section{Data Description}
\label{sec:Data_Description}
The HIV dataset used in this study consists of 787 PSEs of FSWs from 30 SSA countries between 2010 and 2023 obtained from \cite{Oliver2024}. The estimates have been compiled from studies conducted during the aforementioned period using four databases and four systematic reviews. Each record in the dataset contains details of an estimate derived from a traditional method used in the respective study, including study index, study area, country, year, study method, reference population, count estimate, and proportion estimate. Note that we model population proportion estimates as opposed to PSEs to allow comparison of PSEs in different settings \citep{Oliver2024}. We obtain proportion estimates using count estimates and area, gender, year-matched total population sizes presented by \cite{Oliver2024}. Furthermore, we consider logit-transformed proportions for modeling purposes similar to \cite{Laga2023}, to better satisfy the Gaussianity assumption underlying our model. Finally, we estimate the logit-transformed proportions of FSWs in SSA at sub-national level using the fitted models and use area-matched female population sizes aged 15-49 years in 2015 to obtain FSW PSEs.

Figure~\ref{fig:logit_proportion} shows the observed logit-transformed proportions of FSWs in different study areas in SSA. Multiple estimates of logit proportions that include replicates and repeated measurements are recorded at the same location in the dataset. Meanwhile, some sub-national areas do not seem to contain any estimates as shown in the figure. The logit proportions exhibit high variability, without a clear smooth pattern across space. This may indicate the importance of considering factors associated with logit proportions that may explain complex human behaviors. To correspond with the logit-transformed proportions of FSWs, we collect covariates data at sub-national level for all SSA countries. We use the covariates database used by \cite{Laga2023}. Consequently, we aim at estimating the PSEs of FSWs in SSA at sub-national level using spatial covariance functions while capturing heterogeneity using covariates data and accounting for spatial structure.

\begin{figure}[h!]
    \centering
    \includegraphics[width=0.56\textwidth]{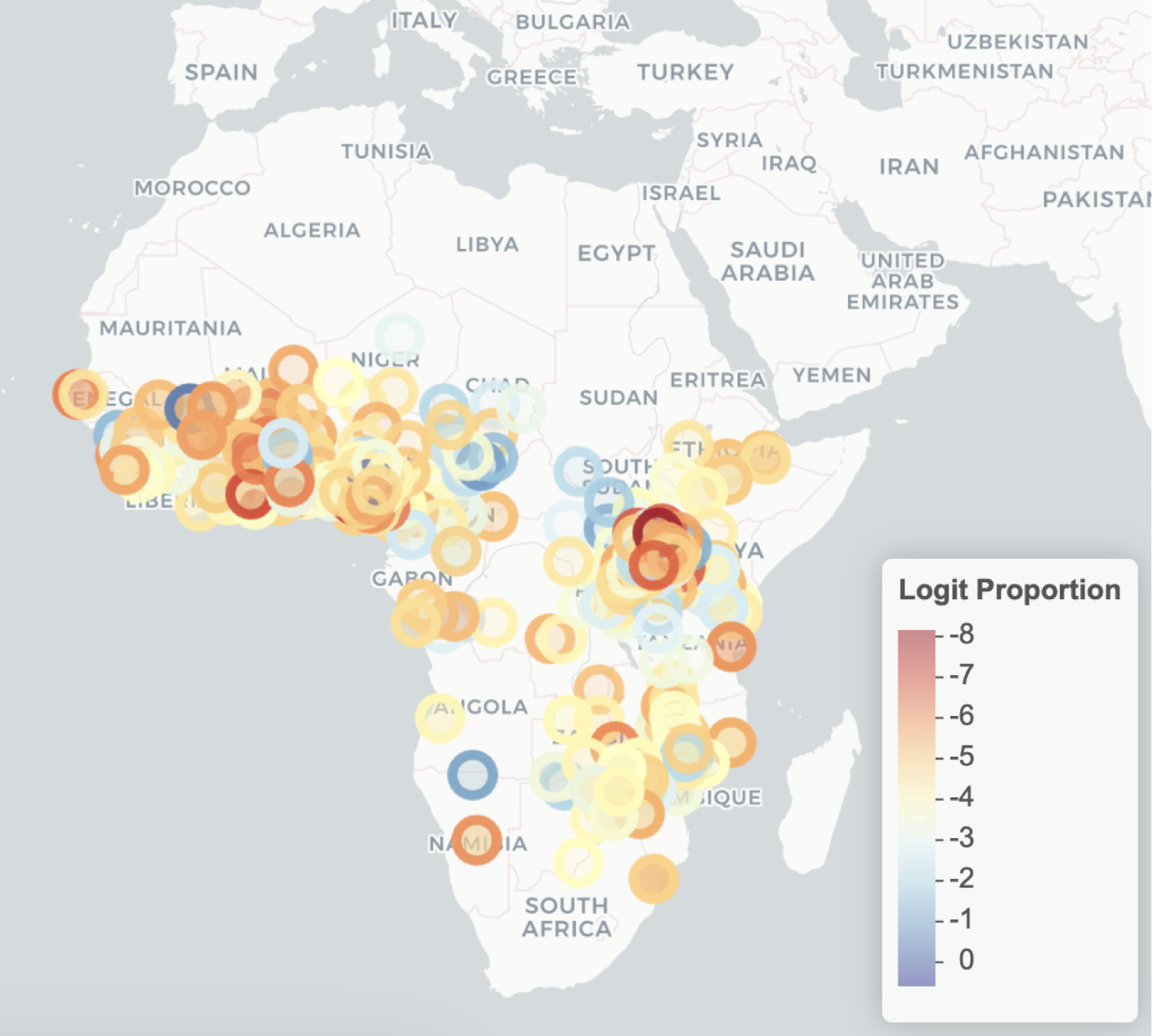}
    \caption{Observed logit transformed proportions of female sex workers in different study areas in sub-Saharan Africa.}
    \label{fig:logit_proportion}
\end{figure}

As noted earlier, the dataset contains multiple estimates of logit proportions observed at the same location. In cases with multiple estimates per location, we aggregate the observations to the level of unique locations. As the aggregation method, we use the arithmetic mean while assigning equal weights to each observation at a certain location assuming that each method used to obtain the estimates is equally informative. Given that $\tilde{Y}(\mathbf{s})_{j}$ is the $j^{th}$ estimate of the logit proportion observed at the location $\mathbf{s}$ where $\mathbf{s}\in \mathcal{D}$, $\mathcal{D}\subseteq \mathbb{R}$ is a domain of interest in $\mathbb{R}^d$, $n_{\mathbf{s}}$ is the number of estimates of the logit proportion observed at location $\mathbf{s}$ and $j=1,\dots,n_{\mathbf{s}}$, we obtain $\bar{Y}(\mathbf{s}) = \frac{1}{n_{\mathbf{s}}} \sum_{j=1}^{n_{\mathbf{s}}} Y(\mathbf{s})_j$. Consequently, the aggregated dataset contains 347 estimates of logit proportions observed at 347 unique locations.

\paragraph*{Covariates Data} We model the mean structure in the averaged logit proportions using an appropriate subset of covariates. Consequently, we consider year and province-matched reference female population sizes aged 15-49 years and administration area-matched percentages of urban population. A logarithmic transformation is applied to both these variables. We average the covariates data over $j=1,...,n_{\mathbf{s}}$ at each location $\mathbf{s}$ similar to logit proportions and obtain $\bar{\textbf{X}}(\mathbf{s}) = \frac{1}{n_{\mathbf{s}}} \sum_{j=1}^{n_{\mathbf{s}}} \tilde{\textbf{X}}(\mathbf{s})_j$.

Initially, we fit the full Bayesian hierarchical model using all covariates considered by \cite{Laga2023}, on the original unaggregated data. According to the fitted model, only the above two covariates appeared to be significant while also having the strongest estimated effects on the logit proportions. Therefore, we consider these two covariates in our analysis to improve interpretability while maintaining model simplicity.

\paragraph*{Spatial Resolution}
Since our main objective is to estimate the FSW PSEs in SSA at the sub-national level, it is important to map the estimates obtained from \cite{Oliver2024} to sub-national areas. We use the procedure followed by \cite{Laga2023} to map the estimates to administrative units using shapefiles and classify each estimate as city, sub-national, national, or miscellaneous. We use only the estimates from cities and sub-national areas in the analysis and consider the mapped administrative units as sub-national areas. Therefore, each sub-national area represents a unique location, $\mathbf{s}$.

We build the adjacency matrix for the CAR-INLA model using spatial polygons of sub-national areas mapped to estimates. Areas that share a common boundary or vertex are defined as neighbors. Meanwhile, covariance estimation methods capture spatial relationships based on pairwise distances between locations. Therefore, we use the coordinates of the centroid of the corresponding spatial polygon to represent the coordinates of each location, i.e., sub-national area.

\paragraph*{Checking the Validity of the Assumptions}
The proposed mixture-based nonparametric method assumes that the spatial process is stationary and that the covariance function is isotropic. These properties can be assessed by investigating the empirical variogram \citep{Handcock1994}. We use empirical directional variograms fitted to the residuals of a linear regression model to evaluate the assumption of isotropy.

The empirical directional variogram in a particular direction is obtained using pairs of residuals that are at most $45^\circ$ angle with the considered direction. For example, the north–south directional variogram considers pairs of residuals that are in the north-south direction and pairs with a connecting line that is at most $45^\circ$ with the north-south direction. We compare directional semivariograms in cardinal and non-cardinal directions with the isotropic semivariogram obtained from all pairs of residuals without considering a certain orientation. As shown in Figure~\ref{fig:directional_variogram}, the empirical directional semivariograms provide little evidence for the presence of significant deviations between the directional semivariograms and the isotropic semivariogram. Thus, there is little evidence of violation of the assumption of isotropy of the covariance function. Furthermore, the histogram and normal Q-Q plot of the logit-proportions indicated approximate normality (see Figures~\ref{fig:histogram_logit_proportions} and~\ref{fig:Qqplot_logit_proportions} in Appendix~\ref{sub: appendix_real_data}), supporting the use of a Gaussian process model.

\begin{figure}[h!]
    \centering
    \includegraphics[width=0.5\textwidth]{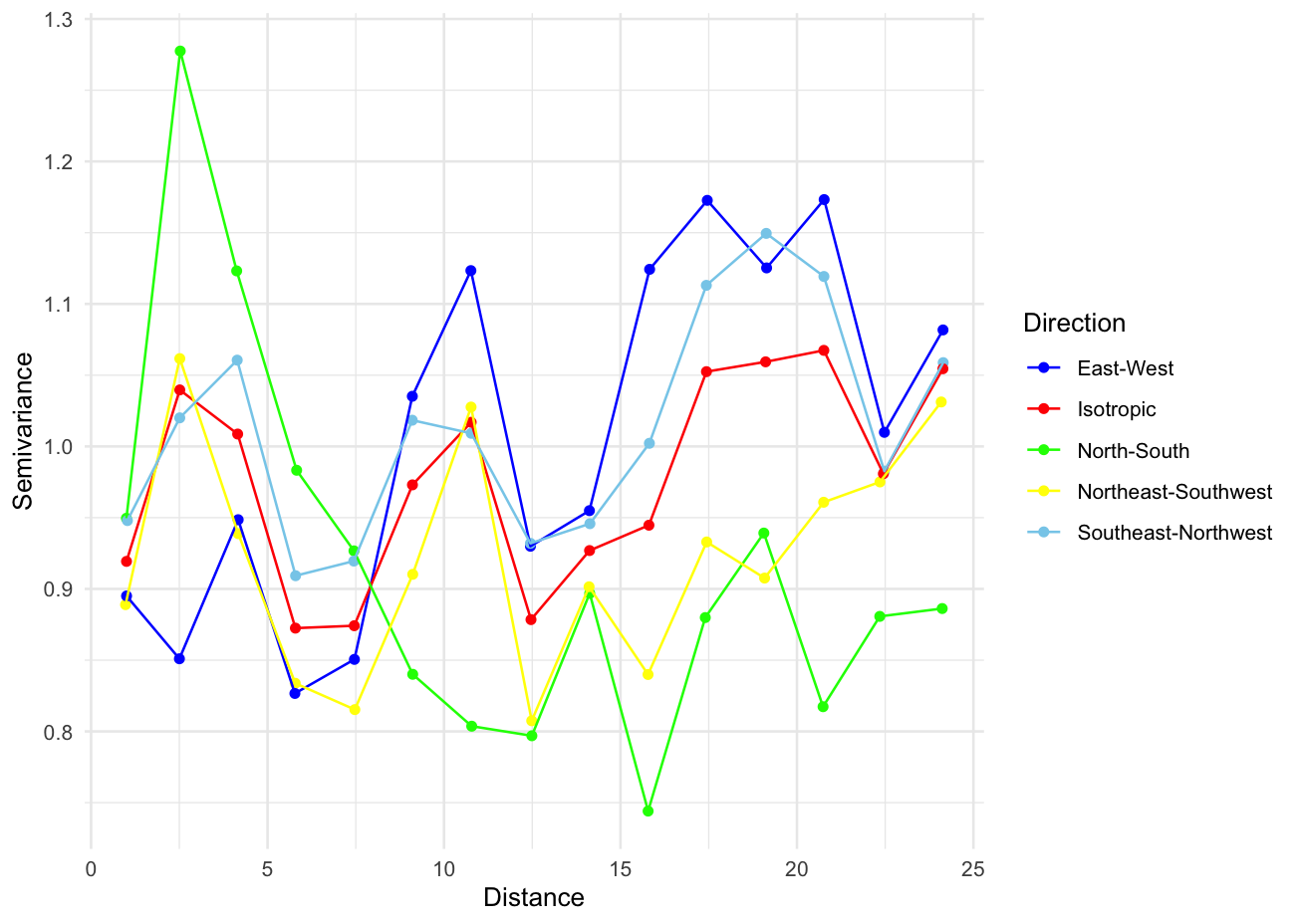}
    \caption{Empirical directional and isotropic semivariograms for the logit proportions.}
    \label{fig:directional_variogram}
\end{figure}

\section{Methodology}
\label{sec:Methods}

In this section, we propose a novel nonparametric method for estimating the covariance structure of a stationary and isotropic spatial process. The key idea is to leverage the fact that any stationary isotropic covariance function that is valid in all dimensions admits a Gaussian scale-mixture representation. We define the nonparametric maximum likelihood estimator over the space of non-negative mixing measures. We propose to estimate the mixing distribution using a finite discrete approximation. This leads to a flexible covariance estimator that avoids restrictive parametric assumptions while preserving the validity of the covariance function.

\subsection{Gaussian Scale-Mixture Model}
We consider a random field $\{Y(\mathbf{s}); \,\mathbf{s}\in \mathcal{D}\}$, where $\mathcal{D} \subseteq \mathbb{R}^d$ denotes the domain of interest. For each $\mathbf{s}\in \mathcal{D}$, we assume that
\begin{equation}
    Y(\mathbf{s})  = \textbf{X}(\mathbf{s})^\top \boldsymbol{\beta} + Z(\mathbf{s})  + \epsilon(\mathbf{s}),
\end{equation}
where $\textbf{X}(\mathbf{s})\in \R^p$ is a known vector of covariates, $\boldsymbol{\beta}\in\R^p$ is an unknown parameter vector, $Z(\mathbf{s})$ is a zero-mean weakly stationary process with an isotropic covariance function, and $\epsilon(\mathbf{s})$ is a zero-mean measurement error process independent of $Z(\mathbf{s})$. We assume that $\text{Var}(\epsilon(\mathbf{s})) = \sigma_e^2$ and $\text{Cov}(\epsilon(\mathbf{s}_1), \epsilon(\mathbf{s}_2)) = 0$, $\forall \mathbf{s}_1 \ne \mathbf{s}_2$. Our goal is to estimate the covariance structure $C_Y:\mathcal{D}\times\mathcal{D}\to\mathbb{R}_+$ of $Y(\mathbf{s})$, given by

\begin{equation}\label{eq: C_Y}
C_Y(\mathbf{s}_1, \mathbf{s}_2):=\text{Cov}(Y(\mathbf{s}_1),Y(\mathbf{s}_2))=C_Z(\mathbf{s}_1, \mathbf{s}_2)+\sigma_e^21[\mathbf{s}_1=\mathbf{s}_2],
\end{equation}
where $C_Z(\mathbf{s}_1,\mathbf{s}_2):=\text{Cov}(Z(\mathbf{s}_1), Z(\mathbf{s}_2))$.

By the isotropy assumption, the covariance of $Z(\mathbf{s})$ depends only on the Euclidean distance $h = \|\s_1-\s_2\|_2$.  Thus, there exists a function $c_Z:\mathbb{R}_+\to \mathbb{R}_+$ such that 
\begin{equation}
C_Z(\mathbf{s}_1,\mathbf{s}_2)=c_Z(\|\mathbf{s}_1-\mathbf{s}_2\|_2).
\end{equation}

Let $\mathcal{D}_d$ denote the class of stationary isotropic covariance functions $c:\R_{+}\to \R$ such that $C(\mathbf{s}_1,\mathbf{s}_2) = c(\|\mathbf{s}_1-\mathbf{s}_2\|_2)$ is a valid stationary isotropic covariance function on $\R^d$. Also let $\mathcal{D}_\infty = \cap_{d\ge 1} \mathcal{D}_d$ denote the class of radial covariance functions that are valid in every dimension. If $c_Z \in \mathcal{D}_d$, it admits an integral representation involving Bessel kernels \citep{Stein1999}. As $d \to \infty$, an appropriately scaled Bessel kernel converges to a Gaussian kernel. Consequently, if $c_Z\in \mathcal{D}_\infty$, that is, if it is valid in all dimensions, then $c_Z$ admits the following representation
\begin{equation}
    c_Z(h)=\int_{0}^\infty e^{-(h/\alpha)^2}\,d\mu_Z(\alpha)
    \label{eq:representation_isotropy_2}
\end{equation}
where $\mu_Z$ is a positive measure supported on $[0,\infty)$. Define the Gaussian Kernel at $\alpha$ by
\begin{equation*}
    K_G(0,h) := 1\{h=0\},\quad K_G(\alpha,h) := e^{-(h/\alpha)^2},\,\,\alpha>0.
\end{equation*}
Then the covariance function of $Y$ can be written as
\begin{equation}\label{eq: C_Y_mixture}
   C_Y(\mathbf{s}_1, \mathbf{s}_2) = \int_0^\infty K_G(\alpha, \|\mathbf{s}_1- \mathbf{s}_2\|_2)\,d\mu_Y(\alpha),
\end{equation}
where $\mu_Y := \mu_Z+\sigma_e^2 \delta_{0}$, with $\delta_0$ denoting the Dirac measure at $0$. Note that this is an infinite mixture of Gaussian kernels. 

Our method aims to directly estimate the mixing measure $\mu_Y$, and thus the covariance function $C_Y$, from the observed spatial data. This formulation is attractive because it includes many commonly used parametric covariance models, including the exponential and Mat\'{e}rn families, while avoiding the need to impose a specific parametric form~\citep{Genton2001, Tronarp2018}.

\subsection{Ideal NPML Estimator and Finite-grid Approximation}
Let $\mathcal{D}_n=\{\boldsymbol{s}_1,\dots,\boldsymbol{s}_n\}$ be the set of locations observed in the data and $\boldsymbol{Y}=[Y(\boldsymbol{s}_1),...,Y(\boldsymbol{s}_n)]^\top\in\mathbb{R}^n$ be the observed response vector. 
We further assume that $Z(\boldsymbol{s})$ and $\epsilon(\boldsymbol{s})$ are Gaussian processes.

Although the Gaussian scale-mixture representation in \eqref{eq: C_Y_mixture} is defined over $[0,\infty)$, for estimation we restrict the candidate support of the mixing measure to a compact set $A=[0,U]$. This corresponds to bounding the range of spatial scales of the estimator. In practice, $U$ is chosen large enough to include spatial components whose effective range is comparable to the maximum observed pairwise distance. 
Let $\mathcal{M}_A$ be the class of finite, positive measures supported on $A$. For $\mu \in \mathcal{M}_A$, define the covariance matrix $\Sigma(\mu) \in \R^{n\times n}$ by
\begin{equation}
    \Sigma(\mu):=\int_A K_G(\alpha,\mathbf{D}) d\mu(\alpha),
\end{equation} where $\D\in\R^{n\times n}$ denotes the pairwise distance matrix with entries $\D_{ij} = \|\boldsymbol{s}_i-\boldsymbol{s}_j\|_2$, and $K_G(\alpha,\D)$ is the matrix with $i,j$ entry $K(\alpha,\D_{ij})$. Equivalently,  $[\Sigma(\mu)]_{ij} = \int_AK_G(\alpha,\D_{ij})d\mu(\alpha)$. 

For a Gaussian vector $\mathbf{Y} \sim N_n(\mathbf{X}\boldsymbol{\beta},\Sigma)$ with a positive definite covariance matrix $\Sigma$, the log-likelihood is
\[\ell(\bbeta,\Sigma)=- \frac{n}{2}\log(2\pi)-\frac{1}{2}\log|\Sigma|-\frac{1}{2}(\boldsymbol{Y}-\mathbf{X}\boldsymbol{\beta})^\top \Sigma^{-1}(\boldsymbol{Y}-\mathbf{X}\boldsymbol{\beta}).\] 
Since for any positive definite $\Sigma$, the maximizer of the likelihood over $\bbeta$ is given as $\hat{\bbeta}_{\Sigma}:  = (\X^\top \Sigma^{-1}\X)^{-1}\X^\top \Sigma^{-1}\Y$, we can define the following profile log-likelihood, viewed as a function of $\Sigma$, by
\begin{equation}
    \ell_P(\Sigma)=-\frac{n}{2}\log(2\pi)-\frac{1}{2}\log |\Sigma|-\frac{1}{2}(\Y-\X\hat{\bbeta}_{\Sigma})^\top \Sigma^{-1}(\Y-\X\hat{\bbeta}_{\Sigma}).
\end{equation}

Finally, define the compact-support mixture covariance class $\mathcal{H}_A:=\{\Sigma(\mu):\,\, \mu \in \mathcal{M}_A\}$. We define the (ideal) NPML covariance estimator over $A$ by 
\begin{equation}\hat{\Sigma}_{ML} \in \argmax_{\Sigma:\; {\Sigma\in \mathcal{H}_A}} \ell_P(\Sigma).\label{eq:npmle}
\end{equation}

This estimator is well-defined in the sense that there exists at least one solution for the problem \eqref{eq:npmle} (see Proposition~\ref{prop:existence} in Section~\ref{sub:properties_NPML}).  Also, this ideal estimator is infinite-dimensional because the optimization is over all finite positive measures supported on $A$. However, even when the optimization problem is defined such that the objective is maximized over all positive measures, the resulting solution is often finite and discrete \citep{Lindsay1983, lindsay1993uniqueness}. In our setting, under mild conditions, Proposition~\ref{prop:finite_support} shows an analogous property that an NPML covariance matrix admits at least one finitely supported representing measure.

Motivated by this property, we propose a finite-grid approximation which leads to a tractable estimation procedure. 
While finite mixtures of Gaussian kernels do not, in general, reproduce covariance models such as exponential or Mat\'{e}rn family, we expect that they provide good approximations given a sufficiently fine grid. We evaluate the quality of this finite grid approximation empirically in Section~\ref{sub_sec:effectiveness_finite_mixture_approximations}. The results show that the approximation error is small for moderately sized grids. In the settings considered, the grid size of $s\ge 10$ was sufficient for good approximations of Mat\'{e}rn covariance functions.

We now define the finite-grid approximation of the NPML estimator $\hat{\Sigma}_{ML}$. 
To improve robustness to the original scale of the map, we perform estimation on a normalized distance scale. Let $h_{lm}=\|\mathbf{s}_l-\mathbf{s}_m\|_2$ be the distance between the random field observed at two arbitrary locations $\mathbf{s}_l$ and $\mathbf{s}_m$. Thus, for any $s_h>0$, it follows that
\begin{equation}
    C_Y(\mathbf{s}_l, \mathbf{s}_m)= \int_{0}^\infty e^{-((h_{lm}/s_h)/(\alpha/s_h))^2}\,d\mu_Y(\alpha) =\int_{0}^\infty e^{-(\tilde{h}_{lm}/\tilde{\alpha})^2}\,d\tilde{\mu}_Y(\tilde{\alpha})
\end{equation}
where we define $\tilde{h}_{lm}=h_{lm}/s_h$, $\tilde{\alpha}=\alpha/s_h$, and $\tilde{\mu}_Y$ denotes the induced measure by the change of variable.

We then approximate the re-parametrized mixing measure $\tilde{\mu}_Y$ by a discrete measure supported on a finite grid. 
Let $\{\tilde\alpha_1,\dots,\tilde\alpha_s\}\subset [0,\infty)$ be a grid with $\tilde \alpha_1 = 0$, and write 
\begin{equation}\label{def: finite_mixing_measure}
    \tilde{\mu}_Y(\cdot;\,\boldsymbol{\theta}): = \sum_{i=1}^s\theta_i\delta_{\{\tilde \alpha_i\}}
\end{equation} where $\delta_{\{\tilde \alpha_i\}}$ denotes the Dirac measure at $\tilde \alpha_i$, and $\theta_i$ is the corresponding weight. This leads to the finite-mixture approximation
\begin{equation}
    C_Y(\mathbf{s}_l, \mathbf{s}_m;\boldsymbol{\theta}) = \int_{0}^\infty e^{-(\tilde{h}_{lm}/\tilde{\alpha})^2}\,d\tilde{\mu}_Y(\tilde{\alpha};\,\boldsymbol{\theta})=1[\tilde{h}_{lm}=0]\theta_1+\sum_{i=2}^s \exp(-\tilde{h}_{lm}^2/\tilde{\alpha}_i^2)\theta_i
    \label{eq:finite_mixture}
\end{equation}
where the weights $\theta_i$ are parameters to be estimated from the data. Note that $\tilde\alpha_i$ are the scale parameters of the Gaussian kernels, where a lower value of $\tilde{\alpha}$ corresponds to a faster decay of $\exp(-\tilde{h}^2 / \tilde{\alpha}^2)$.

% how to set A
We let $s_h$ be the maximum pairwise distance between the observed set of locations so that $0\le \tilde{h}_{lm} \le 1$ for all $l,m$. We select the grid of $\tilde{\alpha}_i$ values, $\{0,\tilde{\alpha}_{\min}, \dots, \tilde{\alpha}_{\max}\}$, such that the \textit{effective range} of the fastest and slowest decaying components, that is, $\exp(-\tilde{h}^2 / \tilde{\alpha}_{\min}^2)$ and $\exp(-\tilde{h}^2 / \tilde{\alpha}_{\max}^2)$, corresponds to  $1/n$ and $0.9$, respectively. The effective range of a correlation function is defined as the distance at which the correlation becomes small, typically below 0.05. Therefore, the effective range of $\exp(-\tilde{h}^2 / \tilde{\alpha}^2)$ is $\tilde{h}$ such that $\exp(-\tilde{h}^2 / \tilde{\alpha}^2) = 0.05$, that is, $\tilde{h} = \tilde{\alpha} \log(20)^{1/2}$. 
Therefore, we set $\tilde{\alpha}_2 = \tilde{\alpha}_{\min}  = (1/n)\log(20)^{-1/2}$ and $\tilde{\alpha}_s = \tilde{\alpha}_{\max}  = (0.9)\log(20)^{-1/2}$, with $\{\tilde{\alpha}_3,\dots,\tilde{\alpha}_{s-1}\}$ chosen by linearly interpolating between $\tilde{\alpha}_2$ and $\tilde{\alpha}_s$ on the logarithmic scale.

 Under the finite-mixture representation in \eqref{eq:finite_mixture}, the covariance structure is parameterized by the non-negative weight vector $\boldsymbol{\theta}\in\R_{+}^s$. Specifically, the covariance matrix can be written as
\begin{align}\Sigma(\boldsymbol{\theta}):=\sum_{i=1}^{s}\theta_iK_G(\tilde{\alpha}_i, \tilde{\mathbf{D}})=\sum_{i=1}^{s}\theta_i\Sigma_G(\tilde{\alpha}_i),\label{eq:matrix_finite_mixture}
\end{align}
where $\tilde{\mathbf{D}}$ is the pairwise distance matrix on the normalized distance scale with entries $\tilde{h}_{jk}$, and $\Sigma_G(\tilde{\alpha}_i)\in\mathbb{R}^{n\times n}$. The first component $\tilde{\alpha}_1=0$ represents the nugget effect so that $\Sigma_G(\tilde{\alpha}_1) = \mathbf{I}_n$. For $i=2,\dots,s$,
\begin{align}
    [\Sigma_G(\tilde{\alpha}_i)]_{jk} = \exp (-\tilde{h}_{jk}^2/\tilde{\alpha}_i^2).
\end{align}
Then the optimization problem \eqref{eq:npmle} is reduced to the finite-dimensional constrained optimization problem
\begin{equation}\label{eq:constrained_opt_problem}
\hat{\boldsymbol{\theta}}_{ML}=\argmax_{\boldsymbol{\theta}; {\boldsymbol{\theta}\ge \boldsymbol{0}}} \ell_P(\Sigma(\boldsymbol{\theta})).
\end{equation}
The corresponding finite-grid Gaussian mixture estimator, which we denote as FGM, is
\begin{equation}
    \hat{\Sigma}_{\rm FGM} = \Sigma(\hat{\boldsymbol{\theta}}_{ML}).
\end{equation}

% briefly disucss WLS
The same finite-grid representation also naturally leads to a weighted-least squares version of the estimator. Least-squares fitting of empirical covariances or semivariograms is commonly used in spatial covariance estimation, and provides a simple moment-based alternative to likelihood-based estimation. 
Recall $\tilde{h}_{ij}=\|\boldsymbol{s}_i - \boldsymbol{s}_j\|_2/s_h$ denotes the normalized pairwise distance. Choose bin endpoints $0=b_0<b_1<\dots<b_L\le 1$, and let $\bar{b}_k:=(b_{k-1}+b_k)/2$ denote the mid-point of the $k$th bin $(b_{k-1},b_k]$. Define $N_k$ as the set of pairs of observation locations whose separation distance falls in the $k$th lag interval, i.e.,
\begin{align*}
   N_k=\{(i,j):\,\,\tilde{h}_{ij}\in(b_{k-1},b_{k}]\},\quad k=1,\dots,L.
\end{align*}
Let $\hat{C}_E$ be the empirical estimator of $C$ defined by
\begin{align*}
    \widehat{C}_E(\tilde h)=
    \begin{cases}
    \frac{1}{|\mathcal{D}_n|} \sum_{\mathbf{s} \in \mathcal{D}_n}\{Y(\mathbf{s})-\bar{Y}\}^2& \tilde{h}=\bar{b}_0= 0\\
      \frac{1}{|N_k|} \sum_{(i,j) \in N_k}\{Y(\mathbf{s}_i)-\bar{Y}\}\{Y(\mathbf{s}_j)-\bar{Y}\} & \tilde{h} \in \{\bar{b}_1,\dots,\bar{b}_{L}\}\\
    \end{cases}
\end{align*}
where $|N_k|$ is the cardinality of the set $N_k$ and $\bar{Y}$ is the sample mean.

Given non-negative weights $w_k$, we can define the finite-grid WLS estimator as 
\begin{equation}\hat{\boldsymbol{\theta}}_{WLS}=\argmin_{\boldsymbol{\theta};\,\, {\boldsymbol{\theta}\ge \boldsymbol{0}}} \sum_{k=0}^{L}w_k \{ \hat{C}_E(\bar{b}_k) - \sum_{i=1}^s K_G(\tilde\alpha_i, \bar{b}_{k})\theta_i\}.\label{eq:wls_minimizer}
\end{equation}

\subsection{Computation}
\label{sub_sec:estimation}

We now discuss how to solve the constrained optimization problem in \eqref{eq:constrained_opt_problem}. First, the following Lemma shows that when the weight vector $\boldsymbol{\theta}$ has at least one non-zero element, the corresponding covariance matrix $\Sigma(\boldsymbol{\theta})$ is positive definite, and therefore the profile log-likelihood $\ell_P$ is well-defined.
\begin{lemma}
    For $\btheta \in \mathbb{R}_{\ge 0}^s$ with at least one $\theta_j >0$, the matrix $\Sigma(\btheta)$ defined in \eqref{eq:matrix_finite_mixture} is positive definite. Consequently, the inverse matrix $\Sigma(\btheta)^{-1}$ is well-defined.
\end{lemma}
\begin{proof}
Note that when $\tilde{\alpha}=0$, $\Sigma_G(\tilde{\alpha}) = \mathbf{I}_n$. When $\tilde\alpha >0$, each element of $\Sigma_G(\tilde{\alpha})_{jk}$ is formed by a Gaussian kernel $K_{\alpha}(\mathbf{s}_i,\mathbf{s}_j)= \exp(-\|\mathbf{s}_i -\mathbf{s}_j\|_2^2/\alpha^2)$. It is a well-known result that for any $\alpha>0$, $K_\alpha$ is a strictly positive definite kernel (e.g.,~\citealp{hofmann2008kernel}). Consequently, $\Sigma_G(\tilde{\alpha}_i)$ is a positive definite matrix for $i=1,\dots,s$. Since each $\theta_i \ge 0$ and at least one is strictly positive by assumption, it follows that $\Sigma(\btheta) = \sum_{i=1}^s \theta_i\Sigma_G(\tilde{\alpha}_i)$ is positive definite.
\end{proof}

To solve the minimization problem instead, define the negative profile log-likelihood $\mathcal{L}_P(\btheta):= -\ell_P\{\Sigma(\btheta)\}$. Furthermore, we let $\Sigma_i := \Sigma_G(\tilde{\alpha}_i)$ for simplicity of exposition. Lemma~\ref{lem: L_P_deriv} presents the gradient and expected Hessian of $\mathcal{L}_P(\btheta)$. Define 
\begin{align*}
    P_\btheta:=\Sigma(\btheta)^{-1} - H_\btheta, \quad H_\btheta:=\Sigma(\btheta)^{-1}\X(\X^\top \Sigma(\btheta)^{-1}\X)^{-1} \X^\top\Sigma(\btheta)^{-1}.
\end{align*}
\begin{lemma}\label{lem: L_P_deriv}
    The gradient $g(\btheta) := \frac{d}{d \theta}\mathcal{L}_P(\btheta)$ has elements
    \begin{align*}
        g(\btheta)_i =  \frac{1}{2}\{ \operatorname{tr}(\Sigma(\btheta)^{-1}\Sigma_i) - \Y^\top P_\btheta \Sigma_i P_\btheta \Y\}.
    \end{align*}
    The expected Hessian $H(\btheta):= E[ \frac{d^2\mathcal{L}_P(\btheta)}{d\theta d \theta^\top} ] $ has elements
    \begin{align*}
    &H(\btheta)_{ij}  =\operatorname{tr}(  P_\btheta \Sigma_i P_\btheta \Sigma_j) - \frac{1}{2} \operatorname{tr}(\Sigma(\btheta)^{-1}\Sigma_j \Sigma(\btheta)^{-1}\Sigma_i)=2\mathcal{I}_R(\btheta)_{ij} - \mathcal{I}_0(\btheta)_{ij},
\end{align*}
where $\mathcal{I}_R(\btheta)$ and $\mathcal{I}_0(\btheta)$ are defined as
\begin{align*}
    \mathcal{I}_R(\btheta)_{ij}:= \frac{1}{2}\operatorname{tr}(  P_\btheta \Sigma_i P_\btheta \Sigma_j),\quad \mathcal{I}_0(\btheta)_{ij} := \frac{1}{2} \operatorname{tr}(\Sigma(\btheta)^{-1}\Sigma_j \Sigma(\btheta)^{-1}\Sigma_i).
\end{align*}
\end{lemma}
The proof of Lemma~\ref{lem: L_P_deriv} is deferred to Section~\ref{supp:lem:L_P_deriv}.

% Line search algorithm
We propose to minimize $\mathcal{L}_P(\btheta)$ subject to $\btheta \ge 0$ using a second-order descent algorithm. At each iteration $t$, we find the search direction by solving the quadratic subproblem
\begin{align*}
    \tilde{\btheta}_{t+1} = \arg \min_{\btheta; \btheta \ge 0} g(\btheta_t)^\top (\btheta - \btheta_t)+ \frac{1}{2}(\btheta - \btheta_t)^\top M(\btheta_t)(\btheta - \btheta_t)
\end{align*}
where $g(\btheta)$ is the gradient function of $\mathcal{L}_P(\btheta)$ and $M(\btheta)$ is any positive definite approximation to the Hessian. Note that the resulting direction $\mathbf{p}_t = \tilde{\btheta}_{t+1} - \btheta_t$ satisfies $g(\btheta_t)^\top \mathbf{p}_t <0$, and hence is a descent direction regardless of the specific choice of $M$. While $\mathcal{I}_R(\btheta)$ and $\mathcal{I}_0(\btheta)$ are positive semi-definite (PSD) matrices, the expected Hessian $H(\btheta)$, as a difference of two PSD matrices, is not necessarily PSD and therefore cannot be directly used as $M$. We use 
\begin{align*}
    M(\btheta_t) = \mathcal{I}_R(\btheta_t) + \epsilon_t \mathbf{I}_s
\end{align*}
where $\epsilon_t = \epsilon 1[\mathcal{I_R}(\btheta_t)\mbox{ is not PD}]$ for a small $\epsilon>0$. The matrix $\mathcal{I}_R(\btheta)$ is a natural approximation of $H(\btheta)$ since $\mathcal{I}_R(\btheta)-H(\btheta)= \mathcal{I}_0(\btheta)-\mathcal{I}_R(\btheta)$ tends to be small when $p \ll n$. Also, while it is not guaranteed, $\mathcal{I}_R(\btheta_t)$ is typically positive definite in practice when $p\ll n$ and the correction $\epsilon_t \mathbf{I}_s$ is rarely needed.

We use line search methods to determine how far to move along the direction $\mathbf{p}_{t}$. This guarantees that at each iteration, we decrease the function value. We set $\boldsymbol{\theta}_{t+1}=\boldsymbol{\theta}_{t}+\lambda \boldsymbol{p}_t$, where $0<\lambda\leq 1$ is a step size chosen by using a backtracking line search to ensure the Armijo sufficient decrease condition \begin{align*}
    \mathcal{L}_p(\boldsymbol{\theta}_{t}+\lambda \boldsymbol{p}_t)-\mathcal{L}_p(\boldsymbol{\theta}_t)\leq c\lambda g(\boldsymbol{\theta}_t)^\top \boldsymbol{p}_t
\end{align*} is satisfied \citep[p.~37]{NocedalWright1999}. We use the initial step size $\bar{\lambda}=1$ and $\rho=c=(1/2)$ for the line search parameters.

\subsection{Theoretical Properties of the NPML estimator}
\label{sub:properties_NPML}
In this subsection, we establish some theoretical properties of the NPML estimator. Specifically, we show that the NPML covariance estimator exists, and that at least one NPML mixing measure can be chosen to have finite support. We establish results under a general mixture kernel framework, where an isotropic covariance function $C(h)$ admits an integral representation
\begin{equation*}
    C(h)=\int_{A}K(\alpha,h)\,\mu(d\alpha)
\end{equation*}
for a base covariance kernel $K:A\times [0,\infty)\to[0,\infty)$, and $A$ is a set of allowable values of $\alpha$. The results for the Gaussian mixture estimator are obtained as the special case where $K=K_G$.

\begin{assumption}[Assumptions on the base covariance kernel] \label{cond:kernel}
We assume the the base covariance kernel $K$ satisfies the following conditions:
\begin{enumerate}
    \item (Normalization): $K(\alpha,0)=1$ for each $\alpha\in A$.
    \item (Positive definiteness): for each finite $\alpha$ and choice of distinct locations $\mathbf{s}_1,...,\mathbf{s}_n\in\mathbb{R}^d$, the matrix $K(\alpha,\D)$ is positive definite, where $\D$ denotes the pairwise distance matrix of the $n$ locations, and $K(\alpha,\D)$ is the matrix with $i,j$ entry $K(\alpha,\D_{ij})$.
    \item (Continuity): for each $h\in [0,\infty)$, the function $\alpha\mapsto K(\alpha,h)$ is continuous on $A$.
\end{enumerate} 
    
\end{assumption}

The normalization and positive definiteness conditions imply $|K(\alpha,h)|< 1$, for all $h>0 $, for each finite $\alpha$. To see this, consider the matrix $\D\in\mathbb{R}^{2\times 2}$ resulting from two points, $\mathbf{s}_1$ and $\mathbf{s}_2$, such that $\|\mathbf{s}_1-\mathbf{s}_2\|_2=h>0$. Then from positive definiteness, the determinant $|K(\alpha,\D)|=1-K(\alpha,h)^2> 0$. 

Now, for a compact set $A\subset[0,\infty)$ and a given distance matrix $\D\in\mathbb{R}^{n\times n}$, we define the set of mixture covariances 
$$\mathcal{H}_A(K)=\left\{\int_{A}K(\alpha,\D)\,\mu(d\alpha)\in\R^{n\times n}:\,\mu\in\mathcal{M}(A)\right\}.$$ The following lemma states that $\mathcal{H}_A(K)$ is closed. This property is important for showing that the NPMLE of $\Sigma$ over $\mathcal{H}_A(K)$ exists.

\begin{lemma}
Assume $A\subset [0,\infty)$ is compact and the base covariance kernel $K$ satisfies Assumption~\ref{cond:kernel}. Then $\mathcal{H}_A(K)$ is closed. \label{lem:closed}
\end{lemma}

\begin{proof}

Let $\Sigma^{(k)}$, $k=1,2,...$, be a sequence in $\mathcal{H}_A(K)$ such that $\Sigma^{(k)}\to S$ in the Frobenius norm for some $S\in\mathbb{R}^{n\times n}$, that is, $\|\Sigma^{(k)}-S\|_F\to 0$. We show $S\in \mathcal{H}_A(K)$. 

For each $k$, let $\mu_k$ denote a representing measure corresponding to $\Sigma^{(k)}$ such that $\Sigma^{(k)} = \int_A K(\alpha,\mathbf{D})d\mu_k(\alpha)$. Since $\Sigma^{(k)}$ converges, there exists $B<\infty$ such that $\|\Sigma^{(k)}\|_F\le B$ for all $k$. Now, $B\geq \|\Sigma^{(k)}\|_F\geq \sqrt{\sum_{i=1}^n (\Sigma^{(k)}_{ii})^{2}}$, for each $k$. By the condition $K(\alpha,0)=1$, for each diagonal entry, 
$\Sigma^{(k)}_{ii}=\int_{A}K(\alpha,0)\,\mu_k(d\alpha)=\int_{A}1\,\mu_k(d\alpha)=\mu_k(A)$. Substituting, we obtain $\mu_{k}(A)\leq B/\sqrt{n}$ for each $k$. Since $k$ was arbitrary, this shows that $\{\mu_k(A)\}$ is uniformly bounded. Then, by Bolzano-Weierstrass Theorem, there exists a subsequence $\mu_{k_j}$ such that $\mu_{k_j}(A) \to m$ for some $m\ge 0$. 

% when m=0
First suppose $m=0$. For each pair $(a,b)$, continuity of $\alpha \to K(\alpha, D_{ab})$ on compact $A$ implies that this function is bounded. Therefore, 
\begin{equation*}
|\int_A K(\alpha, D_{ab})d\mu_{k_j}(\alpha) | \le \sup_{\alpha \in A} |K(\alpha,D_{ab})| \mu_{k_j}(A) \to_j 0.
\end{equation*}
Thus, $\Sigma^{(k_j)}\to \mathbf{0}_{n\times n}$ elementwise, and thus in Frobenius norm. Since $\Sigma^{(k)}\to_k S$ by assumption, $S = \mathbf{0}_{n\times n}$. The zero matrix $\mathbf{0}_{n\times n}$ belongs to $\mathcal{H}_A(K)$ since it is represented by the null measure.

% when m>0
If $m>0$, for each $k$, define a probability measure $\nu_k := \mu_k/\mu_k(A)$. Since $A$ is compact, the sequence $\{\nu_{k_j}\}$ is tight. By Prohorov's theorem (ref. Theorem 5.1 in~\cite{billingsley2013convergence}), there exists a further subsequence $\{\nu_{k_{j_l}}\}$ and a probability measure $\nu$ such that $\nu_{k_{j_l}} \Rightarrow_l \nu$. 
For each pair $(a,b)$, since $\alpha \to K(\alpha,D_{ab})$ is continuous and bounded on $A$, we have
\begin{equation*}
    \Sigma^{(k_{j_l})}_{ab} = \int_A K(\alpha, D_{ab}) \mu_{k_{j_l}}(A) d\nu_{k_{j_l}} (\alpha) \to_l \int_A K(\alpha, D_{ab}) m d\nu (\alpha) = \int_A K(\alpha, D_{ab})  d\mu (\alpha)
\end{equation*}
for $\mu:=m\nu$. Therefore, $S = \int_A K(\alpha,\mathbf{D})d\mu(\alpha)$. This shows that $\mathcal{H}_A(K)$ is closed.

\end{proof}

The next proposition shows an existence of NPML covariance estimator $\hat{\Sigma}_{ML}$.
\begin{proposition}\label{prop:existence}
Assume $A = [0,U]$ with $0<U<\infty$, and suppose the base covariance kernel $K$ satisfies Assumption~\ref{cond:kernel}. Let $\mathbf{X} \in \R^{n\times p}$ have full column rank with $n>p$, and assume $\mathbf{Y}\notin \mathcal{C}(\X)$. Then the profile likelihood optimization problem
\begin{equation}\hat{\Sigma}_{ML}^{(K)} \in \argmax_{\Sigma:\; {\Sigma\in \mathcal{H}_A(K)}} \ell_P(\Sigma).\label{eq:npmle_K}
\end{equation}
attains its optimum over $\mathcal{H}_{[0,U]}(K)$. Moreover, there exist constants $0<L_0<M_0<\infty$ such that every representing measure $\hat{\mu}_{ML}^{(K)}$ of $\hat{\Sigma}_{ML}^{(K)}$ satisfies $L_0\le \hat{\mu}_{ML}^{(K)}([0,U]) \le M_0$.
\end{proposition}

\begin{proof}
We first show that the likelihood maximization can be restricted to representing measures whose total mass is bounded away from both zero and infinity. Let $|\mu|=\mu([0,U])$ be the total mass of $\mu$. We show that there exist constants $0<L_0<M_0<\infty$ such that if $|\mu|<L_0$ or $|\mu|>M_0$, then $\ell_P(\Sigma(\mu))<\ell_P(\Sigma(\mu_0))$ for some fixed non-trivial measure $\mu_0 \in \mathcal{M}_A$; hence, no maximizer can be represented by a measure with total mass smaller than $L_0$ or larger than $M_0$.

Define \begin{align*}
    P_{\mu}:=\Sigma(\mu)^{-1} - H_\mu, \quad H_\mu:=\Sigma(\mu)^{-1}\X(\X^\top \Sigma(\mu)^{-1}\X)^{-1} \X^\top\Sigma(\mu)^{-1}.
\end{align*} 
and 
\begin{align*}
\tilde \ell_P(\mu):=\ell_P(\Sigma(\mu))=-\frac{n}{2}\log(2\pi)-\frac{1}{2}\log |\Sigma(\mu)|-\frac{1}{2}\Y^\top P_{\mu}\Y.\end{align*}
Let $\mu_0$ be an arbitrary nontrivial measure on $A=[0,U]$. 

Observe that $\lambda_\star:=\inf_{\alpha\in [0,U]} \lambda_{\min}(K(\alpha,\D))>0$. This follows from continuity of $K(\alpha,\D)$ in $\alpha$ on the compact set $A$, so that $\inf_{\alpha \in [0,U]} \lambda_{\min}(K(\alpha,\D))=\lambda_{\min}(K(\alpha^\star,\D))$ for some $\alpha^\star\in[0,U]$, combined with the positive definiteness of $K(\alpha,\D)$ for each finite $\alpha$. Similarly, $\lambda^\star:=\sup_{\alpha\in [0,U]} \lambda_{\max}(K(\alpha,\D))\leq n<\infty$, where we use $K(\alpha,h)\leq 1$, for all finite $\alpha$, $h$, combined with the bound $|\lambda_{\max}(K(\alpha,\D))|\leq \sum_{i=1}^n |K(\alpha,\D_{ij})|\leq n$.

 Then \begin{align}\label{prop:existence:lmax}
\lambda_{\max}(\Sigma(\mu))=\lambda_{\max}\left(\int K(\alpha,\D)\mu(d\alpha)\right)= \sup_{v:\|v\|_2=1}\int v^\top K(\alpha,\D) v d\mu(\alpha)\leq \lambda^\star|\mu|.
\end{align} Similarly, $\lambda_{\min}(\Sigma(\mu))\geq \lambda_\star|\mu|$.

We observe \begin{align*}
\Y^\top P_{\mu}\Y&=\underset{{\bbeta}\in\mathbb{R}^{p}}{{\min}}\;(\Y-\X\bbeta)^\top \Sigma(\mu)^{-1}(\Y-\X\bbeta).
\end{align*} Fix a vector $v\in \mathbb{R}^{n}$ such that $v^\top \X=0$, $v^\top \Y\neq 0$, $\|v\|=1$. By assumption on $\Y$, $\Y\notin \mathcal{C}(\X)$, so such a vector $v$ exists.  Then for any $\bbeta$, \begin{align*}
(v^\top \Y)^2=\{v^\top(\Y-\X\bbeta)\}^2\leq (v^\top\Sigma(\mu)v)\{(\Y-\X\bbeta)^\top \Sigma(\mu)^{-1}(\Y-\X\bbeta)\}
\end{align*} and so \begin{align*}
\frac{(v^\top \Y)^2}{v^\top\Sigma(\mu)v}\leq  \;(\Y-\X\bbeta)^\top \Sigma(\mu)^{-1}(\Y-\X\bbeta)
\end{align*} for any $\bbeta$. 

In particular, $(v^\top \Y)^2/v^\top\Sigma(\mu)v \le \Y^\top P_\mu \Y$, and it follows that
\begin{align*}
-\tilde \ell_{P}(\mu)&= \frac{n}{2}\log(2\pi)+\frac{1}{2}\log |\Sigma(\mu)|+\frac{1}{2}\Y^\top P_{\mu}\Y\\
&\geq \frac{n}{2}\log(2\pi)+\frac{1}{2}\log |\Sigma(\mu)|+\frac{1}{2}\frac{(v^\top \Y)^2}{v^\top\Sigma(\mu)v}\\
%&\geq \frac{n}{2}\log(2\pi)+\frac{n}{2}\log \lambda_{\min}(\Sigma(\mu))+\frac{1}{2}\frac{(v^\top \Y)^2}{v^\top\Sigma(\mu)v}\\
&\geq \frac{n}{2}\log(2\pi)+\frac{n}{2}\log \lambda_{\min}(\Sigma(\mu))+\frac{1}{2}\frac{(v^\top \Y)^2}{\lambda_{\max}(\Sigma(\mu))}\\
&\geq \frac{n}{2}\log(2\pi)+\frac{n}{2}\log\{|\mu|\lambda_{\star}\}+\frac{1}{2}\frac{(v^\top \Y)^2}{|\mu|\lambda^{\star}}.
\end{align*} 
Now, consider a function $f(t)=\frac{n}{2}\log\{t\lambda_{\star}\}+\frac{1}{2}\frac{(v^\top \Y)^2}{t\lambda^{\star}}$ for $t>0$. We have $\lim_{t \to 0^+} f(t) = \infty$ and $\lim_{t \to \infty} f(t) = \infty$. Thus, there exist $0<L_0,M_0<\infty$ such that $\tilde \ell_P(\mu)<\tilde \ell_P(\mu_0)$ whenever $|\mu|<L_0$ or $|\mu|>M_0$. Thus no maximizer can exist outside the mass-restricted class.

Define the following mass-restricted covariance mixture class as:
$$\bar{\mathcal{H}}_{A,L_0,M_0}(K):=\{\Sigma(\mu)=\int_A K(\alpha,\D)\mu(d\alpha):\mu\in \mathcal{M}_A,\;\,L_0\le|\mu|\le M_0\}.$$ 
It suffices to show that there $\ell_P$ attains its maximum over $\bar{\mathcal{H}}_{A,L_0,M_0}(K)$.

We first note that $\bar{\mathcal{H}}_{A,L_0,M_0}(K)$ is bounded, since for any $\Sigma(\mu) \in \bar{\mathcal{H}}_{A,L_0,M_0}(K)$,
\begin{equation*}
    \|\Sigma(\mu)\|_F \le \sqrt{n} \|\Sigma(\mu)\|_{\rm op} \le \sqrt{n} \lambda^\star M_0 <\infty
\end{equation*}
where we use \eqref{prop:existence:lmax} for the second inequality. 
Also, $\bar{\mathcal{H}}_{A,L_0,M_0}(K)$ is closed. Recall that $\Sigma(\mu)_{ii} = \int_AK(\alpha,0)d\mu(\alpha) =|\mu|$ for any $1\le i\le n$.
We can write $\bar{\mathcal{H}}_{A,L_0,M_0}(K) $ as 
\begin{equation*}
    \bar{\mathcal{H}}_{A,L_0,M_0}(K) = \mathcal{H}_A(K) \cap \{\Sigma:\, L_0\le \Sigma_{11} \le M_0\}.
\end{equation*}
The first set $\mathcal{H}_A(K)$ is closed by Lemma~\ref{lem:closed}. The second set is also closed since $\Sigma \to \Sigma_{11}$ is a continuous map and $[L_0,M_0]$ is a closed interval. Therefore, $\bar{\mathcal{H}}_{A,L_0,M_0}(K)$ is closed.
Since $\bar{\mathcal{H}}_{A,L_0,M_0}(K)$ is closed and bounded in the finite-dimensional space $\R^{n\times n}$, it is compact.

Finally, since $\ell_P(\Sigma)$ is continuous in $\Sigma$ on $\bar{\mathcal{H}}_{A,L_0,M_0}(K)$, $\ell_P$ attains its maximum on $\bar{\mathcal{H}}_{A,L_0,M_0}(K)$. 
\end{proof}

\begin{remark}[Extension to unbounded support]
A full analysis for the NPMLE with unbounded candidate support set is beyond the scope of this work, but we outline some necessary considerations for extending our results to this setting. First, we think a reasonable unrestricted choice of the candidate support is the compactified real line $U=[0,\infty]$. Typically, $K(\alpha,\D)$ will continuously approach a degenerate limit as the range parameter $\alpha\to\infty$. For example, consider $K(\alpha,h)=\exp\{-h^2/(2\alpha^2)\}$. Then $\underset{\alpha\to\infty}{\rm \lim} K(\alpha,\D)=K(\infty,\D)=\mathbf{J}$, where $\mathbf{J}$ is the matrix of all ones. This matrix is rank one, regardless of the number of points $n$, and thus is not full rank for $n>1$. On the other hand, for each finite $\alpha$ and any collection of distinct points $s_1,...,s_n$, the matrix $K(\alpha,\D)$ is strictly positive definite, and the mixture covariance $\int_{[0,\infty]}\,K(\alpha,\D)\,\mu(d\alpha)$ is positive definite so long as $\mu([0,\infty))>0$. In this example, for the set $[0,\infty)$ \textit{not} including infinity, the set of mixture covariances $\mathcal{H}_{[0,\infty)}(K)$ is not closed. In the compactified real line case $U=[0,\infty]$ with $\infty\in U$, however, $\mathcal{H}_U(K)$ can be shown to be closed, so that including the boundary $\infty$ avoids the problem of non-closedness. However, our existence proof for the NPMLE would still require modification, since our proof uses strict lower boundedness of the eigenvalues of $K(\alpha,\D)$ over $\alpha\in[0,U]$, for $U<\infty$. In the unbounded setting, it seems necessary to consider the spectral properties of $K(\alpha,\D)$ in the so called ``flat limit'' as $\alpha\to\infty$, as in \citet{barthelme2021spectral}.

\end{remark}

Next, we show that the maximum likelihood estimator $\hat{\Sigma}_{ML}^{(K)}$ can be written as $$\hat{\Sigma}_{ML}^{(K)}=\int K(\alpha,\D)\hat{\mu}_{ML}^{(K)}(d\alpha),$$ for a discrete measure $\hat{\mu}_{ML}^{(K)}$ with a finite support set. Let $N$ denote the number of unique nonzero values of the interpoint distances $\|\mathbf{s}_i-\mathbf{s}_j\|_2$, $i,j=1,...,n$. Let $0=h_0\leq\dots\leq h_{N-1}$, denote the unique interpoint distance values. We show there exists a discrete representing measure $\hat{\mu}_{ML}^{(K)}$ for $\hat{\Sigma}_{ML}^{(K)}$, supported on $N$ or fewer points:

\begin{proposition}\label{prop:finite_support}
Assume the same conditions as in Proposition~\ref{prop:existence}.

Then there exists a discrete representing measure $\hat{\mu}_{ML}^{(K)}$ for $\hat{\Sigma}_{ML}^{(K)}$ supported on $N$ or fewer points.
\end{proposition} 

The number of unique interpoint distances $N$ is a fundamental quantity in discrete geometry. \citet{erdos1946sets} studied this ``distinct distances problem" and showed the upper bound $N=O(n/\sqrt{\log n})$ for an $m\times m$ square lattice with $n=m^2$ points. For an $m\times m$ square lattice arrangement of points, $n=m^2$, and the number of support points in the discrete representing measure in Proposition~\ref{prop:finite_support} is bounded of the order $m^2/\sqrt{\log m}$. For a general arrangement of $n$ points in the plane, \citet{guth2015erdHos} showed that the number of unique distances $N$ is \textit{at least} of the order $n/\log n$.

\begin{proof}
Let $0=h_0<h_1<\dots<h_{N-1}$ denote the distinct interpoint distances appearing in $\D$. For any isotropic covariance function $C$, the resulting covariance matrix can be written as $\Sigma=C(\D)$, where $C(\D)_{ij}=C(\D_{ij})$. Furthermore, given the distance matrix $\D$, the covariance matrix $\Sigma$ resulting from $C$ is completely determined by the vector $[C(h_0),\dots,C(h_{N-1})]$. Since $h_0=0$ and $K(\alpha,0)=1$, 
we write $\mathbf{h} = [h_1,\dots,h_{N-1}]^\top \in \R^{N-1}$, and define $K(\alpha,\mathbf{h}) := [K(\alpha,h_1), K(\alpha,h_1),\dots,K(\alpha,h_{N-1})]^\top \in \R^{N-1}$.

Let $\hat{\mu}_{ML}^{(K)}$ be a representing measure of $\hat{\Sigma}^{(K)}_{ML}$ so that 
\begin{equation*}
    \hat{\Sigma}^{(K)}_{ML} = \int_{A} K(\alpha,\D) d\hat{\mu}_{ML}^{(K)}(\alpha).
\end{equation*}
This covariance matrix $\hat{\Sigma}^{(K)}_{ML}$ is completely determined by the vector $\hat{\mathbf{a}}\in \R^{N}$, defined as
\[\hat{\mathbf{a}}_i = \int K(\alpha, h_i) d\hat{\mu}_{ML}^{(K)}(\alpha),\quad i=0,\dots,N-1.\] 
By Proposition~\ref{prop:existence}, the total mass $m:=\hat{\mu}_{ML}^{(K)}([0,U])$ is positive. Define the probability measure $\hat{\nu} := \hat{\mu}_{ML}^{(K)}/m$. Then
\begin{align*}
    &\frac{1}{m}\hat{a}_0 = \frac{1}{m}\int_A K(\alpha,0) d\hat{\mu}_{ML}^{(K)}(\alpha) = 1,\\
    &\hat{\mathbf{v}}:=\frac{1}{m}(\hat{a}_1,\dots,\hat{a}_{N-1})^\top = \int_A K(\alpha,\mathbf{h}) d\hat{\nu}(\alpha).
\end{align*}

Define $S:= \{K(\alpha,\mathbf{h}): \alpha \in [0,U\}\}\subset\mathbb{R}^{N-1}$. Since $\alpha\to K(\alpha,\mathbf{h})$ is continuous and $[0,U]$ is compact, $S$ is compact. Hence, its convex hull $\operatorname{Conv}(S)$ is compact and therefore closed. It follows that $\hat{\mathbf{v}}\in \operatorname{Conv}(S)$.  By Carath\'{e}odory's theorem, $\hat{\mathbf{v}}$ can be written as a convex combination of $N$ or fewer points from $S$. That is, there exist points $\alpha_1,\dots,\alpha_{N} \in [0,U]$, and weights $w_1,\dots,w_N \ge 0$ with $\sum_{i=1}^N w_i = 1$ such that 
\begin{equation*}
    \hat{\mathbf{v}} = \sum_{i=1}^N K(\alpha_i,\mathbf{h}) w_i
\end{equation*}
Then for the discrete measure $\hat{\mu}_{ML}^{(K)} = \sum_{i=1}^N (m w_i) \delta_{\alpha_i}$,
\begin{equation*}
    \hat{\mathbf{a}} = \sum_{i=1}^N\begin{bmatrix}
        1\\
        K(\alpha_i,\mathbf{h}) 
    \end{bmatrix} (m  w_i) = \int  \begin{bmatrix}
        K(\alpha_i, h_0)\\
        K(\alpha_i, h_1)\\
        \dots\\
        K(\alpha_i,h_{N-1}) 
    \end{bmatrix} d\hat{\mu}_{ML}^{(K)}(\alpha).
\end{equation*}
Therefore, $\hat{\mu}_{ML}^{(K)}$ is a discrete representing measure for $\hat{\Sigma}_{ML}^{(K)}$ supported on at most $N$ points in $[0,U]$.

\end{proof}

\begin{remark}[When the base function $K(\alpha,h)$ is a totally positive kernel] The kernel $K(\alpha,h)$ is said to be totally positive, if the matrix with entries $K(\alpha_i,h_j)$ has all of its minors nonnegative, whenever $\alpha_1<...<\alpha_n$ and $h_1<...<h_{n}$. Fundamental results regarding total positivity are given in \citet{karlin1968total}, and in a statistical context, \citet{lindsay1993uniqueness} used total positivity to establish uniqueness results and improved support size bounds for mixture density models. In \citet{song2024weighted}, a similar approach was used to bound the support set size for a nonparametric mixture covariance function estimation problem in a Markov chain setting. It may be possible to leverage total positivity to establish uniqueness results for the NPMLE mixing measure for mixtures of exponential and Gaussian covariances, since the exponential and Gaussian kernel families $K(\alpha,h)=\exp(-h/\alpha)$ and $K(\alpha,h)=\exp(-h^2/(2\alpha^2))$ are totally positive kernels. This would improve on Proposition~\ref{prop:finite_support}, which only establishes the existence of a particular discrete, finitely supported representing measure for $\hat{\Sigma}^{ML}$. Proposition~\ref{prop:finite_support} does not rule out potential alternative representing measures with more support points, or even alternative measures that are not discrete. 

However, the total positivity approach has at least one limitation, which is that not all interesting base covariance kernels are totally positive. For example, \citet{schoenberg1938metric} showed that any isotropic covariance function valid for all point arrangements in three dimensions can be written as $\int_{[0,\infty]}\frac{\sin (h/\alpha)}{(h/\alpha)}\mu(d\alpha)$. But $K(\alpha,h)=\frac{\sin(h/\alpha)}{h/\alpha}$ is not a totally positive kernel, since $\frac{\sin(h/\alpha)}{h/\alpha}$ takes negative values for certain values of $\alpha$ and $h$, e.g., $\alpha_1=1$ and $h_1=3\pi/2$, which leads to $\det\{K(\alpha_1,h_1)\}=K(\alpha_1,h_1)<0$.
\end{remark}

\section{Simulation Study}
\label{sec:Simulation_Study}

We investigate the approximation error involved in using finite mixture approximations to estimate spatial covariance functions. Furthermore, we perform simulation studies to evaluate the performance of our mixture-based nonparametric spatial covariance estimator by comparing it with most commonly used parametric methods and the nonparametric method proposed by \cite{Choi2013} for estimating spatial covariance functions along with the CAR-INLA model.

\subsection{Approximation Error of Finite Mixture Approximations}
\label{sub_sec:effectiveness_finite_mixture_approximations}
We evaluate the approximation error in using finite mixtures for estimating the spatial covariance functions by considering examples without the noise introduced by simulated data. For demonstrations, we consider a zero-mean Gaussian process on a regular two-dimensional grid of 900 spatial locations that covers the domain $[0,10]\times[0,10]$ using a sequence of 30 equally spaced values along each dimension. We use $S=\mathbf{Y} \mathbf{Y} ^\top=\Sigma$ where $\Sigma\in \mathbb{R}^{n\times n}$ is the true covariance matrix corresponding to $n$ observed locations to compute the log-likelihood and its gradients involved in obtaining the NPMLEs of the weights of the mixing measure, as explained in Section~\ref{sec:Methods}. For $\Sigma$, we consider the popular Mat\'{e}rn covariance function with $\nu=0.5$ and $1.5$. The Mat\'{e}rn function follows the form mentioned in Section~\ref{sub_sec:simulation_settings} and we set $\sigma=1$ for simplicity and $\theta=2$ such that the effective range of the correlation function is at least half of the maximum lag distance.

Consequently, we estimate the spatial covariance function using the proposed mixture-based nonparametric method using NPML estimation. We evaluate the approximation error of the  method under different grid sizes $s=5,10,50$ and $100$ by comparing the estimated covariances with the true covariances. As shown in Figures~\ref{fig:approximation_error_0.5} and~\ref{fig:approximation_error_1.5} corresponding to $\Sigma$ with $\nu=0.5$ and $1.5$ respectively, the approximation by the mixture-based nonparametric method using NPML estimation (\textit{FGM ML}) is reasonable except for very coarse grids and shows little variation for moderately large values of the grid size.

\begin{figure}[htbp!]
  \centering
  \subfigure[$\nu=0.5$]{%
    \includegraphics[width=0.48\textwidth]{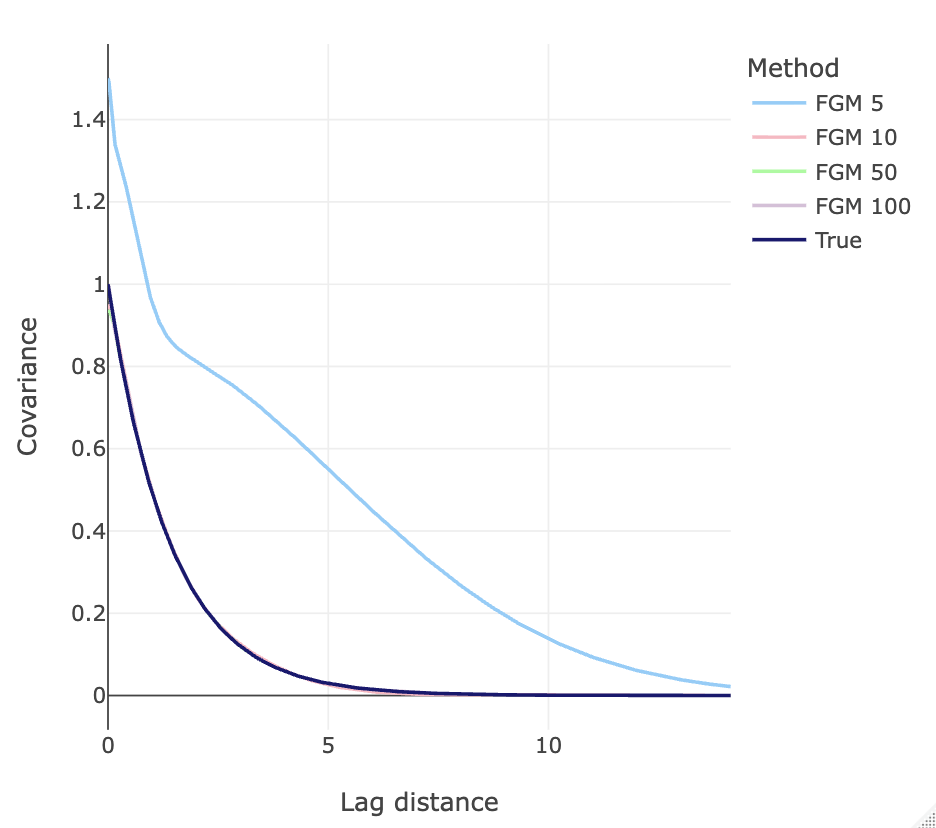}%
    \label{fig:approximation_error_0.5}%
  }
  \hfill
  \subfigure[$\nu=1.5$]{%
    \includegraphics[width=0.48\textwidth]{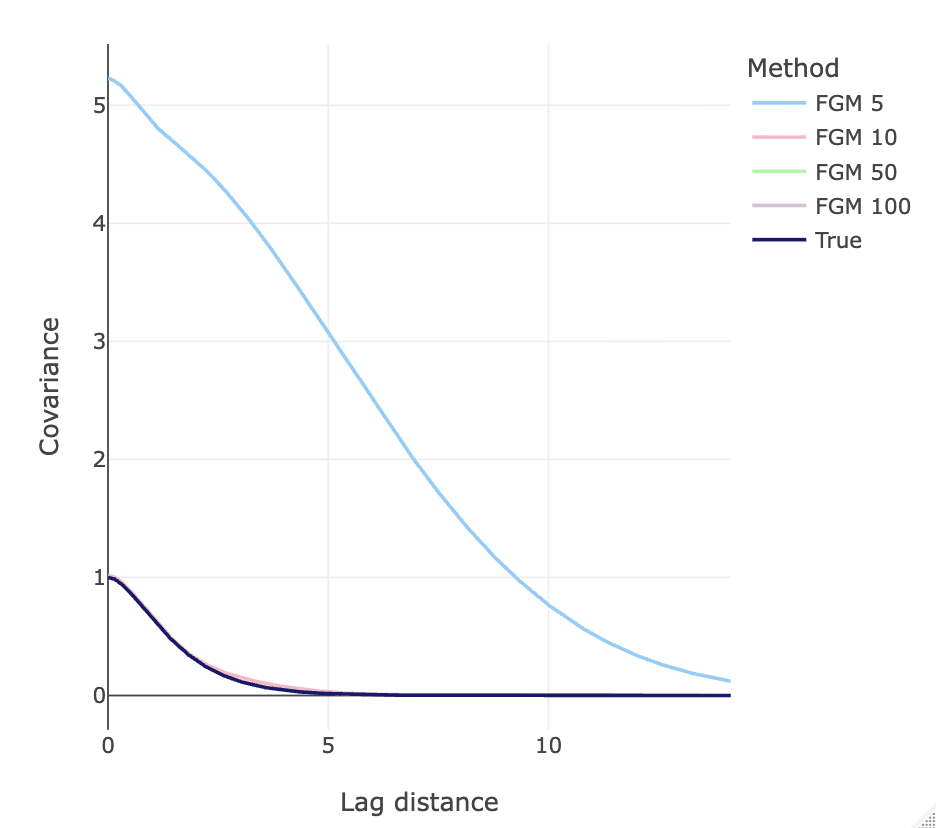}%
    \label{fig:approximation_error_1.5}%
  }
  \caption{Covariances estimated using mixture-based nonparametric method  using NPML estimation, \textit{FGM ML} when grid size, $s=5,10,50$ and $100$ relative to true covariance using Mat\'{e}rn with: 
    (a) $\nu=0.5$ and 
    (b) $\nu=1.5$.}
  \label{fig:approximation_error}
\end{figure}

\subsection{Empirical Comparison with Competing methods}
\paragraph*{Simulation Settings}
\label{sub_sec:simulation_settings}

We consider data generated on a regular two-dimensional grid of 400 spatial locations covering the domain $[0,10]\times[0,10]$ using a sequence of 20 equally spaced values along each dimension. We assume that the data-generating process follows a zero-mean Gaussian process. Given that $Y(\mathbf{s})$ is the Gaussian random field of interest, $Z(\mathbf{s})$ is a weakly stationary process with an isotropic covariance function $C_Z$, $C_Z(\mathbf{s}_1,\mathbf{s}_2):=\text{Cov}(Z(\mathbf{s}_1),Z(\mathbf{s}_2))$ such that $E[Z(\mathbf{s})]=0$ $\forall \mathbf{s}\in\mathcal{D}$ and that $\epsilon(\mathbf{s})$ is an error process independent of $Z(\mathbf{s})$ such that $E[\epsilon(\mathbf{s})]=0$, $\text{Var}(\epsilon(\mathbf{s}))=\sigma_{\epsilon}^2$ $\forall\mathbf{s}\in \mathcal{D}$ and $\text{Cov}(\epsilon(\mathbf{s}_1),\epsilon(\mathbf{s}_2))=0$, $\forall \mathbf{s}_1\neq \mathbf{s}_2$, the simulation model is given by:
\begin{equation*}
Y(\mathbf{s})=Z(\mathbf{s})+\epsilon(\mathbf{s})
\end{equation*}
The covariance function \(C_Y\) of \(Y(\mathbf{s})\) is,
\begin{equation*}
C_Y(\mathbf{s}_1, \mathbf{s}_2):=\text{Cov}(Y(\mathbf{s}_1),Y(\mathbf{s}_2))=C_Z(\mathbf{s}_1, \mathbf{s}_2)+\sigma_e^21[\mathbf{s}_1=\mathbf{s}_2].
\end{equation*}

To assess the performance of our nonparametric method, we consider data generated under varying configurations of nugget effects and different parametric covariance models. Consequently, we set the error standard deviation, i.e., $\sigma_e$ to 0, 1 and 2 to investigate the cases with no nugget effect and with varying magnitudes of nugget effects. Furthermore, we consider different parametric covariance structures, $C_Z$ of $Z(\mathbf{s})$, namely, the exponential, Mat\'{e}rn, spherical and linear Mat\'{e}rn models. Parametric models such as exponential, Mat\'{e}rn and spherical models are selected because they represent covariance structures that are most commonly observed in data. Moreover, the linear Mat\'{e}rn model is used to examine the performance of the proposed method under complex covariance structures. We consider the following covariance models for $Z(\mathbf{s})$ in the data generation process. Let $C:\mathbb{R}_+\to\mathbb{R}_+$ be a covariance function.\\

\begin{itemize}
\item Exponential model:
\begin{equation*}
C(h;\sigma^2,\theta)=\sigma^2\exp(-\frac{h}{\theta})
\end{equation*}

\item Mat\'{e}rn model:
\begin{equation*}
C(h;\sigma^2,\theta,\nu)=\sigma^2\frac{1}{\Gamma(\nu)2^{\nu-1}}(\frac{2\sqrt{\nu}h}{\theta})^\nu \kappa _{\nu}(\frac{2\sqrt{\nu}h}{\theta})
\end{equation*}

\item Spherical model:
\begin{equation*}
C(h;\sigma^2,\theta) = 
\begin{cases}
\sigma^2\{1-\frac{3h}{2\theta}+\frac{1}{2}(\frac{h}{\theta})^3\} & \text{if } h \leq \theta \\
0 & \text{if } h > \theta
\end{cases}
\end{equation*}

\item Linear Mat\'{e}rn model:
\begin{equation*}
C(h;\sigma^2,\theta_1,\theta_2,\nu)=\frac{1}{2}\sigma^2\frac{1}{\Gamma(\nu)2^{\nu-1}}
((\frac{2\sqrt{\nu}h}{\theta_1})^\nu \kappa _{\nu}(\frac{2\sqrt{\nu}h}{\theta_1})+
(\frac{2\sqrt{\nu}h}{\theta_2})^\nu \kappa _{\nu}(\frac{2\sqrt{\nu}h}{\theta_2}))
\end{equation*}
\end{itemize}

In these models, $h$ is the lag distance, $\sigma$ is the variance parameter, $\theta, \theta_1$ and $\theta_2$ are the range parameters, and $\nu$ is the smoothness parameter. $\nu$ determines the smoothness of the process with larger values implying greater smoothness. The exponential and Gaussian covariance functions are special cases of the Mat\'ern covariance function with $\nu$ set to $0.5$ and $\infty$, respectively. Note that $h\in [0, 14.14214]$ as a result of the configurations used to set up the grid. We set the variance parameter, i.e., $\sigma$ to 1 for simplicity, without loss of generality. The values of the range and smoothness parameters in the models are chosen such that the effective range of the correlation function is at least half of the maximum lag distance. Consequently, $\theta$ is set to 2, 2 and 5 in the exponential, Mat\'{e}rn and spherical models, respectively. $\theta_1$ and $\theta_2$ are set to 1 and 2 in the linear Mat\'{e}rn model. $\nu$ is set to 1.5 in the Mat\'{e}rn and linear Mat\'{e}rn models. Figure~\ref{fig:simulation_covariance} in Appendix~\ref{sub:appendix_simulation} demonstrates the above covariance functions used in simulating data. For each nugget effect - covariance model combination, we generate 100 independent realizations from the Gaussian random field, where each realization is a vector of length 400, corresponding to the observations at 400 distinct spatial locations in the two-dimensional grid.

\paragraph*{Comparison Methods}
\label{sub_sec:simulation_estmiating_covariance}

We assess the performance of the proposed nonparametric method by comparing it with the true covariance model with zero-mean (\textit{True}), parametric covariance-based methods, another nonparametric covariance-based method, and the CAR-INLA model. Consequently, we use four parametric methods that are the most frequently used, the nonparametric method proposed by \cite{Choi2013} and our mixture-based nonparametric method to estimate the covariance function of each realization from the Gaussian random field and make spatial predictions. The parameters in all parametric and nonparametric covariance-based methods are estimated using both the WLS and ML estimation methods for each simulated realization. Furthermore, we apply the CAR-INLA model to the simulated data to generate spatial predictions using neighborhood structures.

Although the simulated process has a true mean of zero, we fit all models under an intercept-only mean structure, with the intercept estimated from the data along with other parameters since we do not have knowledge of the true mean structure in reality. When the parameters of the covariance functions are estimated using WLS, a linear regression model is fitted to the simulated data to estimate the intercept as the first step. The spatial covariance parameters are then estimated using WLS at the residual level, assuming that the residuals of the fitted linear regression model are a realization from a stationary isotropic process. A similar approach is followed for the nonparametric method proposed by \cite{Choi2013} under ML estimation. The intercept is estimated using the ordinary least squares (OLS) method, and then the parameters of the nonparametric covariance estimator are estimated using ML estimation using the residuals of the linear regression model. Under parametric covariance-based methods and our mixture-based nonparametric method, we estimate the intercept jointly with the spatial covariance parameters using ML estimation. In the CAR-INLA model, the intercept and spatial random effects are estimated jointly in a Bayesian CAR model using INLA. The estimated intercept is obtained using the posterior marginal distributions of the coefficient.
%likfit gives GLS estimates for betas under ML estimation

The first three parametric models follow the forms mentioned earlier in Section~\ref{sub_sec:simulation_settings} and the Gaussian model is defined as follows, with notation consistent with the first three models.\\ 

\begin{itemize}
    \item Exponential model: Follows the form mentioned earlier in Section~\ref{sub_sec:simulation_settings}.
    \item Mat\'{e}rn model: Follows the form mentioned earlier in Section~\ref{sub_sec:simulation_settings}.
    \item Spherical model: Follows the form mentioned earlier in Section~\ref{sub_sec:simulation_settings}.
    \item Gaussian model:
\begin{equation*}
C(h;\sigma^2,\theta)=\sigma^2\exp(-\frac{h^2}{\theta^2})
\end{equation*}
\end{itemize}

In estimating the nugget effect and the covariance function with respect to parametric models, we estimate the error variance as well as the variance, range, and smoothness parameters using both the WLS and ML estimation methods for each simulated realization. We obtain parameter estimates using the WLS and ML methods using the \textit{variofit} and \textit{likfit} functions from the R package \textit{geoR} \citep{Diggle1998}, respectively. \textit{variofit} function estimates the parameters of the theoretical variogram by minimizing the weighted sum of squared differences between the empirical and theoretical variograms. We set the weights to be the number of pairs of observations contained in each bin. Note that the estimation of the Mat\'{e}rn model requires the estimation of four parameters due to the additional smoothness parameter. Thus, we fit the Mat\'{e}rn model under two scenarios: first, all four parameters are estimated simultaneously (\textit{Mat\'{e}rn}); second, $\nu$ is fixed at $1.5$ while the remaining three parameters are estimated (\textit{Mat\'{e}rn 1.5}). Consequently, we apply parametric covariance-based methods, namely \textit{exponential WLS, exponential ML, Mat\'{e}rn WLS, Mat\'{e}rn ML, Mat\'{e}rn 1.5 WLS, Mat\'{e}rn 1.5 ML, spherical WLS, spherical ML, Gaussian WLS,} and \textit{Gaussian ML}.

The nonparametric method proposed by \cite{Choi2013} leverages the idea that a valid covariance function can be constructed using completely monotone functions. They use the representation of completely monotone functions where the finite measure is a non-decreasing function and estimate the finite measure using a linear combination of B-splines over $[0,1]$ where the coefficients of B-splines are constrained to be non-negative. They define completely monotone functions by integrating the product of the support of the measure raised to the squared lag distance and the obtained B-spline basis functions over the support of the measure. The defined completely monotone functions serve as basis functions in deriving the estimator of the covariance function. Consequently, they approximate the completely monotone covariance function using a linear combination of completely monotone functions with non-negative coefficients. They utilize the Cox–de Boor recursion formula of B-splines to obtain an explicit form of these completely monotone functions that avoids the computation of the integral in them. Given that $p$ is the order, $m$ is the number of knots for the B-spline bases, $\{\beta_j,j=1,...,m+p\}$ are the non-negative coefficients to be estimated, and $\{f_j^{[q]},j=1,...,m+p\}$ are a set of completely monotone functions, the nonparametric method proposed by \cite{Choi2013} takes the following form. \
\begin{equation*}
\hat{C}(h)=\sum_{j=1}^{m+p}\beta_jf_j^{[q]}(h^2)
\end{equation*}

In our nonparametric method, we use the number of support points $s=100$, while in the nonparametric method proposed by \cite{Choi2013}, we set order $p=3$ and the number of knots $m=5$ as suggested by the authors of the paper. \cite{Choi2013} presents the performance of their method based on the coefficients estimated using the WLS estimation method only, arguing that the ML estimation method can suffer from numerical problems and is highly sensitive to the initial parameter values. In this paper, we evaluate the performance of their method based on both WLS and ML estimation methods (\textit{Bspl WLS} and \textit{Bspl ML}) to be consistent with the parametric methods and our mixture-based nonparametric method employed using WLS and NPML estimation (\textit{FGM WLS} and \textit{FGM ML}) as explained in Section~\ref{sec:Methods}. Additionally, \cite{Choi2013} slightly modified their method under the WLS estimation approach during data application to account for the observed nugget effect in the data. Given that $\hat{C}$ is the estimated covariance using their nonparametric model and $\hat{C}_E$ is the empirical covariance estimator, the estimated nugget effect $\hat{\eta}$ in their modified approach is as follows. Under the modified approach, $\hat{C}$ is computed by obtaining the covariance estimator without using the empirical estimator evaluated at zero. In this paper, we apply the nonparametric method proposed by \cite{Choi2013} along with this modified approach, which will be referred to as \textit{Bspl WLS Nugget}. Consequently, we apply nonparametric covariance-based methods, namely \textit{Bspl WLS, Bspl ML, Bspl WLS Nugget, FGM WLS,} and \textit{FGM ML}.
\begin{equation*}
\hat{\eta} = \hat{C}_E(0)-\hat{C}(0)
\end{equation*}

In the CAR-INLA model, spatial dependence is modeled via a prior using a neighborhood adjacency structure. We use a queen neighborhood structure so that grid cells that share an edge or vertex are considered neighbors. Such a structure assumes that diagonal neighbors are correlated almost as much as horizontal or vertical neighbors, and it is reasonable since we assume that the covariance functions considered in data generation and estimation are isotropic. Thus, it gives a closer approximation to the underlying Gaussian process in simulated data while making the \textit{CAR-INLA} model comparable to covariance-based methods.

\paragraph*{Comparison Metrics}
To evaluate the estimation and predictive performance of the proposed nonparametric method and other methods mentioned above, we use the integrated squared error (ISE) and the mean squared prediction error (MSPE) respectively. Let $C_0=C_Y$ be the true covariance structure that is used to simulate the data, and $\hat{C}$ be the estimated covariance function using a parametric or nonparametric method. ISE, which measures how accurately the methods estimate the covariance function is given by:

\begin{equation*}
\text{ISE}=\int_0^\infty \{C_0(h)-\hat{C}(h)\}^2dh.
\end{equation*}

The integral can be calculated over a finite interval $[0,h_m]$ where $h_m$ is one third of the maximum distance observed, since $C_0(h)\to 0$ as $h\to \infty$. Let $N$ be the number of discretized sub-intervals, $\Delta h=\frac{h_m}{N}$ and $h_i=\frac{h_m}{N}.i$ where $i=1,..,N$. In our simulation studies, we set $N=100$. Consequently, the ISE is approximated using the Riemann sum as follows.

\begin{equation*}
\text{ISE}\approx\sum_{i=1}^{N} \{C_0(h_i)-\hat{C}(h_i)\}^2\Delta h
\end{equation*}

MSPE which measures the predictive performance of the methods is calculated by comparing the predictions made on a test dataset using simple kriging based on the covariances estimated by different covariance-based methods and the CAR-INLA model. We assign 20\% of the observations corresponding to 20\% of the spatial locations on the grid in each simulated dataset to each test dataset. Given that $Y_{(i)}$ is the $i^{\text{th}}$ observed value of the Gaussian random field, $\hat{Y}_{(i)}$ is the predicted value of the $i^{\text{th}}$ observation and $n_{\text{test}}$ is the number of observations in the test set, we calculate the MSPE as follows. Note that only MSPE can be computed for CAR-INLA since it does not estimate a covariance function parameterized by lag distance.

\begin{equation*}
\text{MSPE}=\frac{1}{n_{\text{test}}}\sum_{i=1}^{n_{\text{test}}}\{Y_{(i)}-\hat{Y}_{(i)}\}^2
\end{equation*}

The final predictions for covariance-based methods and CAR-INLA are obtained as follows.

\paragraph*{The Proposed Nonparametric Method and Other Covariance-based Methods}
\label{sub_sec:gaussian_hiv}
Consider the model:
\begin{equation*}
Y(\mathbf{s})=\textbf{X}(\mathbf{s})^\top \boldsymbol{\beta}+Z(\mathbf{s})+\epsilon(\mathbf{s})
\end{equation*}

where $\mathbf{X}(\mathbf{s}) = [1]$ represents the design vector containing only the intercept term in this case and $\boldsymbol{\beta} = [\beta_0]$ is the corresponding coefficient vector. $Z(\mathbf{s})$ is a zero-mean weakly stationary process and $\epsilon(\mathbf{s})$ is an error process independent of $Z(\mathbf{s})$ such that $E[\epsilon(\mathbf{s})]=0$, $\text{Var}(\epsilon(\mathbf{s}))=\sigma_{\epsilon}^2$ $\forall\mathbf{s}\in \mathcal{D}$ and $\text{Cov}(\epsilon(\mathbf{s}_1)_j,\epsilon(\mathbf{s}_2)_{j})=0$, $\forall \mathbf{s}_1\neq \mathbf{s}_2$.
Suppose that $\hat{\bbeta}$'s are the OLS or ML estimates obtained as described above.
Let $r(\mathbf{s})$ be the residual at location $\mathbf{s}$. Then
\begin{equation*}
r(\mathbf{s})= Y(\mathbf{s})-\textbf{X}(\mathbf{s})^\top \hat{\boldsymbol{\beta}}.
\end{equation*}
Note that
\begin{equation*}
C_Y(\mathbf{s}_1, \mathbf{s}_2):=\text{Cov}(Y(\mathbf{s}_1),Y(\mathbf{s}_2))=C_Z(\mathbf{s}_1, \mathbf{s}_2)+\sigma_e^21[\mathbf{s}_1=\mathbf{s}_2].
\end{equation*}
where $C_Z(\mathbf{s}_1, \mathbf{s}_2):=\text{Cov}(Z(\mathbf{s}_1),Z(\mathbf{s}_2))$. Consequently, $C_Y$ can be estimated using the proposed nonparametric method as described in Section~\ref{sec:Methods}. Similarly, the covariance function of $Y(\mathbf{s})$ can be estimated using the four parametric methods and the nonparametric method proposed by \cite{Choi2013} using WLS and ML estimation. Additional details related to these models are provided at the beginning of this section.

We use simple kriging to make spatial predictions for the random field, $Y(\mathbf{s})$ based on the covariances estimated using the four parametric methods and the two nonparametric methods. Let $\hat{r}(\mathbf{s})$ be the prediction for $r(\mathbf{s})$ at the location $\mathbf{s}$ and $\hat{C}$ be the estimated covariance. Then, using kriging,
\begin{equation}
\hat{r}(\mathbf{s}) =
\begin{pmatrix}
r_{1} \\
\vdots \\
r_{n}
\end{pmatrix}^{\!\top}
\begin{pmatrix}
\hat{C}(h_{11}) & \cdots & \hat{C}(h_{1n}) \\
\vdots & \ddots & \vdots \\
\hat{C}(h_{n1}) & \cdots & \hat{C}(h_{nn})
\end{pmatrix}^{-1}
\begin{pmatrix}
\hat{C}(h_1) \\
\vdots \\
\hat{C}(h_n)
\end{pmatrix}
\label{eq:kriging_prediction}
\end{equation}
where $\mathbf{s}_1,...,\mathbf{s}_n$ are the locations in the training set, $h_{ij} = \|\mathbf{s}_i - \mathbf{s}_j\|_2$ and $h_{i} = \|\mathbf{s} - \mathbf{s}_i\|_2$ for $i,j=1,\dots,n$.

The final prediction for $Y(\mathbf{s})$ is obtained by mapping the OLS or ML estimate to the spatial prediction at the location $\mathbf{s}$, obtained using kriging based on estimated covariances. Then, the final prediction for $Y(\mathbf{s})$ is given by;
\begin{equation*}
\hat{Y}(\mathbf{s})=\textbf{X}(\mathbf{s})^\top \boldsymbol{\hat{\beta}}+\hat{r}(\mathbf{s})
\end{equation*}

Note that the final prediction for the \textit{True} model is simply the spatial prediction from kriging using the true covariance $C_0$ instead of $\hat{C}$ since we simulate using a zero-mean Gaussian process.

\paragraph*{CAR-INLA Model}
\label{sub_sec:car_inla_hiv}

Let $\epsilon'(\mathbf{s})$ be the error term associated with the response observed at the location $\mathbf{s}$. Assume that $\epsilon'(\mathbf{s})\sim \mathcal{N}(0,1)$. Consider the model:
\begin{equation*}
Y(\mathbf{s})=\mu(\mathbf{s})+\epsilon'(\mathbf{s})
\end{equation*}
\begin{equation*}
\mu(\mathbf{s})=\textbf{X}(\mathbf{s})^\top\boldsymbol{\beta'}+\psi(\mathbf{s})
\end{equation*}
where $\mathbf{X}(\mathbf{s}) = [1]$ is the design vector that contains only the intercept term in this case, $\boldsymbol{\beta} = [\beta_0]$ is the vector of fixed effect coefficients to be estimated and $\psi(\mathbf{s})$ is the spatial random effect for location $\mathbf{s}$. Let $n'_{\mathbf{s}}$ be the number of neighbors of node $\mathbf{s}$ and $\tau$ be the precision parameter. Then

\begin{equation*}
\psi(\mathbf{s}) \mid \psi(\mathbf{k}),\, \mathbf{s} \ne \mathbf{k},\, \tau \sim \mathcal{N}(\frac{1}{n'_{\mathbf{s}}}\sum_{\mathbf{s}\sim \mathbf{k}}\psi(\mathbf{k}),\,\frac{1}{n'_{\mathbf{s}}\tau})
\end{equation*}
Here, $\mathbf{s} \sim \mathbf{k}$ indicates that the two nodes $\mathbf{s}$ and $\mathbf{k}$ are neighbors. Define $\theta_1=\log \tau$ so that the prior can be defined on $\theta_1$. We fit the CAR-INLA model assuming that the response follows a Gaussian distribution. We consider the Besag intrinsic conditional autoregressive (ICAR) model for spatial effects and define the prior on $\theta_1$ as $\theta_1=\log(\tau) \sim \text{Log-Gamma}(1,\, 0.5\times 10^{-4})$.

\paragraph*{Results}
Tables~\ref{tab:ISE_ML_s1} and~\ref{tab:MSPE_ML_s1} summarize our results in terms of the ISE and MSPE values, respectively, calculated with respect to the CAR-INLA model and different covariance-based methods using ML estimation in the data simulated under $\sigma_e = 1$. The results corresponding to the ISE and MSPE values for the data simulated under $\sigma_e = 0$ and $2$ are given in Tables~\ref{tab:ISE_ML} and~\ref{tab:MSPE_ML} in Appendix~\ref{sub:appendix_simulation}, respectively. Furthermore, additional simulation results with covariance-based methods estimated using WLS and REML methods are given in Tables~\ref{tab:ISE_WLS}-~\ref{tab:MSPE_REML} in Appendix~\ref{sub:appendix_simulation}. The method names and abbreviations appearing in the tables follow the definitions introduced earlier in this section.

First, note that the means of ISEs are relatively higher when covariances are estimated using the WLS method, particularly in parametric methods such as exponential, Mat\'{e}rn and Gaussian models. We mainly rely on results based on ML estimation for subsequent comparisons of ISEs. Additionally, ISE values based on ML estimates appear to be better than those based on REML estimates. According to the tables, parametric methods based on ML estimates generally yield the lowest mean ISEs across all simulation models and nugget effect configurations, however, they often result in considerably high failure percentages. For example, as shown in Table~\ref{tab:ISE_ML_s1}, when data are generated under $\sigma_e = 1$ (i.e., in the presence of a nugget effect), the exponential, Mat\'{e}rn or Gaussian models based on ML estimates yield the lowest mean ISEs in all simulation models. Nevertheless, the \textit{Mat\'{e}rn ML} and \textit{Gaussian ML} models result in failure percentages of 15-24\% and 14\%, respectively, in  this scenario. Furthermore, parametric models appear to result in significantly higher mean ISEs and standard errors when misspecified. For instance, the exponential model based on ML estimates leads to relatively lower mean ISEs in many cases. However, it results in a relatively higher mean ISE and a standard error when data are generated using the Mat\'{e}rn model without a nugget effect. It must be further noted that the spherical model based on ML estimates results in relatively higher mean ISEs, both when correctly and incorrectly specified under different nugget effect settings. Meanwhile, our mixture-based nonparametric method based on NPML estimates, \textit{FGM ML} appears to produce favorable results in terms of estimation error with no failures across varying simulation settings. Our method seems to perform better than the other nonparametric method using ML estimation, especially when there is a nugget effect.

We further evaluate the prediction error of different methods using MSPE values.
%In general, the estimation and prediction performance of all models decrease as the nugget effect increases.
It should be noted that the minor differences in the failure percentages corresponding to the calculation of the ISE and MSPE values arise from the training-test data split used in the calculation of the MSPE. In general, ML estimation results in lower prediction errors compared to WLS estimation, across different estimation methods, except in a few cases, particularly when the \textit{Bspl WLS Nugget} method is used in the presence of a nugget effect as opposed to \textit{Bspl ML}. However, \textit{Bspl WLS} often report lower performance than its ML counterpart, \textit{Bspl ML}. Furthermore, ML estimates often produce similar or slightly better results than REML estimates in terms of mean MSPEs. Among models based on ML estimates, the Mat\'{e}rn model generally yields the lowest mean MSPEs after the true covariance matrix. Nevertheless, as noted earlier, the Mat\'{e}rn model appears to fail in certain cases when all four parameters are estimated using the ML method, and in some cases the failure percentages are as high as 52\%. The failure percentages of the Mat\'{e}rn model using ML estimation tend to increase as the magnitude of the nugget effect increases. When $\sigma_e=2$ in the simulated data, the model fails in approximately 42\% of cases or more, irrespective of the simulation model being used. While fixing $\nu$ at 1.5 and estimating the remaining three parameters using the ML method appears to avoid failures in the estimation of the Mat\'{e}rn model, it sometimes reduces the performance of the model. For example, \textit{Mat\'{e}rn ML} with unfixed $\nu$ outperforms \textit{Mat\'{e}rn 1.5 ML} across all simulation models under $\sigma_e=2$. In fact, the \textit{Mat\'{e}rn ML} model yields the lowest mean MSPEs after the true covariance matrix when data are simulated with $\sigma_e=2$ across all simulation models except for the case where our \textit{FGM ML} model yields the best performance under data simulated using the spherical model. Consequently, our nonparametric \textit{FGM ML} method seems to outperform the Mat\'{e}rn model with $\nu$ fixed at 1.5 across all simulation models under this nugget effect setting. In general, our \textit{FGM ML} method yields prediction errors as low as the best performing parametric model and the true covariance model under different configurations of the nugget effect and simulation models without any failures. Furthermore, the \textit{FGM ML} method outperforms the \textit{CAR-INLA} model and the other nonparametric method under ML estimation, \textit{Bspl ML} in all simulation settings, especially when there is no nugget effect and when there is a relatively stronger nugget effect (i.e., $\sigma_e = 2$) while the \textit{CAR-INLA} model outperforms the parametric \textit{spherical ML} model. Additionally, the \textit{CAR-INLA} model outperforms the other nonparametric method under ML estimation, \textit{Bspl ML}. Consequently, the proposed mixture-based nonparametric method based on NPML estimates, \textit{FGM ML} seems to perform well in terms of mean MSPEs and their standard errors with no failures, across all simulation models, and in cases with and without nugget effects, and as the strength of the nugget effect increases.

\begin{table}[ht]
\centering
\caption{Mean of ISEs, its standard error (in parentheses) (unit: $10^{-2}$) and percentage of failures of four parametric methods (exponential, Mat\~ern, spherical and Gaussian) with initial value for $\sigma_e = 10^{-1}$ and two nonparametric methods (B-splines method with order $p=3$ and number of knots $m=5$ and finite Gaussian mixtures method with number of support points $s=100$) using ML estimation when $\sigma_e = 1$}
\label{tab:ISE_ML_s1}
\vspace{-0.5em}
\scriptsize
\par\vspace{0.8em}
\par\smallskip
\begin{tabular}{ccccc}
  \hline
Method & Exponential & Matérn & Linear Matérn & Spherical \\ 
  \hline
\makecell{Exponential ML} & \makecell{19.49 \\[0.6ex] (2.29); 0\%} & \makecell{18.80 \\[0.6ex] (3.33); 0\%} & \textbf{\makecell{14.83 \\[0.6ex] (1.11); 0\%}} & \makecell{36.15 \\[0.6ex] (6.45); 0\%} \\ 
  \makecell{Matérn ML} & \textbf{\makecell{18.06 \\[0.6ex] (1.66); 15\%}} & \textbf{\makecell{13.57 \\[0.6ex] (1.84); 24\%}} & \makecell{15.10 \\[0.6ex] (1.24); 21\%} & \makecell{22.78 \\[0.6ex] (3.05); 22\%} \\ 
  \makecell{Matérn 1.5 ML} & \makecell{19.67 \\[0.6ex] (1.81); 0\%} & \makecell{14.73 \\[0.6ex] (2.28); 0\%} & \makecell{15.00 \\[0.6ex] (1.11); 0\%} & \makecell{30.90 \\[0.6ex] (5.64); 0\%} \\ 
  \makecell{Spherical ML} & \makecell{73.48 \\[0.6ex] (0.18); 0\%} & \makecell{73.46 \\[0.6ex] (0.16); 0\%} & \makecell{73.37 \\[0.6ex] (0.13); 0\%} & \makecell{90.99 \\[0.6ex] (0.23); 0\%} \\ 
  \makecell{Gaussian ML} & \makecell{23.77 \\[0.6ex] (1.46); 5\%} & \makecell{14.15 \\[0.6ex] (1.25); 14\%} & \makecell{19.19 \\[0.6ex] (1.17); 0\%} & \textbf{\makecell{20.59 \\[0.6ex] (1.81); 14\%}} \\ 
  \makecell{Bspl ML} & \makecell{126.75 \\[0.6ex] (11.58); 0\%} & \makecell{96.98 \\[0.6ex] (8.04); 0\%} & \makecell{85.98 \\[0.6ex] (2.53); 0\%} & \makecell{169.55 \\[0.6ex] (26.29); 0\%} \\ 
  \makecell{FGM ML} & \makecell{23.20 \\[0.6ex] (2.00); 0\%} & \makecell{19.17 \\[0.6ex] (2.08); 0\%} & \makecell{19.51 \\[0.6ex] (1.22); 0\%} & \makecell{28.72 \\[0.6ex] (3.55); 0\%} \\ 
   \hline
\end{tabular}
\end{table}

\begin{table}[ht]
\centering
\caption{Mean of MSPEs, its standard error (in parentheses) (unit: $10^{-2}$) and percentage of failures of true covariance model, CAR-INLA model, and four parametric methods (exponential, Mat\~ern, spherical and Gaussian) with initial value for $\sigma_e = 10^{-1}$ and two nonparametric methods (B-splines method with order $p=3$ and number of knots $m=5$ and finite Gaussian mixtures method with number of support points $s=100$) using ML estimation when $\sigma_e = 1$}
\label{tab:MSPE_ML_s1}
\vspace{-0.5em}
\scriptsize
\par\vspace{0.8em}
\par\smallskip
\begin{tabular}{ccccc}
  \hline
Method & Exponential & Matérn & Linear Matérn & Spherical \\ 
  \hline
\makecell{True} & \makecell{133.73 \\[0.6ex] (2.19); 0\%} & \makecell{122.86 \\[0.6ex] (2.02); 0\%} & \makecell{138.10 \\[0.6ex] (2.27); 0\%} & \makecell{123.67 \\[0.6ex] (2.00); 0\%} \\ 
%  \makecell{True Int} & \makecell{133.76 \\[0.6ex] (2.19); 0\%} & \makecell{122.88 \\[0.6ex] (2.02); 0\%} & \makecell{138.13 \\[0.6ex] (2.27); 0\%} & \makecell{123.69 \\[0.6ex] (2.00); 0\%} \\ 
   \hline
\makecell{CAR-INLA} & \makecell{134.89 \\[0.6ex] (2.19); 0\%} & \makecell{124.45 \\[0.6ex] (2.03); 0\%} & \makecell{139.49 \\[0.6ex] (2.28); 0\%} & \makecell{125.53 \\[0.6ex] (2.03); 0\%} \\ 
  \makecell{Exponential ML} & \makecell{134.17 \\[0.6ex] (2.20); 0\%} & \makecell{123.50 \\[0.6ex] (2.04); 0\%} & \makecell{137.87 \\[0.6ex] (2.30); 0\%} & \makecell{124.72 \\[0.6ex] (2.04); 0\%} \\ 
  \makecell{Matérn ML} & \textbf{\makecell{133.00 \\[0.6ex] (2.28); 27\%}} & \textbf{\makecell{122.36 \\[0.6ex] (2.20); 29\%}} & \textbf{\makecell{136.61 \\[0.6ex] (2.32); 28\%}} & \makecell{124.84 \\[0.6ex] (2.46); 30\%} \\ 
  \makecell{Matérn 1.5 ML} & \makecell{134.33 \\[0.6ex] (2.20); 0\%} & \makecell{123.23 \\[0.6ex] (2.03); 0\%} & \makecell{137.88 \\[0.6ex] (2.28); 0\%} & \textbf{\makecell{124.62 \\[0.6ex] (2.04); 0\%}} \\ 
  \makecell{Spherical ML} & \makecell{184.99 \\[0.6ex] (3.75); 0\%} & \makecell{190.62 \\[0.6ex] (4.02); 0\%} & \makecell{194.19 \\[0.6ex] (3.73); 0\%} & \makecell{187.73 \\[0.6ex] (4.19); 0\%} \\ 
  \makecell{Gaussian ML} & \makecell{135.17 \\[0.6ex] (2.28); 5\%} & \makecell{123.93 \\[0.6ex] (2.11); 4\%} & \makecell{138.68 \\[0.6ex] (2.28); 1\%} & \makecell{127.08 \\[0.6ex] (2.39); 22\%} \\ 
  \makecell{Bspl ML} & \makecell{177.87 \\[0.6ex] (3.11); 0\%} & \makecell{162.22 \\[0.6ex] (2.84); 0\%} & \makecell{173.35 \\[0.6ex] (3.06); 0\%} & \makecell{169.17 \\[0.6ex] (2.93); 0\%} \\ 
  \makecell{FGM ML} & \makecell{134.62 \\[0.6ex] (2.20); 0\%} & \makecell{123.58 \\[0.6ex] (2.02); 0\%} & \makecell{138.36 \\[0.6ex] (2.29); 0\%} & \makecell{124.83 \\[0.6ex] (2.04); 0\%} \\ 
   \hline
\end{tabular}
\end{table}

Tables~\ref{tab:ISE_ML_s1},~\ref{tab:MSPE_ML_s1} and~\ref{tab:ISE_ML}-~\ref{tab:MSPE_REML} report overall failure percentages. These overall failures include cases where the model faces convergence issues and fails to produce estimates, which we call "convergence failures", and cases where the model converges but produces estimates that lead to non-finite covariances, which we call "numerical failures". Among the models that experienced failures, most were limited to convergence issues, whereas the \textit{Matérn WLS} model exhibited both convergence and numerical failures. Table~\ref{tab:failure_proportions_matern_wls} presents the convergence and numerical failure proportions of \textit{Matérn WLS} model and their standard errors by the simulation model and $\sigma_e$. The model seems to be more affected by numerical failures than by convergence failures, primarily due to the considerably high values of $\nu$ estimated by the model. These numerical failures occur in all combinations of the simulation model and $\sigma_e$. The higher proportions of numerical failures when computing ISEs appeared to result from the combination of large estimates of $\nu$ and small lag distances, which can cause the modified Bessel function of the second kind in the Matérn covariance function to attain extremely large values, leading to numerical overflow.

\begin{table}[!h]
\centering
\caption{\label{tab:failure_proportions_matern_wls}Convergence and numerical failure proportions and their standard errors (in parenthesis) for the Matérn WLS model}
\scriptsize
\centering
\begin{tabular}[t]{ccrrrr}
\toprule
\multicolumn{2}{c}{ } & \multicolumn{2}{c}{ISE} & \multicolumn{2}{c}{MSPE} \\
\cmidrule(l{3pt}r{3pt}){3-4} \cmidrule(l{3pt}r{3pt}){5-6}
$\sigma_e$ & Covariance & Convergence & Numerical & Convergence & Numerical\\
\midrule
0 & Exponential & 0.02 (0.01) & 0.15 (0.04) & 0.04 (0.02) & 0.03 (0.02)\\
0 & Matérn & 0.03 (0.02) & 0.08 (0.03) & 0.03 (0.02) & 0.03 (0.02)\\
0 & Linear Matérn & 0.00 (0.00) & 0.05 (0.02) & 0.03 (0.02) & 0.01 (0.01)\\
0 & Spherical & 0.01 (0.01) & 0.20 (0.04) & 0.02 (0.01) & 0.02 (0.01)\\
\midrule
1 & Exponential & 0.04 (0.02) & 0.21 (0.04) & 0.01 (0.01) & 0.01 (0.01)\\

1 & Matérn & 0.00 (0.00) & 0.09 (0.03) & 0.01 (0.01) & 0.01 (0.01)\\
1 & Linear Matérn & 0.00 (0.00) & 0.02 (0.01) & 0.03 (0.02) & 0.03 (0.02)\\
1 & Spherical & 0.03 (0.02) & 0.15 (0.04) & 0.04 (0.02) & 0.04 (0.02)\\
\midrule
2 & Exponential & 0.03 (0.02) & 0.16 (0.04) & 0.03 (0.02) & 0.03 (0.02)\\
2 & Matérn & 0.02 (0.01) & 0.07 (0.03) & 0.01 (0.01) & 0.00 (0.00)\\

2 & Linear Matérn & 0.00 (0.00) & 0.05 (0.02) & 0.02 (0.01) & 0.01 (0.01)\\
2 & Spherical & 0.04 (0.02) & 0.14 (0.03) & 0.04 (0.02) & 0.03 (0.02)\\
\bottomrule
\end{tabular}
\end{table}

\section{FSW Data Analysis}
\label{sec:Data_Analysis}
We use the proposed mixture-based nonparametric approach to estimate the FSW population sizes in SSA while comparing it with the CAR-INLA model and covariance-based methods, including the four parametric models and the nonparametric method proposed by \cite{Choi2013}. A description of the motivating HIV dataset and covariates data is provided in Section~\ref{sec:Data_Description}.

\paragraph*{Comparison Methods} We evaluate the predictive performance of our nonparametric methods, \textit{FGM WLS} and \textit{FGM ML} along with all other methods defined in Section~\ref{sec:Simulation_Study}, namely \textit{CAR-INLA}, parametric covariance-based methods such as \textit{exponential WLS, exponential ML, Mat\'{e}rn WLS, Mat\'{e}rn ML, Mat\'{e}rn 1.5 WLS, Mat\'{e}rn 1.5 ML, spherical WLS, spherical ML, Gaussian WLS} and \textit{Gaussian ML} and nonparametric covariance-based methods such as \textit{Bspl WLS, Bspl ML} and \textit{Bspl WLS Nugget}. The method names and abbreviations follow the
definitions introduced earlier in the previous section. We perform a five-fold cross-validation and evaluate the MSPE on the probability and logit scales, for each fold consisting of 20\% of the data (about 70 data points). Details related to computing the MSPE are provided in the previous section.

The application of the above models follows the same procedure described in Section~\ref{sec:Simulation_Study}, except for the specification of the mean structure. In Section~\ref{sec:Simulation_Study}, the mean structure consists only of an intercept; however, in this section it includes both an intercept and covariates. In our nonparametric method, we set the number of support points $s=100$, consistent with simulation studies. Meanwhile, in the nonparametric method proposed by \cite{Choi2013}, we set the number of knots $m=20$ and order $p=3$, with the number of knots slightly higher than the value suggested by the authors in the paper. During residual diagnostics, we detected two significant outliers in the data. Consequently, we evaluate the performance of the models with and without these two outliers.

\paragraph*{Results} Table~\ref{tab:coef_comparison_methods_with} presents the regression coefficient estimates obtained using OLS estimation, ML estimation with our \textit{FGM ML} model, and the CAR-INLA method, along with their standard errors based on data with outliers. We obtain similar results based on data without outliers (refer Table~\ref{tab:coef_comparison_methods_without} in Appendix~\ref{sub: appendix_real_data}). Both covariates appear to be significant at a significance level of 5\% according to OLS and ML estimates. The three methods agree with each other in terms of the direction of the relationship between the logit proportions and the covariates. It is surprising to see the negative correlation between the logit proportions and the logarithmic transformed reference population. One possible explanation could be that even though the absolute values of FSWs are higher in regions with larger populations, it might be the case that FSWs grow slower than the total reference population \citep{Abhirup2019}.

\begin{table}[ht]
\centering
\caption{Coefficient estimates and their standard errors (in parentheses) across methods based on data with outliers}
\label{tab:coef_comparison_methods_with}
\begin{tabular}{lccc}
\hline
& OLS & ML & CAR-INLA \\
\hline
Intercept
& -4.147 (0.054)
& -4.072 (0.157)
& -4.108 (0.056) \\

Urban Population Percentage
& 0.233 (0.062)
& 0.217 (0.068)
& 0.165 (0.083) \\

Reference Population
& -0.579 (0.046)
& -0.538 (0.045)
& -0.568 (0.055) \\
\hline
\end{tabular}
\end{table}

\subparagraph*{Estimating the Covariance Function}
Figures~\ref{fig:cov_est_mle} and~\ref{fig:cov_est_wls} display the covariances estimated utilizing different covariance estimation methods using ML and WLS methods, respectively, together with empirical covariances, based on data with and without outliers. Under the ML method, the estimated covariances appear approximately similar in both cases, that is, with and without outliers, with minor differences in the rate of decay. The covariances estimated using the proposed mixture-based nonparametric method based on NPML estimates, \textit{FGM ML} (dark blue line) decay relatively slowly, similar to certain parametric models such as \textit{exponential ML} and \textit{Mat\'{e}rn ML}. Meanwhile, the covariances estimated using \textit{spherical ML} model tend to decay very quickly, particularly within a lag distance of 1. The \textit{Gaussian ML} model behaves similarly to the \textit{spherical ML} model when the outliers are removed; however, it encounters computational challenges when outliers are present. The covariances estimated using the nonparametric method proposed by \cite{Choi2013} based on ML estimation, \textit{Bspl ML} depicted unusually high values at lower lags. Therefore, the corresponding line has been removed from the plots. Under the WLS method, all parametric models behave similarly with respect to estimated covariances in cases with and without outliers. In general, the covariances estimated using the two nonparametric methods based on WLS estimation; \textit{FGM WLS} (dark blue line) and \textit{Bspl WLS} (light green line) follow similar patterns while oscillating around the empirical covariances (dotted line). Note that \textit{variofit} obtains parameter estimates of parametric models by considering the empirical variogram as opposed to the empirical covariance as described in Section~\ref{sec:Simulation_Study}.

\begin{figure}[h!]
  \centering
  \subfigure[With outliers]{%
    \includegraphics[width=0.48\textwidth]{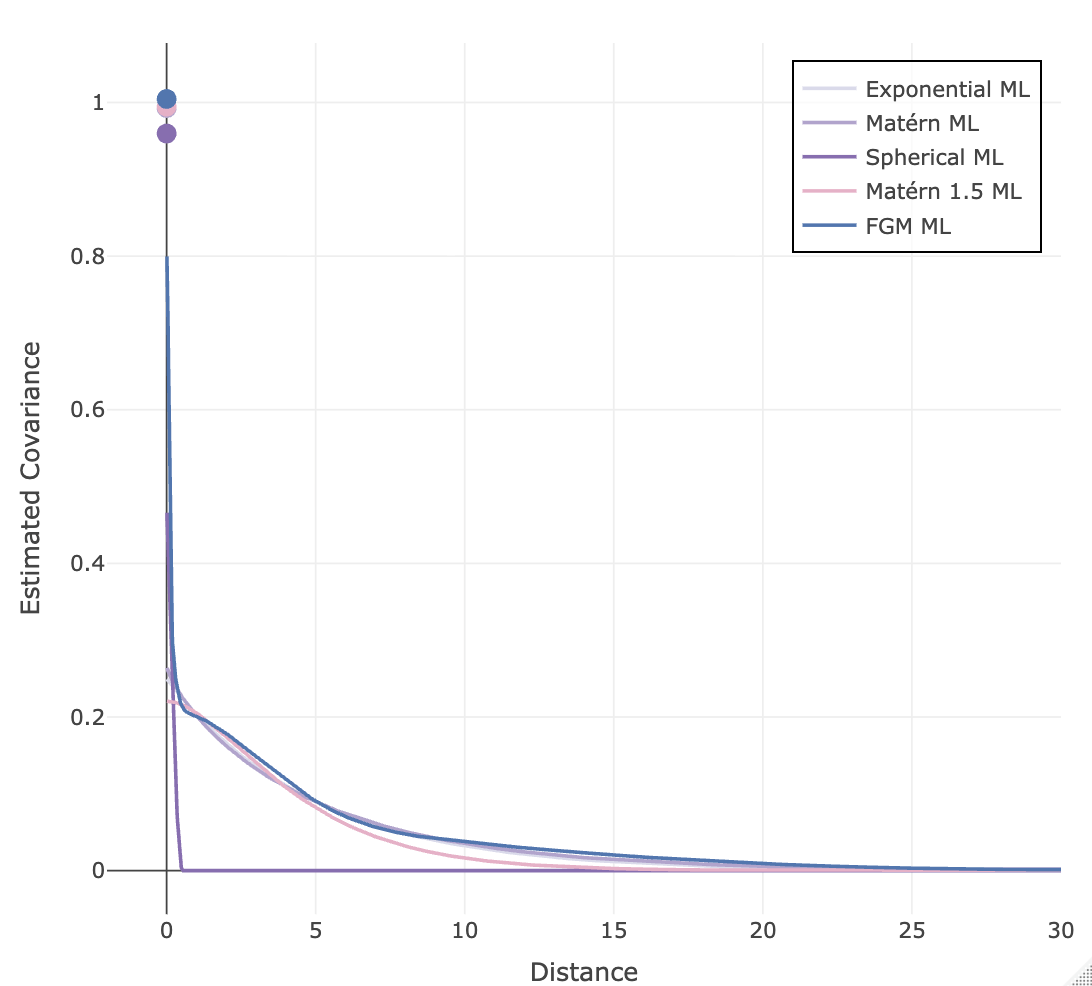}%
    \label{fig:cov_est_with_mle}%
  }
  \hfill
  \subfigure[Without outliers]{%
    \includegraphics[width=0.48\textwidth]{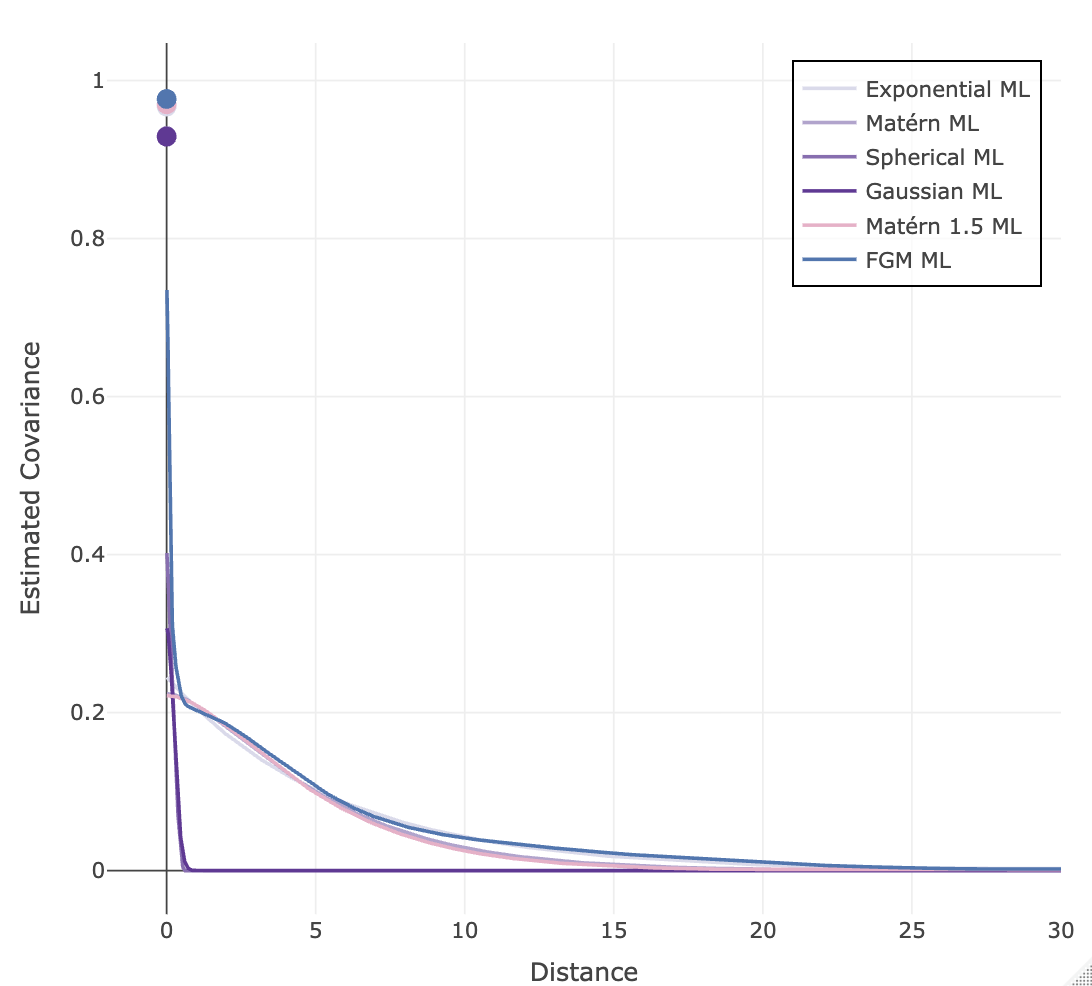}%
    \label{fig:cov_est_without_mle}%
  }
  \caption{Covariances estimated using the four parametric models and our nonparametric model using ML estimation: 
    (a) using the dataset with outliers, 
    (b) using the dataset without outliers.}
  \label{fig:cov_est_mle}
\end{figure}

\begin{figure}[h!]
  \centering
  \subfigure[With outliers]{%
    \includegraphics[width=0.48\textwidth]{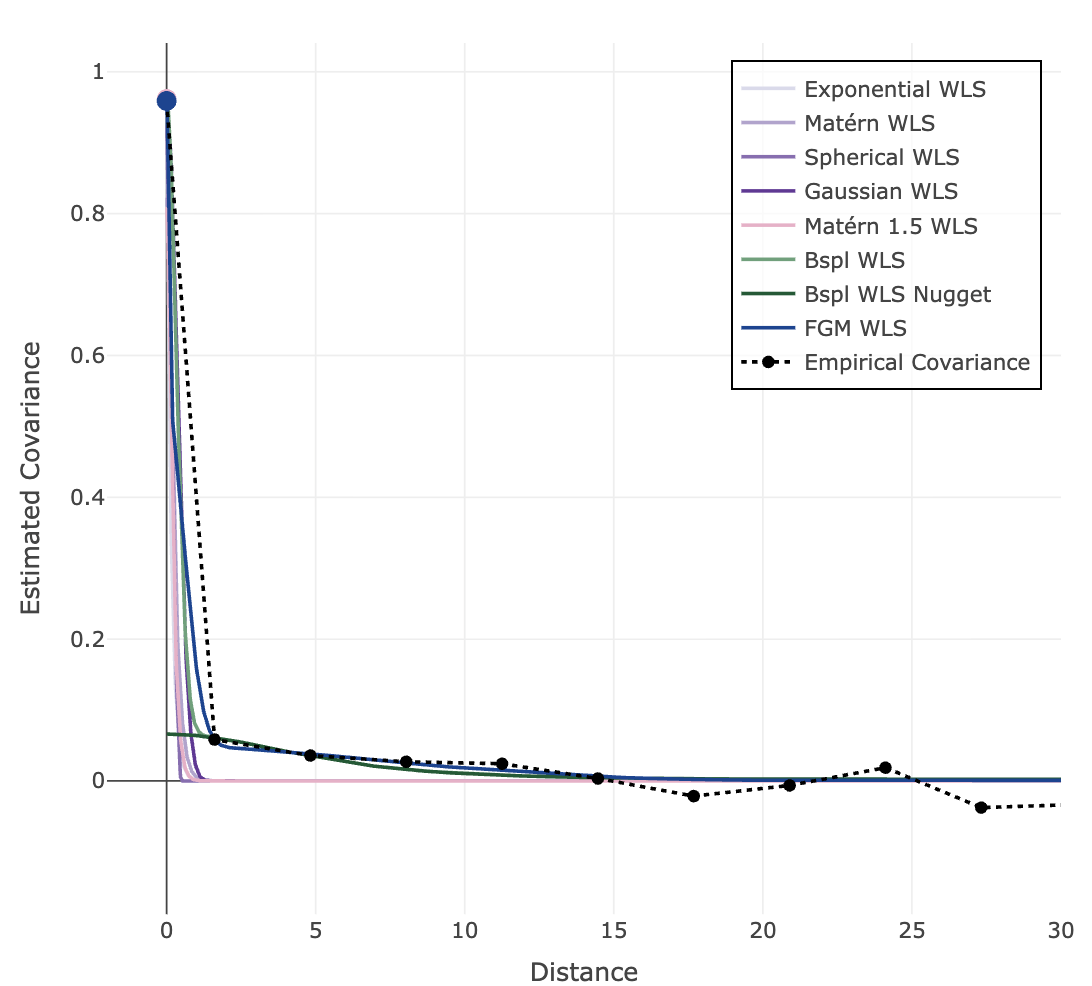}%
    \label{fig:cov_est_with_wls}%
  }
  \hfill
  \subfigure[Without outliers]{%
    \includegraphics[width=0.48\textwidth]{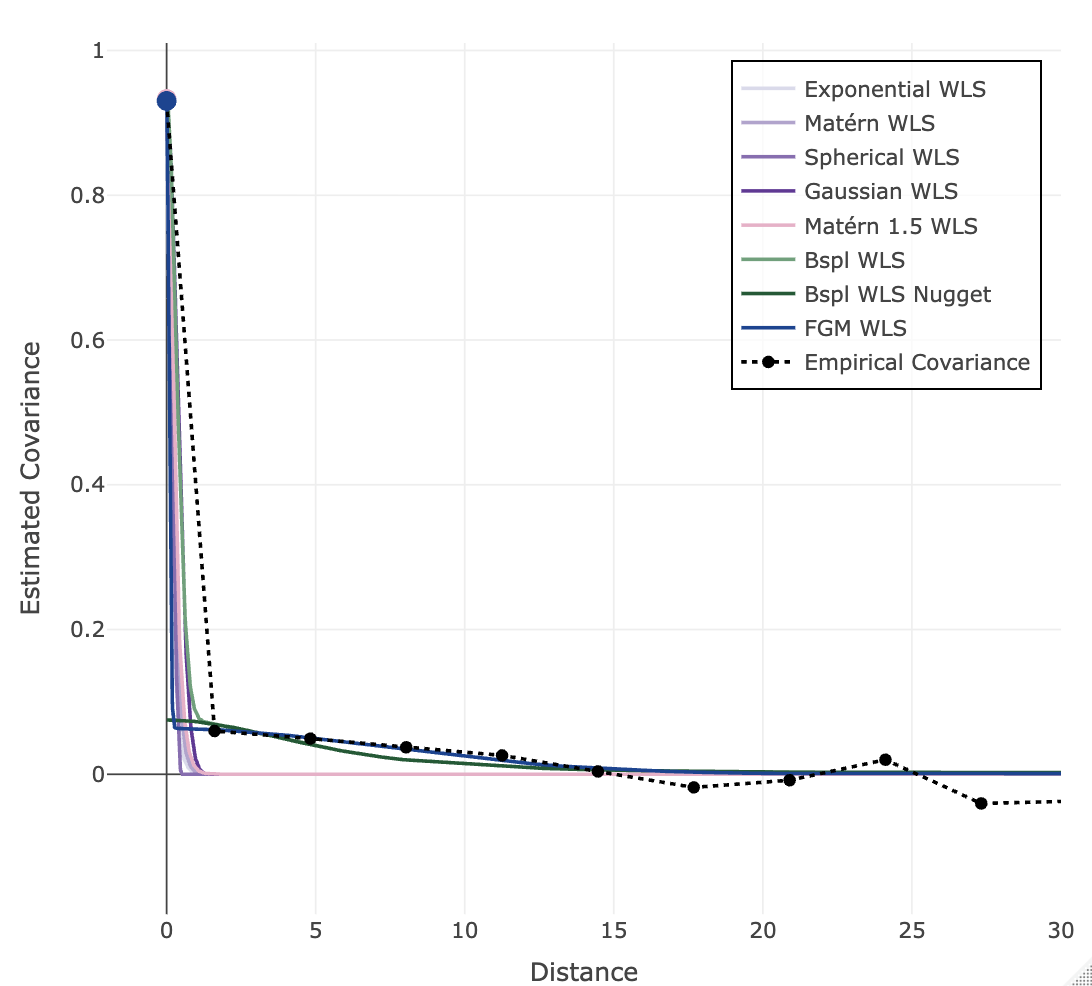}%
    \label{fig:cov_est_without_wls}%
  }
  \caption{Empirical covariances and covariances estimated using the four parametric models and the two nonparametric models using WLS estimation: 
    (a) using the dataset with outliers, 
    (b) using the dataset without outliers.}
  \label{fig:cov_est_wls}
\end{figure}

\subparagraph*{Prediction Error} Table~\ref{tab:mspe_realdata_ML1} summarizes the results of the five-fold cross-validation using means of MSPEs and their standard errors obtained on the logit and probability scales, from the CAR-INLA model and covariance-based methods using ML estimation. Additionally, we consider all covariance-based methods under WLS estimation as described earlier. The corresponding results can be found in Table~\ref{tab:mspe_realdata_WLS1} (see Appendix~\ref{sub: appendix_real_data}). Furthermore, Table~\ref{tab:mspe_realdata_REML1} in Appendix~\ref{sub: appendix_real_data} provides the results of the parametric models and our nonparametric method under REML estimation. The method names and abbreviations follow the
definitions introduced earlier in Section~\ref{sec:Simulation_Study}. As shown by the mean MSPEs on the logit scale in Table~\ref{tab:mspe_realdata_ML1}, all parametric methods based on ML estimates outperform the \textit{CAR-INLA} model regardless of the presence of outliers except for the \textit{Gaussian ML} model which reports a competitive mean MSPE. The spherical model seems to perform better than the \textit{CAR-INLA} model even with WLS estimates regardless of the presence of outliers. In general, the exponential model yields a relatively higher performance than other parametric models regardless of the estimation method used. It must also be noted that fixing the smoothness parameter, $\nu$ at 1.5 in the Mat'{e}rn covariance function results in slightly higher mean MSPEs when estimates are obtained using the WLS method. Additionally, parametric methods based on ML and REML estimates appear to outperform the corresponding model with the WLS estimation. The Gaussian model based on WLS estimates appears to be highly sensitive to the initial values used in optimizations. As shown in Tables~\ref{tab:mspe_realdata_WLS1} and~\ref{tab:mspe_realdata_WLS2}, the \textit{Gaussian WLS} model results in a considerably high mean MSPE and standard error when the initial value of $\sigma_e$ changes from $10^{-1}$ to $0$ (see Appendix~\ref{sub: appendix_real_data}). Similarly, nonparametric methods based on WLS estimates report poor performance compared to their counterparts based on ML and REML estimation. For some simulation seeds, our \textit{FGM WLS} method produced poor predictions due to parameter estimates that resulted in ill-conditioned covariance matrices. It seemed that the WLS method is sensitive to the empirical covariance estimator, which can be sensitive to the particular realization of the simulated data.

Among the \textit{CAR INLA} model and all covariance-based methods based on ML estimates, the proposed mixture-based nonparametric method, \textit{FGM ML} reports the lowest mean MSPE with a relatively low standard error on the logit scale. The exponential model based on the REML estimation with the initial value of $\sigma_e$ set to $10^{-1}$ appears to outperform our \textit{FGM ML} method. However, the performance of the \textit{exponential REML} model decreases considerably when the initial value of $\sigma_e$ is changed to $0$ in the case without outliers (see Table~\ref{tab:mspe_realdata_REML2}). In general, our \textit{FGM ML} outperforms the CAR-INLA model, all parametric methods based on WLS, ML and REML estimates and other nonparametric methods based on WLS and ML estimates. When we compare the mean MSPEs of different methods on the probability scale, the proposed \textit{FGM ML} method yields mean MSPEs as low as the best performing method.

% latex table generated in R 4.6.0 by xtable 1.8-8 package
% Sun Jul 26 17:32:12 2026
\begin{table}[ht]
\centering
\caption{Mean of MSPEs and its standard error (in parentheses) for the CAR-INLA model and the four parametric methods (exponential, Mat\'{e}rn, spherical and Gaussian) with initial value for $\sigma_e = 10^{-1}$ and the two nonparametric methods (B-splines method with order $p=3$ and number of knots $m=20$ and finite Gaussian mixtures method with number of grid points $s=100$) using ML estimation from five-fold cross-validation. Values on the logit scale are reported in units of $10^{-2}$, whereas values on the probability scale are reported in units of $10^{-4}$.} 
\label{tab:mspe_realdata_ML1}
\begin{tabular}{lcccc}
  \multicolumn{1}{c}{Method} & \multicolumn{2}{c}{Logit Scale} & \multicolumn{2}{c}{Probability Scale} \\

\cmidrule(lr){2-3} \cmidrule(lr){4-5}

 & With Outliers & Without Outliers & With Outliers & Without Outliers \\
 \hline
  \hline
CAR-INLA & 97.63 (10.43) & 94.06 (5.00) & 20.63 (7.35) & 20.66 (7.35) \\ 
  Exponential ML & 89.00 (11.82) & 84.98 (6.70) & 20.45 (7.21) & 20.35 (7.18) \\ 
  Mat\'{e}rn ML & 90.17 (12.39) & 86.14 (9.45) & 20.50 (7.29) & 23.11 (8.61) \\ 
  Mat\'{e}rn 1.5 ML & 90.08 (12.46) & 85.80 (7.36) & 20.50 (7.28) & 20.36 (7.22) \\ 
  Spherical ML & 95.53 (10.37) & 92.29 (4.76) & 20.33 (7.06) & \textbf{20.19 (6.96)} \\ 
  Gaussian ML & 98.51 (13.20) & 95.34 (7.72) & \textbf{13.72 (3.24)} & 25.89 (10.69) \\ 
  Bspl ML & 1027.99 (439.13) & 1048.64 (452.92) & 189.89 (88.30) & 190.79 (89.70) \\ 
  FGM ML & \textbf{88.17 (9.38)} & \textbf{82.42 (4.26)} & 20.56 (7.28) & 20.32 (7.16) \\ 
   \hline
\end{tabular}
\end{table}

\subparagraph*{Failure Proportions}
We observed earlier in Sections~\ref{sec:Simulation_Study} and~\ref{sec:Data_Analysis} that parametric models face certain convergence challenges. Therefore, we evaluate failure proportions of the \textit{CAR-INLA} model and all covariance-based methods with respect to the HIV dataset. Consequently, we perform 100 repeated runs of five-fold cross validation by considering 100 different random seeds to partition the dataset into five folds, on each version of the dataset, with and without outliers. In five-fold cross validation, each fold is kept for validation while using the remaining folds for training the models, resulting in five training sets within each run of cross-validation. Within each run, we estimate the parameters of all models, namely the \textit{CAR-INLA} model and all covariance-based methods using the WLS and ML estimation methods using data in each of these five training sets. Note that we also estimate parameters of parametric methods and our mixture-based nonparametric method using REML method. If the considered estimation method faces convergence issues in at least one of the five training sets under the considered random seed and fails to produce estimates, we consider the method to fail under that seed.

According to Tables~\ref{tab:failure_para} and~\ref{tab:failure_CARINLA_nonpara}, Mat\'{e}rn and Gaussian models fail to converge and produce estimates under certain seeds in both cases; with and without outliers in the dataset, under the ML and REML estimation methods. The \textit{Mat\'{e}rn ML} model seems to fail about 34\% and 52\% of the time in the two cases; with and without outliers, respectively, while the \textit{Mat\'{e}rn REML} model seems to have relatively lower failure percentages of about 16\% and 28\%, respectively. Surprisingly, the failure percentages of the Mat\'{e}rn model have increased slightly when the outliers are removed. Meanwhile, the \textit{Gaussian ML} model seems to have considerably higher failure percents of about 85\% and 70\% in the two cases. Furthermore, the failure percent of the Gaussian model increases to about 77\% under REML estimation in the case without outliers. It must be noted that the \textit{CAR-INLA} model, parametric methods such as exponential and spherical models, and all nonparametric methods, particularly our mixture-based nonparametric method based on NPML estimates, \textit{FGM ML} exhibit no convergence failures across all 100 repeated runs of five-fold cross-validation, both with and without outliers, indicating their robustness.

\begin{table}[!h]
\centering
\caption{Failure proportions of parametric covariance estimation methods and their standard errors (in parentheses) evaluated on the HIV data with and without outliers.}
\label{tab:failure_para}
\begin{tabular}{lccc|ccc}
\toprule
& \multicolumn{3}{c}{With Outliers} & \multicolumn{3}{c}{Without Outliers} \\
\cmidrule(lr){2-4} \cmidrule(lr){5-7}
Method & ML & REML & WLS & ML & REML & WLS\\
\midrule
Exponential & 0.00 (0.00) & 0.00 (0.00) & 0.00 (0.00) &
0.00 (0.00) & 0.00 (0.00) & 0.00 (0.00)\\

Mat\'ern & 0.34 (0.05) & 0.16 (0.04) & 0.00 (0.00) &
0.52 (0.05) & 0.28 (0.04) & 0.00 (0.00)\\

Mat\'ern 1.5 & 0.00 (0.00) & 0.00 (0.00) & 0.00 (0.00) &
0.00 (0.00) & 0.00 (0.00) & 0.00 (0.00)\\

Spherical & 0.00 (0.00) & 0.00 (0.00) & 0.00 (0.00) &
0.00 (0.00) & 0.00 (0.00) & 0.00 (0.00)\\

Gaussian & 0.85 (0.04) & 0.85 (0.04) & 0.00 (0.00) &
0.70 (0.05) & 0.77 (0.04) & 0.00 (0.00)\\
\bottomrule
\end{tabular}
\end{table}

\begin{table}[!h]
\centering
\caption{Failure proportions of the CAR-INLA model and nonparametric covariance estimation methods and their standard errors (in parentheses) evaluated on the HIV data with and without outliers.}
\label{tab:failure_CARINLA_nonpara}
\begin{tabular}{lcc}
\toprule
Method & With Outliers & Without Outliers\\
\midrule
CAR-INLA        & 0.00 (0.00) & 0.00 (0.00)\\
Bspl ML         & 0.00 (0.00) & 0.00 (0.00)\\
Bspl WLS        & 0.00 (0.00) & 0.00 (0.00)\\
Bspl WLS Nugget & 0.00 (0.00) & 0.00 (0.00)\\
FGM ML          & 0.00 (0.00) & 0.00 (0.00)\\
FGM REML        & 0.00 (0.00) & 0.00 (0.00)\\
FGM WLS         & 0.00 (0.00) & 0.00 (0.00)\\
\bottomrule
\end{tabular}
\end{table}

\subparagraph*{Number of Support Points in the Proposed Nonparametric Method}

To investigate the effect of the number of support points, $s$ empirically with respect to the HIV dataset, we fit our mixture-based nonparametric method using ML, REML and WLS estimation with 8 different grid sizes, $s = 10, 50, 100, 200, 300,$ $400, 500$ and $1000$ on both versions of the dataset; with and without outliers and compare the estimated covariances and MSPEs over different grid sizes. As shown in Figures~\ref{fig:grid_mle},~\ref{fig:grid_reml} and~\ref{fig:grid_wls} (see Appendix~\ref{sub: appendix_real_data} for the latter two figures), the estimated covariances using the proposed mixture-based nonparametric method remain nearly unchanged over different grid sizes regardless of the estimation method used and the presence of outliers. According to Table~\ref{tab:mspe_support_points} (see the Appendix~\ref{sub: appendix_real_data}), the MSPE changes slightly under different grid sizes under WLS estimation; however, it remains nearly the same under ML and REML estimation.

\begin{figure}[h!]
  \centering
  \subfigure[With outliers]{%
    \includegraphics[width=0.45\textwidth]{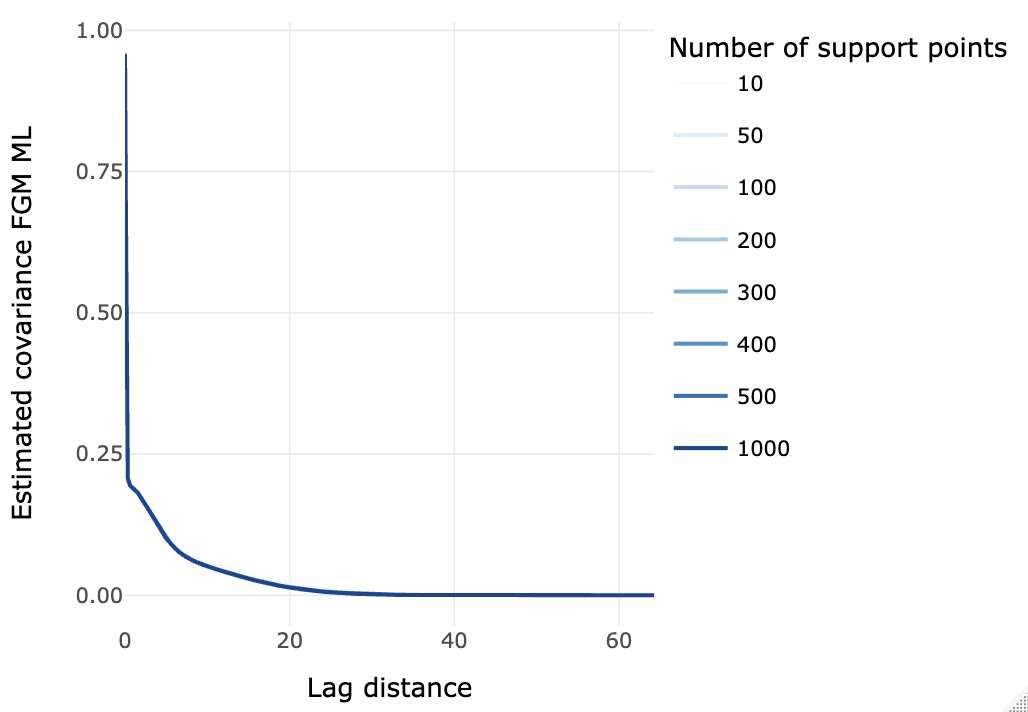}%
    \label{fig:grid_with_mle}%
  }
  \hfill
  \subfigure[Without outliers]{%
    \includegraphics[width=0.45\textwidth]{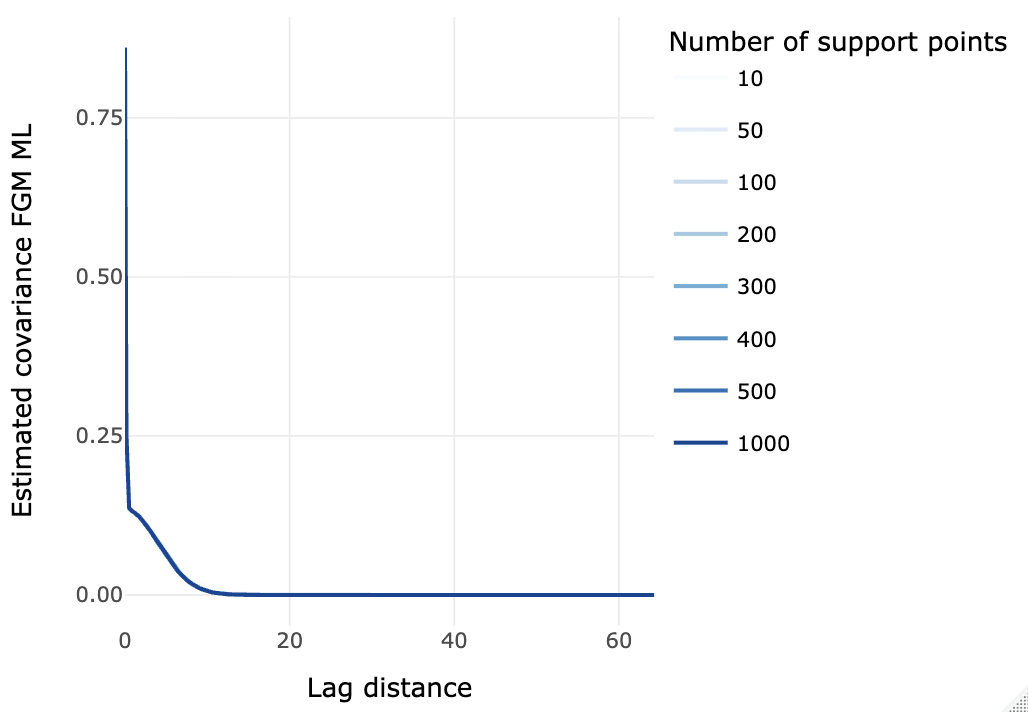}%
    \label{fig:grid_without_mle}%
  }
  \caption{Estimated covariances using the nonparametric method, FGM using ML estimation over different grid sizes: 
    (a) using the dataset with outliers, 
    (b) using the dataset without outliers.}
  \label{fig:grid_mle}
\end{figure}

\paragraph*{Estimation of FSW PSEs in SSA}
Empirical results suggest that our mixture-based nonparametric method using NPML estimation, \textit{FGM ML} outperforms other methods in estimating logit proportions of FSWs. Furthermore, all models based on ML estimates yield better results on the logit scale when outliers are removed (see Table~\ref{tab:mspe_realdata_ML1}). Consequently, we apply the \textit{FGM ML} method on the dataset without outliers to produce estimates of logit proportions and PSEs for FSWs in 44 SSA countries at sub-national level. As shown in Figures~\ref{fig:logit} and~\ref{fig:PSEs}, respectively, FSW logit proportions and PSEs appear to have spatial clusters, with areas depicting similar logit proportions and PSEs clustering together. Sub-national areas in countries such as Democratic Republic of the Congo, Ethiopia, Kenya, and South Africa appear to have relatively higher PSEs. Overall, most sub-national areas have PSEs below 4000. The estimated PSEs of subnational areas appear to vary in a wide range from 25.14 to 18391.38 with a median of 989.71.

\begin{figure}[h!]
  \centering
  \subfigure[Logit proportion estimates]{%
    \includegraphics[width=0.48\textwidth]{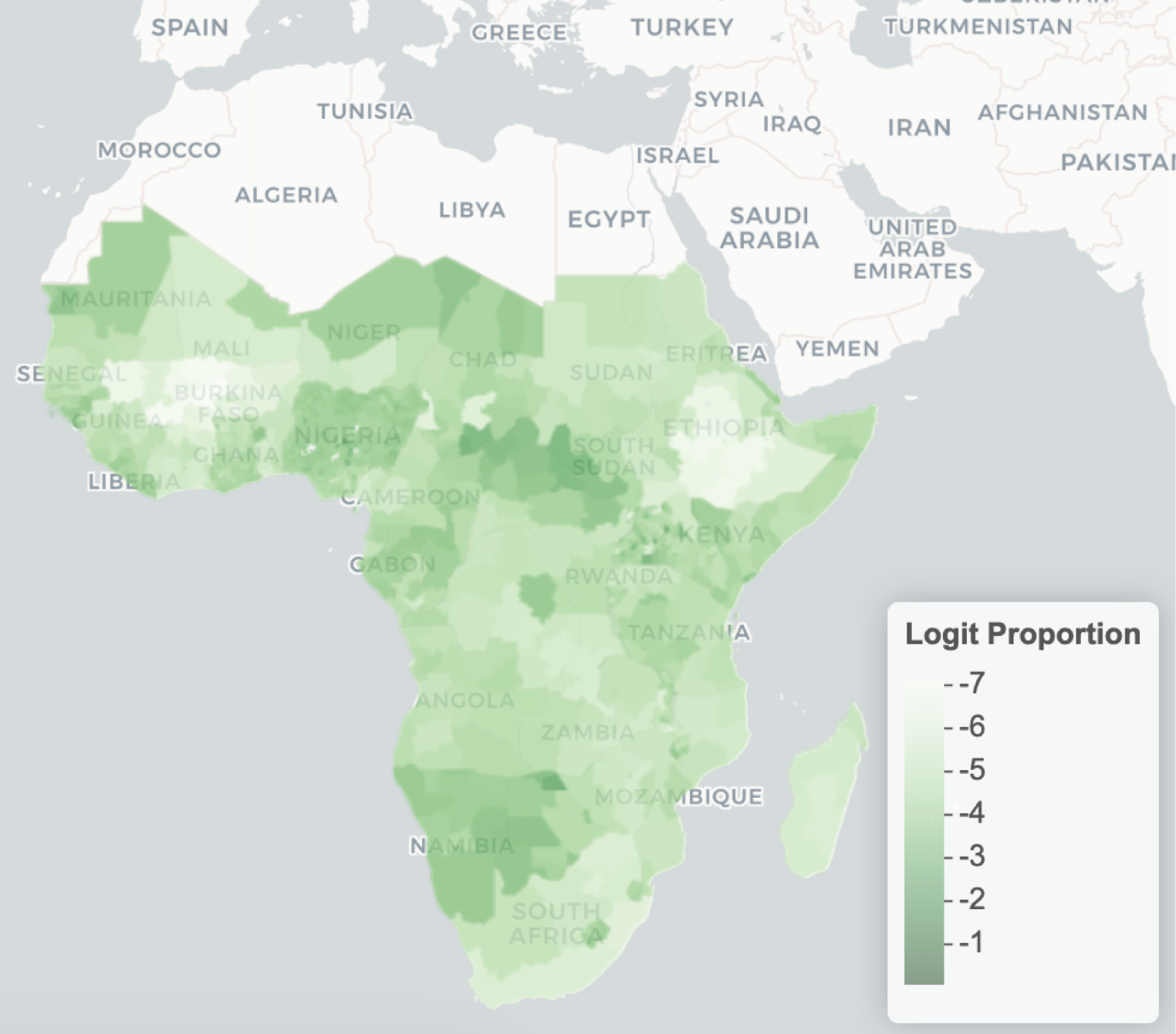}%
    \label{fig:logit}%
  }
  \hfill
  \subfigure[Population size estimates]{%
    \includegraphics[width=0.48\textwidth]{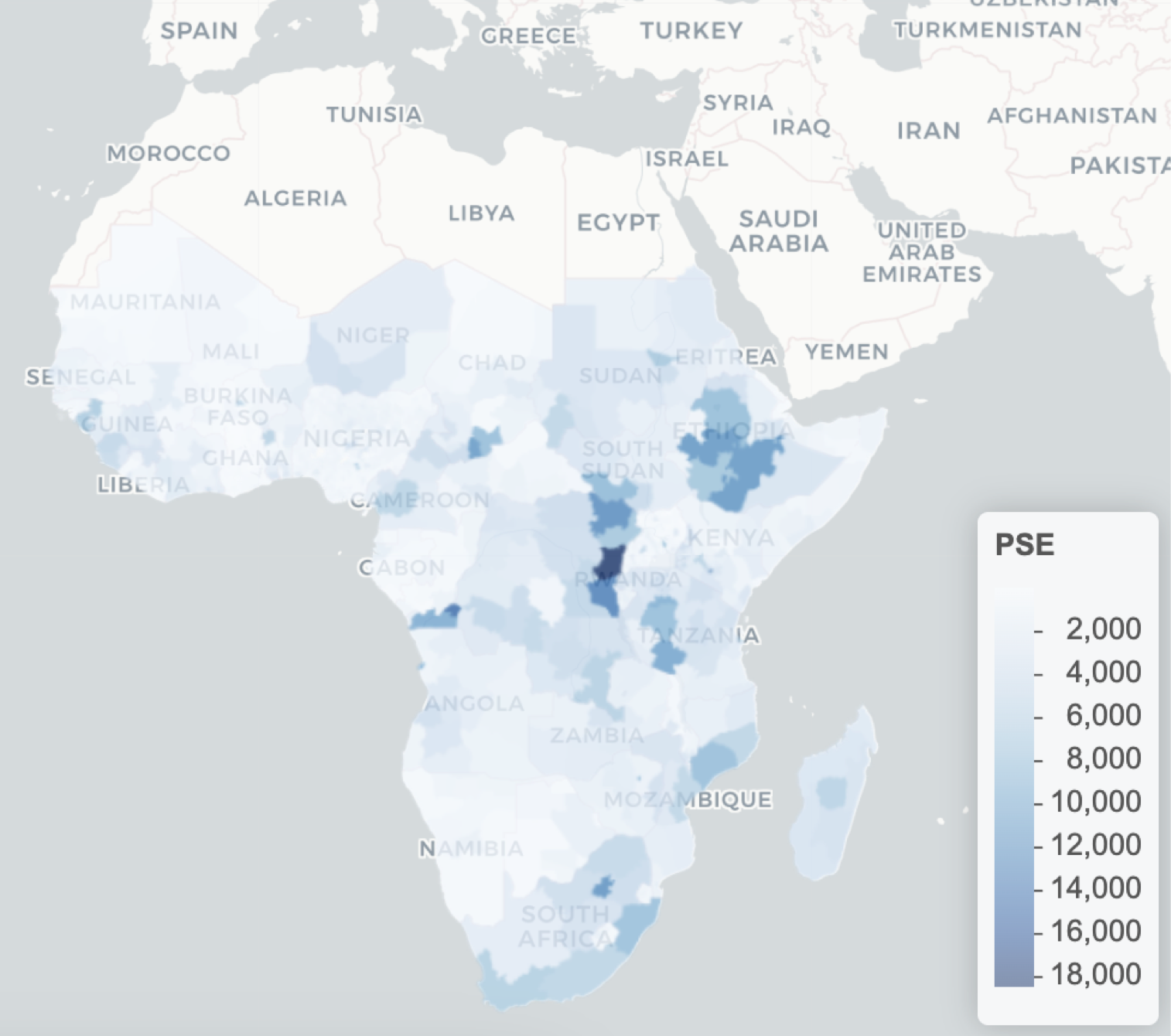}%
    \label{fig:PSEs}%
  }
  \caption{Estimates of logit proportions and population sizes of female sex workers in sub-Saharan Africa at sub-national level: 
    (a) Logit proportion estimates, 
    (b) Population size estimates.}
  \label{fig:estimates}
\end{figure}

\section{Discussion}
\label{sec:Discussion}

In this paper, we propose a novel nonparametric method based on mixtures of Gaussian kernels for estimating the spatial covariance functions with applications to obtaining PSEs of FSWs in SSA at sub-national level. We leverage the infinite mixture representation admitted by covariance functions of stationary isotropic processes valid in all dimensions and estimate the weights of the mixing measure using WLS and NPML estimation methods. Additionally, we develop computationally efficient methods to solve optimizations using non-negative least squares and second-order decent updates. We establish some theoretical properties of the NPML estimator. We show that the NPML covariance estimator exists, and that at least one NPML mixing measure can be chosen to have finite support. Furthermore, we empirically examine the approximation error in using finite approximations of infinite mixtures to estimate the covariance function relative to the true covariances, and show the effectiveness of the approach in estimating Mat\'{e}rn covariance functions across a broad range of grid resolutions, with particularly consistent performance across sufficiently fine grids. In addition, we investigate the estimated covariances from our method over different grid sizes relative to the HIV dataset and demonstrate the robustness of the method over the grid size regardless of the presence of outliers and the estimation method used. According to empirical results, the proposed nonparametric method performs well in terms of estimation as well as prediction errors, across varying simulation models, in cases with and without nugget effects and outliers, and as the strength of the nugget effect increases without any failures. Our nonparametric method using the NPML estimation approach accounts for local and long range spatial relationships in the data while capturing information in covariates. The method exhibits better performance than the CAR-INLA model and all parametric methods based on WLS, ML and REML estimates and other nonparametric methods based on WLS and ML estimates in estimating logit proportions of FSWs in SSA, thus can be useful in producing PSEs of FSWs and other key populations.

In general, ML performs similar or better than its REML counterpart. Furthermore, most models show better performance with ML estimates as opposed to WLS estimates. WLS estimates are often sensitive to empirical estimators, which are sensitive to the underlying data. Additionally, the Mat\'{e}rn model appears to result in considerably higher estimates of $\nu$ that lead to numerical failures under WLS estimation. In addition to the risk of model mis-specification and lower accuracies resulting from it, parametric methods involve certain computational challenges. We evaluate the failure percentages of different covariance estimation methods. Although nonparametric methods report no failures, certain parametric models such as Mat\'{e}rn and Gaussian models fail to produce parameter estimates, particularly in ML estimation. ML estimation involves obtaining the derivatives of the covariance function with respect to all its parameters. In Mat\'{e}rn covariance, the derivative with respect to the smoothness parameter $\nu$ is based on the derivative of the modified second-kind Bessel function with respect to $\nu$ \citep{Geoga2023}. Furthermore, \cite{Diggle2007} discussed the difficulty in estimating all three parameters in the Mat\'{e}rn model due to ridges or plateus in the log-likelihood surface and suggested selecting $\nu$. While fixing $\nu$ often appears to avoid failures in the Mat\'{e}rn function, estimating $\nu$ together with other parameters often produces higher accuracies, as indicated in simulation studies and real data analysis. These results agree with arguments made in the literature that emphasize the importance of $\nu$ \citep{Stein1999, Hong2021, Geoga2023}. \cite{Hong2021} argued that inappropriate $\nu$ can result in incorrect estimations of prediction error and inefficient predictions. Moreover, parametric models, particularly the Gaussian model, seem to be highly sensitive to the initial values of the parameters used in optimizations. The \textit{geoR} package itself encourages trying different initial values \citep{RibeiroDiggle2006}. Nevertheless, most covariance estimation methods based on ML estimates perform better than the \textit{CAR-INLA} model. Thus, covariance-based methods can be useful in producing FSW PSEs.

The proposed nonparametric method exhibits some limitations. Similar to popular kernels such as Mat\'{e}rn, the proposed nonparametric method requires that the process be stationary and isotropic. Moreover, at present, our method does not handle multiple estimates observed at the same location, hence, it can be extended to accommodate multiple estimates in the future. Another possible area of research would be to extend our nonparametric spatial covariance function to the spatio-temporal setting. Additionally, we use our mixture-based nonparametric method to estimate FSW PSEs in this study. The method can also be applied to produce PSEs of other key populations.

%%%%%%%%%%%%%%%
\bibliographystyle{plainnat}
\bibliography{refs}
\label{lastpage}

\section{Appendix}
\label{appendix}
\subsection{Proofs for Section~\ref{sec:Methods}}
\subsubsection{Proof for Lemma~\ref{lem: L_P_deriv}}\label{supp:lem:L_P_deriv}
\renewcommand{\L}{\mathcal{L}}
\newcommand{\A}{\mathbf{A}}
\begin{proof}
Recall $\Y \sim N_n(\mu, \Sigma(\btheta))$ with $\mu = \X\bbeta$, $\Sigma(\btheta) = \sum_{i=1}^s \theta_i \Sigma_i$, and
\begin{align*}
    \L_P(\btheta) &:= const +\frac{1}{2}(\Y-\X\hat{\bbeta}_\btheta )^\top \Sigma(\btheta)^{-1}(\Y-\X\hat{\bbeta}_\btheta)+\frac{1}{2}\log |\Sigma(\btheta)|.
\end{align*}
where we define
$\hat{\bbeta}_{\btheta}:  = (\X^\top \Sigma(\btheta)^{-1}\X)^{-1}\X^\top \Sigma(\btheta)^{-1}\Y$.

Let $\A:= \X(\X^\top\Sigma(\btheta)^{-1}\X)^{-1}\X^\top \Sigma(\btheta)^{-1}$. We have,
\begin{align*}
    (\Y-\X\hat{\bbeta}_\btheta )^\top \Sigma(\btheta)^{-1}(\Y-\X\hat{\bbeta}_\btheta) 
    &=\Y^\top(\mathbf{I}_n-\A)^{\top} \Sigma(\btheta)^{-1 }(\mathbf{I}_n-\A)\Y.
\end{align*}
Also,
\begin{align*}
    (\mathbf{I}_n-\A)^{\top} \Sigma(\btheta)^{-1 }(\mathbf{I}_n-\A) 
    %&= (\Sigma(\btheta)^{-1}-\A^\top \Sigma(\btheta)^{-1} )(\mathbf{I}_n-\A)\\
    &=\Sigma(\btheta)^{-1} - \Sigma(\btheta)^{-1}\A - \A^\top\Sigma(\btheta)^{-1} +\A^\top \Sigma(\btheta)^{-1} \A.
\end{align*}
Note $\Sigma(\btheta)^{-1}\A =\A^\top \Sigma(\btheta)^{-1} =  \Sigma(\btheta)^{-1} \X(\X^\top\Sigma(\btheta)^{-1}\X)^{-1}\X^\top \Sigma(\btheta)^{-1}$, and also
\begin{align*}
    \A^\top \Sigma(\btheta)^{-1} \A
    &= \Sigma(\btheta)^{-1} \X(\X^\top \Sigma(\btheta)^{-1}\X)^{-1}(\X^\top \Sigma(\btheta)^{-1} \X)(\X^\top \Sigma(\btheta)^{-1}\X)^{-1}\X^\top\Sigma(\btheta)^{-1} \\
    &= \Sigma(\btheta)^{-1}\X(\X^\top \Sigma(\btheta)^{-1} \X)^{-1}\X^\top\Sigma(\btheta)^{-1}.
\end{align*}
Therefore,
\begin{align*}
    (\mathbf{I}_n-\A)^{\top}\Sigma(\btheta)^{-1} (\mathbf{I}_n-\A) = \Sigma(\btheta)^{-1} - \Sigma(\btheta)^{-1}\X(\X^\top \Sigma(\btheta)^{-1}\X)^{-1} \X^\top\Sigma(\btheta)^{-1}=: P_\btheta,
\end{align*}
and
\begin{equation*}
    \L_P(\btheta)=const+\frac{1}{2}\Y^\top P_\btheta \Y + \frac{1}{2}\log |\Sigma(\btheta)|.
\end{equation*}

Taking the derivative with respect to $\btheta$, for each $i=1,\dots,s$,
\begin{align*}
    g(\btheta)_i = \frac{\partial}{\partial\theta_i} \L_P(\btheta) = \frac{1}{2}\Y^\top (\frac{\partial P_\btheta}{\partial \theta_i}) \Y + \frac{1}{2}\frac{\partial }{\partial \theta_i} \log |\Sigma(\btheta)|.
\end{align*}
We now use the following identity, whose proof is deferred to the end of the proof.
\begin{lemma}\label{lem:dP_theta}
For $i=1,\dots,s$
    \begin{align*}
        \frac{\partial P_\btheta}{\partial \theta_i}=-P_\btheta \Sigma_i P_\btheta.
    \end{align*}
\end{lemma}

Combining Lemma~\ref{lem:dP_theta} and the fact that $\frac{\partial}{\partial \theta_i} \log |\Sigma(\btheta)| = \operatorname{tr}(\Sigma(\btheta)^{-1}\Sigma_i)$, we have,

\begin{align*}
    g(\btheta)_i 
    &= -\frac{1}{2}\Y^\top P_\btheta \Sigma_i  P_\btheta \Y +\frac{1}{2} \operatorname{tr}(\Sigma(\btheta)^{-1}\Sigma_i )\\
    &= -\frac{1}{2}\{ (P_\btheta \Y)^\top \Sigma_i (P_\btheta \Y) - \operatorname{tr}(\Sigma(\btheta)^{-1}\Sigma_i)\}, \quad i=1,\dots,s.
\end{align*}

Now we compute Hessian and expected Hessian for $\L_P(\btheta)$. Using the chain rule and Lemma~\ref{lem:dP_theta} again,
\begin{align*}
    \frac{\partial}{\partial\theta_j} g(\btheta)_i
    &= -\frac{1}{2}\Y^\top (\frac{\partial}{\partial \theta_j} P_\btheta) \Sigma_i P_\btheta \Y  -\frac{1}{2}\Y^\top P_\btheta \Sigma_i(\frac{\partial}{\partial \theta_j} P_\btheta)\Y +\frac{1}{2} \operatorname{tr}(\frac{\partial}{\partial\theta_j}\Sigma(\btheta)^{-1}\Sigma_i)\\
    &= -\frac{1}{2}\Y^\top (-P_\btheta \Sigma_j P_\btheta) \Sigma_i P_\btheta \Y  -\frac{1}{2}\Y^\top P_\btheta \Sigma_i(-P_\btheta \Sigma_j P_\btheta)\Y +\frac{1}{2}\operatorname{tr}(\Sigma(\btheta)^{-1}\{-\frac{\partial}{\partial\theta_j}\Sigma(\btheta)\}\Sigma(\btheta)^{-1}\Sigma_i)\\
    &= \Y^\top P_\btheta \Sigma_i P_\btheta \Sigma_j P_\btheta \Y -\frac{1}{2}\operatorname{tr}(\Sigma(\btheta)^{-1}\Sigma_j \Sigma(\btheta)^{-1}\Sigma_i).
\end{align*}
since $\Y^\top (P_\btheta \Sigma_j P_\btheta) \Sigma_i P_\btheta \Y = (\Y^\top (P_\btheta \Sigma_j P_\btheta) \Sigma_i P_\btheta \Y)^\top =\Y^\top P_\btheta \Sigma_i P_\btheta \Sigma_j P_\btheta \Y .$

For the expected Hessian, note
\begin{align*}
    P_\btheta \X = \Sigma(\btheta)^{-1} \X- \Sigma(\btheta)^{-1}\X(\X^\top \Sigma(\btheta)^{-1}\X)^{-1} \X^\top\Sigma(\btheta)^{-1}\X =0.
\end{align*}
and
\begin{align*}
    P_\btheta \Sigma(\btheta)P_\btheta &= P_\btheta \Sigma(\btheta)[\Sigma(\btheta)^{-1} - \Sigma(\btheta)^{-1}\X(\X^\top \Sigma(\btheta)^{-1}\X)^{-1} \X^\top\Sigma(\btheta)^{-1}]\\
    &=P_\btheta - P_\btheta \X(\X^\top \Sigma(\btheta)^{-1}\X)^{-1} \X^\top\Sigma(\btheta)^{-1} = P_\btheta.
\end{align*}
Therefore $P_\btheta \Y\sim N(0, P_\btheta)$. Then,
\begin{align*}
    E[(P_\btheta \Y)^\top \Sigma_i P_\btheta \Sigma_j (P_\btheta \Y)] = \operatorname{tr}( (\Sigma_i P_\btheta \Sigma_j) P_\btheta ),
\end{align*}
and
\begin{align*}
   H(\btheta)_{ij} = E[ \frac{\partial}{\partial\theta_j} g(\btheta)_i] =  \operatorname{tr}(  P_\btheta \Sigma_i P_\btheta \Sigma_j) - \frac{1}{2} \operatorname{tr}(\Sigma(\btheta)^{-1}\Sigma_j \Sigma(\btheta)^{-1}\Sigma_i).
\end{align*}

\newcommand{\M}{\mathbf{M}}
\newcommand{\W}{\mathbf{W}}
\begin{proof}[Proof of Lemma~\ref{lem:dP_theta}]
For $i=1,\dots,s$, we have
\begin{align*}
    \frac{\partial}{\partial\theta_i}P_\btheta 
    &= \frac{\partial}{\partial\theta_i}\Sigma(\btheta)^{-1} - \frac{\partial}{\partial\theta_i}\{\Sigma(\btheta)^{-1}\X(\X^\top \Sigma(\btheta)^{-1}\X)^{-1} \X^\top\Sigma(\btheta)^{-1}\}\\
    &=\frac{\partial}{\partial\theta_i}\Sigma(\btheta)^{-1}
    -\frac{\partial}{\partial\theta_i}\{\Sigma(\btheta)^{-1}\}\X(\X^\top \Sigma(\btheta)^{-1}\X)^{-1} \X^\top\Sigma(\btheta)^{-1}\\
    &\quad -\Sigma(\btheta)^{-1}\X\{\frac{\partial}{\partial\theta_i}(\X^\top \Sigma(\btheta)^{-1}\X)^{-1} \}\X^\top\Sigma(\btheta)^{-1}\\
    &\quad -\Sigma(\btheta)^{-1}\X(\X^\top \Sigma(\btheta)^{-1}\X)^{-1} \X^\top\frac{\partial}{\partial\theta_i}\{\Sigma(\btheta)^{-1}\}.
\end{align*}
Let 
\begin{align*}
    \M_i&:=\frac{\partial}{\partial\theta_i}\Sigma(\btheta)^{-1} = -\Sigma(\btheta)^{-1}\Sigma_i \Sigma(\btheta)^{-1},\\
    \W&:= (\X^\top \Sigma(\btheta)^{-1}\X)^{-1}
\end{align*}
We have that
\begin{align*}
    \frac{\partial}{\partial\theta_i}(\X^\top \Sigma(\btheta)^{-1}\X)^{-1}  = -(\X^\top \Sigma(\btheta)^{-1}\X)^{-1}\X^\top \frac{\partial}{\partial\theta_i}\Sigma(\btheta)^{-1}\X(\X^\top \Sigma(\btheta)^{-1}\X)^{-1} = -\W \X^\top \M_i\X \W.
\end{align*}
Then,
\begin{align*}
    \frac{\partial P_\btheta}{\partial \theta_i}  
    %&= \M_i-\M_i\X\W \X^\top\Sigma(\btheta)^{-1}+\Sigma(\btheta)^{-1}\X(\W\X^\top\M_i\X\W)\X^\top \Sigma(\btheta)^{-1} -\Sigma(\btheta)^{-1}\X\W\X^\top\M_i\\
    &=\M_i-\M_i(\X\W \X^\top\Sigma(\btheta)^{-1})+(\X\W\X^\top \Sigma(\btheta)^{-1})^\top \M_i(\X\W\X^\top \Sigma(\btheta)^{-1}) -(\X\W\X^\top \Sigma(\btheta)^{-1})^\top\M_i\\
    &=(\mathbf{I}_n-\X\W \X^\top\Sigma(\btheta)^{-1})^\top\M_i(\mathbf{I}_n-\X\W\X^\top \Sigma(\btheta)^{-1})\\
    &=-(\mathbf{I}_n-\Sigma(\btheta)^{-1}\X\W \X^\top)\Sigma(\btheta)^{-1}\Sigma_i \Sigma(\btheta)^{-1}(\mathbf{I}_n-\X\W\X^\top \Sigma(\btheta)^{-1})\\
    &= -P_\btheta \Sigma_i P_\btheta
\end{align*}
since $P_\btheta = \Sigma(\btheta)^{-1}-\Sigma(\btheta)^{-1}\X\W \X^\top\Sigma(\btheta)^{-1}$.
\end{proof}

\end{proof}

\subsection{Additional Simulation Results}
\label{sub:appendix_simulation}

Figure~\ref{fig:simulation_covariance} and Tables~\ref{tab:ISE_ML}-\ref{tab:MSPE_REML} provide additional simulation results.

\begin{figure}[h!]
    \centering
    \includegraphics[width=0.8\textwidth]{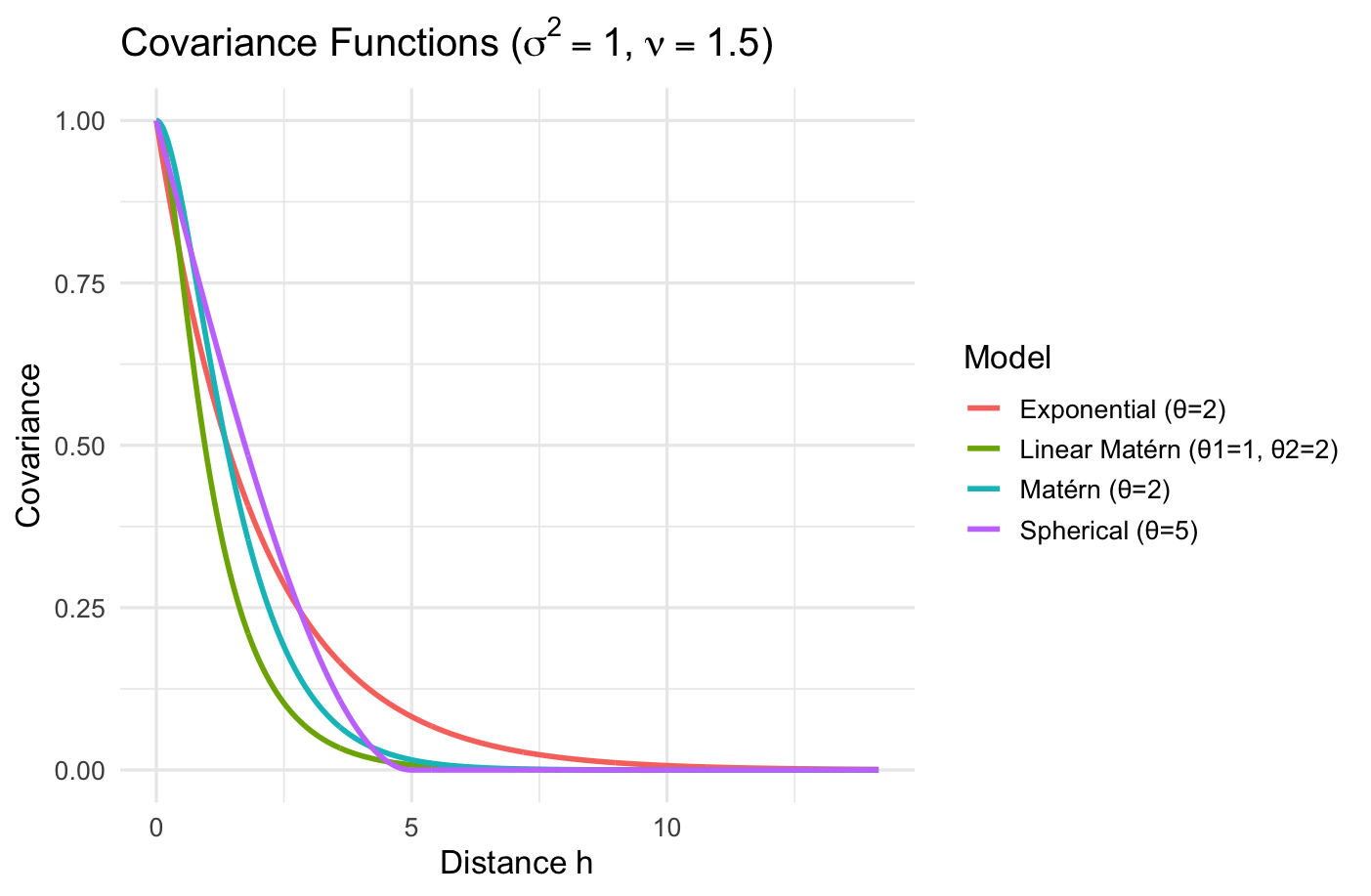}
    \caption{Covariance functions used in simulating data.}
    \label{fig:simulation_covariance}
\end{figure}

\begin{table}[ht]
\centering
\caption{Mean of ISEs, its standard error (in parentheses) (unit: $10^{-2}$) and percentage of failures of four parametric methods (exponential, Mat\~ern, spherical and Gaussian) with initial value for $\sigma_e = 10^{-1}$ and two nonparametric methods (B-splines method with order $p=3$ and number of knots $m=5$ and finite Gaussian mixtures method with number of support points $s=100$) using ML estimation}
\label{tab:ISE_ML}
\vspace{-0.5em}
\scriptsize
\par\vspace{0.8em}
\noindent\text{(a) When $\sigma_e = 0$}\par\smallskip
\begin{tabular}{ccccc}
  \hline
Method & Exponential & Matérn & Linear Matérn & Spherical \\ 
  \hline
\makecell{Exponential ML} & \textbf{\makecell{15.96 \\[0.6ex] (2.48); 0\%}} & \makecell{326.44 \\[0.6ex] (31.43); 0\%} & \makecell{36.77 \\[0.6ex] (6.06); 0\%} & \makecell{69.18 \\[0.6ex] (12.19); 0\%} \\ 
  \makecell{Matérn ML} & \makecell{16.16 \\[0.6ex] (1.46); 0\%} & \makecell{10.30 \\[0.6ex] (1.15); 0\%} & \makecell{14.83 \\[0.6ex] (1.02); 0\%} & \makecell{25.92 \\[0.6ex] (3.81); 0\%} \\ 
  \makecell{Matérn 1.5 ML} & \makecell{18.66 \\[0.6ex] (1.31); 0\%} & \textbf{\makecell{8.11 \\[0.6ex] (1.26); 0\%}} & \makecell{13.48 \\[0.6ex] (0.82); 0\%} & \textbf{\makecell{19.32 \\[0.6ex] (2.20); 0\%}} \\ 
  \makecell{Spherical ML} & \makecell{78.17 \\[0.6ex] (0.28); 0\%} & \makecell{78.88 \\[0.6ex] (0.32); 0\%} & \makecell{78.15 \\[0.6ex] (0.24); 0\%} & \makecell{96.03 \\[0.6ex] (0.35); 0\%} \\ 
  \makecell{Gaussian ML} & \makecell{32.21 \\[0.6ex] (1.26); 0\%} & \makecell{27.19 \\[0.6ex] (1.04); 0\%} & \makecell{29.38 \\[0.6ex] (0.90); 0\%} & \makecell{34.46 \\[0.6ex] (1.79); 0\%} \\ 
  \makecell{Bspl ML} & \makecell{25.67 \\[0.6ex] (5.76); 0\%} & \makecell{13.76 \\[0.6ex] (2.12); 0\%} & \makecell{13.10 \\[0.6ex] (1.01); 0\%} & \makecell{76.96 \\[0.6ex] (28.46); 0\%} \\ 
  \makecell{FGM ML} & \makecell{15.97 \\[0.6ex] (1.44); 0\%} & \makecell{11.67 \\[0.6ex] (1.41); 0\%} & \textbf{\makecell{12.42 \\[0.6ex] (0.94); 0\%}} & \makecell{22.62 \\[0.6ex] (3.07); 0\%} \\ 
   \hline
\end{tabular}
\par\vspace{0.8em}
\noindent\text{(c) When $\sigma_e = 2$}\par\smallskip
\begin{tabular}{ccccc}
  \hline
Method & Exponential & Matérn & Linear Matérn & Spherical \\ 
  \hline
\makecell{Exponential ML} & \textbf{\makecell{32.74 \\[0.6ex] (3.62); 0\%}} & \textbf{\makecell{26.48 \\[0.6ex] (3.57); 0\%}} & \textbf{\makecell{34.41 \\[0.6ex] (3.71); 0\%}} & \textbf{\makecell{37.34 \\[0.6ex] (5.55); 0\%}} \\ 
  \makecell{Matérn ML} & \makecell{57.19 \\[0.6ex] (5.62); 35\%} & \makecell{40.33 \\[0.6ex] (4.93); 43\%} & \makecell{56.78 \\[0.6ex] (7.35); 40\%} & \makecell{57.05 \\[0.6ex] (6.34); 39\%} \\ 
  \makecell{Matérn 1.5 ML} & \makecell{44.66 \\[0.6ex] (4.61); 0\%} & \makecell{28.59 \\[0.6ex] (3.90); 0\%} & \makecell{40.50 \\[0.6ex] (4.51); 0\%} & \makecell{52.47 \\[0.6ex] (7.01); 0\%} \\ 
  \makecell{Spherical ML} & \makecell{171.68 \\[0.6ex] (2.53); 0\%} & \makecell{169.46 \\[0.6ex] (2.56); 0\%} & \makecell{171.25 \\[0.6ex] (2.37); 0\%} & \makecell{188.11 \\[0.6ex] (2.77); 0\%} \\ 
  \makecell{Gaussian ML} & \makecell{48.05 \\[0.6ex] (7.42); 17\%} & \makecell{27.83 \\[0.6ex] (4.14); 8\%} & \makecell{37.50 \\[0.6ex] (6.45); 3\%} & \makecell{42.89 \\[0.6ex] (7.61); 26\%} \\ 
  \makecell{Bspl ML} & \makecell{1651.00 \\[0.6ex] (34.06); 0\%} & \makecell{1532.12 \\[0.6ex] (32.35); 0\%} & \makecell{1579.60 \\[0.6ex] (32.18); 0\%} & \makecell{1610.36 \\[0.6ex] (34.59); 0\%} \\ 
  \makecell{FGM ML} & \makecell{98.99 \\[0.6ex] (11.10); 0\%} & \makecell{90.37 \\[0.6ex] (11.18); 0\%} & \makecell{95.00 \\[0.6ex] (11.29); 0\%} & \makecell{101.25 \\[0.6ex] (11.64); 0\%} \\ 
   \hline
\end{tabular}
\end{table}

\begin{table}[ht]
\centering
\caption{Mean of MSPEs, its standard error (in parentheses) (unit: $10^{-2}$) and percentage of failures of true covariance model, CAR-INLA model, and four parametric methods (exponential, Mat\~ern, spherical and Gaussian) with initial value for $\sigma_e = 10^{-1}$ and two nonparametric methods (B-splines method with order $p=3$ and number of knots $m=5$ and finite Gaussian mixtures method with number of support points $s=100$) using ML estimation}
\label{tab:MSPE_ML}
\vspace{-0.5em}
\scriptsize
\par\vspace{0.8em}
\noindent\text{(a) When $\sigma_e = 0$}\par\smallskip
\begin{tabular}{ccccc}
  \hline
Method & Exponential & Matérn & Linear Matérn & Spherical \\ 
  \hline
\makecell{True} & \makecell{21.98 \\[0.6ex] (0.38); 0\%} & \makecell{5.54 \\[0.6ex] (0.12); 0\%} & \makecell{18.21 \\[0.6ex] (0.37); 0\%} & \makecell{13.32 \\[0.6ex] (0.23); 0\%} \\ 
%  \makecell{True Int} & \makecell{21.99 \\[0.6ex] (0.38); 0\%} & \makecell{5.54 \\[0.6ex] (0.12); 0\%} & \makecell{18.21 \\[0.6ex] (0.37); 0\%} & \makecell{13.33 \\[0.6ex] (0.23); 0\%} \\ 
   \hline
\makecell{CAR-INLA} & \makecell{24.65 \\[0.6ex] (0.43); 0\%} & \makecell{12.51 \\[0.6ex] (0.26); 0\%} & \makecell{26.95 \\[0.6ex] (0.50); 0\%} & \makecell{15.78 \\[0.6ex] (0.29); 0\%} \\ 
  \makecell{Exponential ML} & \textbf{\makecell{21.99 \\[0.6ex] (0.38); 0\%}} & \makecell{7.11 \\[0.6ex] (0.15); 0\%} & \makecell{19.55 \\[0.6ex] (0.37); 0\%} & \textbf{\makecell{13.54 \\[0.6ex] (0.24); 0\%}} \\ 
  \makecell{Matérn ML} & \makecell{22.15 \\[0.6ex] (0.39); 1\%} & \makecell{5.56 \\[0.6ex] (0.12); 0\%} & \makecell{17.75 \\[0.6ex] (0.35); 0\%} & \makecell{13.61 \\[0.6ex] (0.24); 1\%} \\ 
  \makecell{Matérn 1.5 ML} & \makecell{22.10 \\[0.6ex] (0.38); 0\%} & \textbf{\makecell{5.53 \\[0.6ex] (0.12); 0\%}} & \textbf{\makecell{17.69 \\[0.6ex] (0.35); 0\%}} & \makecell{13.61 \\[0.6ex] (0.24); 0\%} \\ 
  \makecell{Spherical ML} & \makecell{86.02 \\[0.6ex] (2.55); 0\%} & \makecell{91.53 \\[0.6ex] (2.75); 0\%} & \makecell{95.07 \\[0.6ex] (2.33); 0\%} & \makecell{88.84 \\[0.6ex] (3.08); 0\%} \\ 
  \makecell{Gaussian ML} & \makecell{22.92 \\[0.6ex] (0.38); 0\%} & \makecell{6.03 \\[0.6ex] (0.13); 0\%} & \makecell{18.48 \\[0.6ex] (0.38); 0\%} & \makecell{14.36 \\[0.6ex] (0.25); 0\%} \\ 
  \makecell{Bspl ML} & \makecell{24.09 \\[0.6ex] (0.43); 0\%} & \makecell{5.61 \\[0.6ex] (0.13); 0\%} & \makecell{17.92 \\[0.6ex] (0.36); 0\%} & \makecell{14.79 \\[0.6ex] (0.26); 0\%} \\ 
  \makecell{FGM ML} & \makecell{22.15 \\[0.6ex] (0.39); 0\%} & \makecell{5.61 \\[0.6ex] (0.12); 0\%} & \makecell{17.85 \\[0.6ex] (0.36); 0\%} & \makecell{13.61 \\[0.6ex] (0.24); 0\%} \\ 
   \hline
\end{tabular}
\par\vspace{0.8em}
\noindent\text{(c) When $\sigma_e = 2$}\par\smallskip
\begin{tabular}{ccccc}
  \hline
Method & Exponential & Matérn & Linear Matérn & Spherical \\ 
  \hline
\makecell{True} & \makecell{444.55 \\[0.6ex] (7.18); 0\%} & \makecell{438.04 \\[0.6ex] (7.08); 0\%} & \makecell{454.09 \\[0.6ex] (7.34); 0\%} & \makecell{433.86 \\[0.6ex] (6.98); 0\%} \\ 
%  \makecell{True Int} & \makecell{444.58 \\[0.6ex] (7.18); 0\%} & \makecell{438.03 \\[0.6ex] (7.09); 0\%} & \makecell{454.16 \\[0.6ex] (7.34); 0\%} & \makecell{433.84 \\[0.6ex] (6.98); 0\%} \\ 
   \hline
\makecell{CAR-INLA} & \makecell{455.33 \\[0.6ex] (7.94); 0\%} & \makecell{449.26 \\[0.6ex] (7.91); 0\%} & \makecell{470.16 \\[0.6ex] (8.32); 0\%} & \makecell{442.95 \\[0.6ex] (7.59); 0\%} \\ 
  \makecell{Exponential ML} & \makecell{447.74 \\[0.6ex] (7.31); 0\%} & \makecell{440.45 \\[0.6ex] (7.21); 0\%} & \makecell{455.71 \\[0.6ex] (7.47); 0\%} & \makecell{437.44 \\[0.6ex] (7.11); 0\%} \\ 
  \makecell{Matérn ML} & \textbf{\makecell{445.16 \\[0.6ex] (9.43); 42\%}} & \textbf{\makecell{437.00 \\[0.6ex] (7.59); 46\%}} & \textbf{\makecell{451.19 \\[0.6ex] (8.58); 42\%}} & \makecell{438.39 \\[0.6ex] (9.38); 52\%} \\ 
  \makecell{Matérn 1.5 ML} & \makecell{451.16 \\[0.6ex] (7.46); 0\%} & \makecell{443.67 \\[0.6ex] (7.18); 0\%} & \makecell{460.61 \\[0.6ex] (7.63); 0\%} & \makecell{441.98 \\[0.6ex] (7.56); 0\%} \\ 
  \makecell{Spherical ML} & \makecell{482.51 \\[0.6ex] (8.23); 0\%} & \makecell{488.35 \\[0.6ex] (8.48); 0\%} & \makecell{492.03 \\[0.6ex] (8.34); 0\%} & \makecell{485.16 \\[0.6ex] (8.51); 0\%} \\ 
  \makecell{Gaussian ML} & \makecell{449.19 \\[0.6ex] (8.19); 13\%} & \makecell{442.02 \\[0.6ex] (7.67); 6\%} & \makecell{458.16 \\[0.6ex] (7.67); 3\%} & \makecell{439.91 \\[0.6ex] (7.89); 11\%} \\ 
  \makecell{Bspl ML} & \makecell{631.02 \\[0.6ex] (10.83); 0\%} & \makecell{615.80 \\[0.6ex] (10.62); 0\%} & \makecell{625.94 \\[0.6ex] (10.76); 0\%} & \makecell{623.01 \\[0.6ex] (10.71); 0\%} \\ 
  \makecell{FGM ML} & \makecell{447.46 \\[0.6ex] (7.29); 0\%} & \makecell{440.70 \\[0.6ex] (7.19); 0\%} & \makecell{456.25 \\[0.6ex] (7.45); 0\%} & \textbf{\makecell{437.42 \\[0.6ex] (7.11); 0\%}} \\ 
   \hline
\end{tabular}
\end{table}

\begin{table}[ht]
\centering
\caption{Mean of ISEs, its standard error (in parentheses) (unit: $10^{-2}$) and percentage of failures of four parametric methods (exponential, Mat\~ern, spherical and Gaussian) with initial value for $\sigma_e = 10^{-1}$ and two nonparametric methods (B-splines method with order $p=3$ and number of knots $m=5$ and finite Gaussian mixtures method with number of support points $s=100$) using WLS estimation}
\label{tab:ISE_WLS}
\vspace{-0.5em}
\scriptsize
\par\vspace{0.8em}
\noindent\text{(a) When $\sigma_e = 0$}\par\smallskip
\begin{tabular}{ccccc}
  \hline
Method & Exponential & Matérn & Linear Matérn & Spherical \\ 
  \hline
\makecell{Exponential WLS} & \makecell{1414821319.60 \\[0.6ex] (690838356.23); 0\%} & \makecell{1900242750.54 \\[0.6ex] (1060831785.93); 0\%} & \makecell{1182889825.34 \\[0.6ex] (720422087.61); 0\%} & \makecell{1689303448.50 \\[0.6ex] (791120814.56); 0\%} \\ 
  \makecell{Matérn WLS} & \makecell{41455431.10 \\[0.6ex] (31836837.72); 15\%} & \makecell{125526806.67 \\[0.6ex] (122440227.27); 9\%} & \makecell{1172025.39 \\[0.6ex] (978574.42); 5\%} & \makecell{58884184.79 \\[0.6ex] (53459566.58); 20\%} \\ 
  \makecell{Matérn 1.5 WLS} & \makecell{216681920.26 \\[0.6ex] (90099378.21); 0\%} & \makecell{204972868.65 \\[0.6ex] (114006813.51); 0\%} & \makecell{127712850.16 \\[0.6ex] (65888755.83); 0\%} & \makecell{257095876.29 \\[0.6ex] (159663336.76); 0\%} \\ 
  \makecell{Spherical WLS} & \makecell{78.29 \\[0.6ex] (0.26); 0\%} & \makecell{79.09 \\[0.6ex] (0.29); 0\%} & \makecell{78.35 \\[0.6ex] (0.21); 0\%} & \makecell{96.21 \\[0.6ex] (0.32); 0\%} \\ 
  \makecell{Gaussian WLS} & \makecell{3763662.89 \\[0.6ex] (1653890.97); 0\%} & \makecell{3466330.44 \\[0.6ex] (2008282.10); 0\%} & \makecell{1783373.10 \\[0.6ex] (1206924.58); 0\%} & \makecell{5209052.50 \\[0.6ex] (2592122.42); 0\%} \\ 
  \makecell{Bspl WLS} & \textbf{\makecell{22.36 \\[0.6ex] (1.41); 0\%}} & \textbf{\makecell{15.87 \\[0.6ex] (1.27); 0\%}} & \textbf{\makecell{21.04 \\[0.6ex] (1.28); 0\%}} & \textbf{\makecell{24.31 \\[0.6ex] (1.98); 0\%}} \\ 
  \makecell{Bspl WLS Nugget} & \makecell{25.28 \\[0.6ex] (1.59); 0\%} & \makecell{17.56 \\[0.6ex] (1.41); 0\%} & \makecell{26.70 \\[0.6ex] (1.85); 0\%} & \makecell{26.02 \\[0.6ex] (2.12); 0\%} \\ 
  \makecell{FGM WLS} & \makecell{24.23 \\[0.6ex] (1.58); 0\%} & \makecell{17.14 \\[0.6ex] (1.33); 0\%} & \makecell{23.70 \\[0.6ex] (1.39); 0\%} & \makecell{24.49 \\[0.6ex] (2.00); 0\%} \\ 
   \hline
\end{tabular}
\par\vspace{0.8em}
\noindent\text{(b) When $\sigma_e = 1$}\par\smallskip
\begin{tabular}{ccccc}
  \hline
Method & Exponential & Matérn & Linear Matérn & Spherical \\ 
  \hline
\makecell{Exponential WLS} & \makecell{8567567935.30 \\[0.6ex] (7286526641.39); 0\%} & \makecell{1853461775.42 \\[0.6ex] (929187517.31); 0\%} & \makecell{1012324224.85 \\[0.6ex] (440114498.04); 0\%} & \makecell{8064674281.42 \\[0.6ex] (6890540694.87); 0\%} \\ 
  \makecell{Matérn WLS} & \makecell{5726348.02 \\[0.6ex] (3352044.88); 21\%} & \makecell{15527631.35 \\[0.6ex] (10825285.05); 9\%} & \makecell{21004938.22 \\[0.6ex] (16038765.59); 2\%} & \makecell{25941755.49 \\[0.6ex] (21596596.96); 16\%} \\ 
  \makecell{Matérn 1.5 WLS} & \makecell{305228511.64 \\[0.6ex] (151126191.22); 0\%} & \makecell{263199664.72 \\[0.6ex] (139555889.46); 0\%} & \makecell{200588768.96 \\[0.6ex] (122362106.88); 0\%} & \makecell{175932437.91 \\[0.6ex] (87701447.22); 0\%} \\ 
  \makecell{Spherical WLS} & \makecell{72.79 \\[0.6ex] (0.04); 0\%} & \makecell{73.13 \\[0.6ex] (0.06); 0\%} & \makecell{72.99 \\[0.6ex] (0.05); 0\%} & \makecell{90.29 \\[0.6ex] (0.05); 0\%} \\ 
  \makecell{Gaussian WLS} & \makecell{3495568.57 \\[0.6ex] (1869281.64); 0\%} & \makecell{3679565.57 \\[0.6ex] (2498354.84); 0\%} & \makecell{2502749.75 \\[0.6ex] (1639547.93); 0\%} & \makecell{3352681.51 \\[0.6ex] (1952862.97); 0\%} \\ 
  \makecell{Bspl WLS} & \makecell{40.06 \\[0.6ex] (1.28); 0\%} & \makecell{35.17 \\[0.6ex] (1.44); 0\%} & \makecell{41.26 \\[0.6ex] (0.94); 0\%} & \makecell{41.94 \\[0.6ex] (2.09); 0\%} \\ 
  \makecell{Bspl WLS Nugget} & \textbf{\makecell{26.23 \\[0.6ex] (1.65); 0\%}} & \textbf{\makecell{20.34 \\[0.6ex] (1.78); 0\%}} & \makecell{30.15 \\[0.6ex] (2.32); 0\%} & \makecell{27.41 \\[0.6ex] (2.25); 0\%} \\ 
  \makecell{FGM WLS} & \makecell{27.13 \\[0.6ex] (1.60); 0\%} & \makecell{20.37 \\[0.6ex] (1.56); 0\%} & \textbf{\makecell{29.94 \\[0.6ex] (1.70); 0\%}} & \textbf{\makecell{26.24 \\[0.6ex] (2.09); 0\%}} \\ 
   \hline
\end{tabular}
\par\vspace{0.8em}
\noindent\text{(c) When $\sigma_e = 2$}\par\smallskip
\begin{tabular}{ccccc}
  \hline
Method & Exponential & Matérn & Linear Matérn & Spherical \\ 
  \hline
\makecell{Exponential WLS} & \makecell{6297132322.78 \\[0.6ex] (2612061789.68); 0\%} & \makecell{54382863733.70 \\[0.6ex] (49885716704.61); 0\%} & \makecell{14193568401.23 \\[0.6ex] (11214683451.75); 0\%} & \makecell{4243318619.49 \\[0.6ex] (1549783877.31); 0\%} \\ 
  \makecell{Matérn WLS} & \makecell{9365579.32 \\[0.6ex] (4527320.90); 16\%} & \makecell{157634553.62 \\[0.6ex] (122996244.64); 8\%} & \makecell{26057727.69 \\[0.6ex] (11980667.69); 5\%} & \makecell{259547275.51 \\[0.6ex] (157820294.96); 15\%} \\ 
  \makecell{Matérn 1.5 WLS} & \makecell{360365718.88 \\[0.6ex] (206124815.18); 0\%} & \makecell{382851345.46 \\[0.6ex] (218093365.55); 0\%} & \makecell{185622753.10 \\[0.6ex] (88621767.73); 0\%} & \makecell{532679719.69 \\[0.6ex] (368471491.35); 0\%} \\ 
  \makecell{Spherical WLS} & \makecell{89.50 \\[0.6ex] (0.53); 0\%} & \makecell{88.01 \\[0.6ex] (0.52); 0\%} & \makecell{88.35 \\[0.6ex] (0.48); 0\%} & \makecell{106.09 \\[0.6ex] (0.57); 0\%} \\ 
  \makecell{Gaussian WLS} & \makecell{1130500.50 \\[0.6ex] (1130379.77); 0\%} & \makecell{2873996.98 \\[0.6ex] (2521308.24); 0\%} & \makecell{1357100.79 \\[0.6ex] (1126498.87); 0\%} & \makecell{2617620.04 \\[0.6ex] (2617490.75); 0\%} \\ 
  \makecell{Bspl WLS} & \makecell{327.62 \\[0.6ex] (8.82); 0\%} & \makecell{335.73 \\[0.6ex] (9.52); 0\%} & \makecell{321.35 \\[0.6ex] (8.67); 0\%} & \makecell{345.91 \\[0.6ex] (9.52); 0\%} \\ 
  \makecell{Bspl WLS Nugget} & \textbf{\makecell{32.60 \\[0.6ex] (2.41); 0\%}} & \textbf{\makecell{30.24 \\[0.6ex] (3.77); 0\%}} & \textbf{\makecell{40.27 \\[0.6ex] (3.90); 0\%}} & \textbf{\makecell{32.54 \\[0.6ex] (2.74); 0\%}} \\ 
  \makecell{FGM WLS} & \makecell{104.83 \\[0.6ex] (13.48); 0\%} & \makecell{107.12 \\[0.6ex] (14.59); 0\%} & \makecell{172.07 \\[0.6ex] (17.22); 0\%} & \makecell{62.07 \\[0.6ex] (8.96); 0\%} \\ 
   \hline
\end{tabular}
\end{table}

\begin{table}[ht]
\centering
\caption{Mean of ISEs, its standard error (in parentheses) (unit: $10^{-2}$) and percentage of failures of four parametric methods (exponential, Mat\~ern, spherical and Gaussian) with initial value for $\sigma_e = 10^{-1}$ and nonparametric method (finite Gaussian mixtures method with number of support points $s=100$) using REML estimation}
\label{tab:ISE_REML}
\vspace{-0.5em}
\scriptsize
\par\vspace{0.8em}
\noindent\text{(a) When $\sigma_e = 0$}\par\smallskip
\begin{tabular}{ccccc}
  \hline
Method & Exponential & Matérn & Linear Matérn & Spherical \\ 
  \hline
\makecell{Exponential REML} & \makecell{37.17 \\[0.6ex] (8.85); 0\%} & \makecell{4328.83 \\[0.6ex] (1108.96); 1\%} & \makecell{116.94 \\[0.6ex] (21.08); 0\%} & \makecell{442.83 \\[0.6ex] (99.09); 0\%} \\ 
  \makecell{Matérn REML} & \makecell{23.24 \\[0.6ex] (3.54); 0\%} & \makecell{11.88 \\[0.6ex] (1.62); 0\%} & \makecell{13.62 \\[0.6ex] (1.00); 0\%} & \makecell{67.59 \\[0.6ex] (13.46); 0\%} \\ 
  \makecell{Matérn 1.5 REML} & \textbf{\makecell{17.42 \\[0.6ex] (1.32); 0\%}} & \textbf{\makecell{10.08 \\[0.6ex] (1.74); 0\%}} & \textbf{\makecell{12.37 \\[0.6ex] (0.78); 0\%}} & \textbf{\makecell{21.51 \\[0.6ex] (3.23); 0\%}} \\ 
  \makecell{Spherical REML} & \makecell{78.14 \\[0.6ex] (0.28); 0\%} & \makecell{78.85 \\[0.6ex] (0.32); 0\%} & \makecell{78.12 \\[0.6ex] (0.24); 0\%} & \makecell{96.00 \\[0.6ex] (0.35); 0\%} \\ 
  \makecell{Gaussian REML} & \makecell{31.18 \\[0.6ex] (1.26); 0\%} & \makecell{26.49 \\[0.6ex] (1.04); 0\%} & \makecell{28.70 \\[0.6ex] (0.90); 0\%} & \makecell{33.06 \\[0.6ex] (1.79); 0\%} \\ 
  \makecell{FGM REML} & \makecell{29.89 \\[0.6ex] (5.25); 0\%} & \makecell{21.62 \\[0.6ex] (3.34); 0\%} & \makecell{12.65 \\[0.6ex] (1.03); 0\%} & \makecell{40.62 \\[0.6ex] (6.71); 0\%} \\ 
   \hline
\end{tabular}
\par\vspace{0.8em}
\noindent\text{(b) When $\sigma_e = 1$}\par\smallskip
\begin{tabular}{ccccc}
  \hline
Method & Exponential & Matérn & Linear Matérn & Spherical \\ 
  \hline
\makecell{Exponential REML} & \makecell{50.47 \\[0.6ex] (12.43); 0\%} & \makecell{46.01 \\[0.6ex] (10.38); 0\%} & \makecell{16.07 \\[0.6ex] (1.87); 0\%} & \makecell{176.39 \\[0.6ex] (40.70); 0\%} \\ 
  \makecell{Matérn REML} & \makecell{885.68 \\[0.6ex] (436.29); 13\%} & \makecell{22.78 \\[0.6ex] (4.89); 18\%} & \makecell{14.63 \\[0.6ex] (1.32); 18\%} & \makecell{2309.97 \\[0.6ex] (1325.08); 20\%} \\ 
  \makecell{Matérn 1.5 REML} & \makecell{22.56 \\[0.6ex] (3.37); 0\%} & \makecell{20.29 \\[0.6ex] (3.92); 0\%} & \textbf{\makecell{14.06 \\[0.6ex] (1.17); 0\%}} & \makecell{56.68 \\[0.6ex] (12.64); 0\%} \\ 
  \makecell{Spherical REML} & \makecell{73.51 \\[0.6ex] (0.18); 0\%} & \makecell{73.48 \\[0.6ex] (0.17); 0\%} & \makecell{73.39 \\[0.6ex] (0.13); 0\%} & \makecell{91.01 \\[0.6ex] (0.23); 0\%} \\ 
  \makecell{Gaussian REML} & \textbf{\makecell{22.33 \\[0.6ex] (1.49); 7\%}} & \textbf{\makecell{13.44 \\[0.6ex] (1.26); 15\%}} & \makecell{17.95 \\[0.6ex] (1.14); 0\%} & \textbf{\makecell{21.69 \\[0.6ex] (2.04); 22\%}} \\ 
  \makecell{FGM REML} & \makecell{45.12 \\[0.6ex] (7.89); 0\%} & \makecell{35.42 \\[0.6ex] (5.11); 0\%} & \makecell{22.43 \\[0.6ex] (1.92); 0\%} & \makecell{56.26 \\[0.6ex] (8.28); 0\%} \\ 
   \hline
\end{tabular}
\par\vspace{0.8em}
\noindent\text{(c) When $\sigma_e = 2$}\par\smallskip
\begin{tabular}{ccccc}
  \hline
Method & Exponential & Matérn & Linear Matérn & Spherical \\ 
  \hline
\makecell{Exponential REML} & \makecell{308386.57 \\[0.6ex] (179965.71); 0\%} & \makecell{56.31 \\[0.6ex] (12.39); 0\%} & \makecell{36.47 \\[0.6ex] (4.56); 0\%} & \makecell{73498.54 \\[0.6ex] (64216.55); 0\%} \\ 
  \makecell{Matérn REML} & \makecell{1766.80 \\[0.6ex] (1284.69); 25\%} & \makecell{333.49 \\[0.6ex] (210.50); 39\%} & \makecell{52.79 \\[0.6ex] (5.89); 35\%} & \makecell{1186.38 \\[0.6ex] (587.07); 40\%} \\ 
  \makecell{Matérn 1.5 REML} & \makecell{58.86 \\[0.6ex] (8.78); 0\%} & \makecell{42.03 \\[0.6ex] (7.05); 0\%} & \makecell{37.75 \\[0.6ex] (4.43); 0\%} & \makecell{83.42 \\[0.6ex] (17.17); 0\%} \\ 
  \makecell{Spherical REML} & \makecell{172.43 \\[0.6ex] (2.55); 0\%} & \makecell{170.20 \\[0.6ex] (2.58); 0\%} & \makecell{172.01 \\[0.6ex] (2.38); 0\%} & \makecell{188.86 \\[0.6ex] (2.79); 0\%} \\ 
  \makecell{Gaussian REML} & \textbf{\makecell{45.59 \\[0.6ex] (7.23); 19\%}} & \textbf{\makecell{26.60 \\[0.6ex] (3.94); 13\%}} & \textbf{\makecell{35.91 \\[0.6ex] (6.68); 8\%}} & \textbf{\makecell{51.10 \\[0.6ex] (9.04); 23\%}} \\ 
  \makecell{FGM REML} & \makecell{129.64 \\[0.6ex] (15.78); 0\%} & \makecell{120.84 \\[0.6ex] (14.00); 0\%} & \makecell{104.46 \\[0.6ex] (12.17); 0\%} & \makecell{147.17 \\[0.6ex] (16.85); 0\%} \\ 
   \hline
\end{tabular}
\end{table}

\begin{table}[ht]
\centering
\caption{Mean of MSPEs, its standard error (in parentheses) (unit: $10^{-2}$) and percentage of failures of four parametric methods (exponential, Mat\~ern, spherical and Gaussian) with initial value for $\sigma_e = 10^{-1}$ and two nonparametric methods (B-splines method with order $p=3$ and number of knots $m=5$ and finite Gaussian mixtures method with number of support points $s=100$) using WLS estimation}
\label{tab:MSPE_WLS}
\vspace{-0.5em}
\scriptsize
\par\vspace{0.8em}
\noindent\text{(a) When $\sigma_e = 0$}\par\smallskip
\begin{tabular}{ccccc}
  \hline
Method & Exponential & Matérn & Linear Matérn & Spherical \\ 
  \hline
\makecell{Exponential WLS} & \makecell{27.84 \\[0.6ex] (1.38); 1\%} & \makecell{11.85 \\[0.6ex] (1.17); 0\%} & \makecell{27.78 \\[0.6ex] (1.45); 0\%} & \makecell{17.21 \\[0.6ex] (1.15); 0\%} \\ 
  \makecell{Matérn WLS} & \makecell{32.90 \\[0.6ex] (1.35); 4\%} & \makecell{21.01 \\[0.6ex] (2.09); 3\%} & \makecell{33.90 \\[0.6ex] (1.94); 3\%} & \makecell{20.78 \\[0.6ex] (1.22); 2\%} \\ 
  \makecell{Matérn 1.5 WLS} & \makecell{34.50 \\[0.6ex] (1.88); 0\%} & \makecell{20.42 \\[0.6ex] (2.49); 0\%} & \makecell{34.40 \\[0.6ex] (2.38); 0\%} & \makecell{22.81 \\[0.6ex] (1.81); 0\%} \\ 
  \makecell{Spherical WLS} & \makecell{86.02 \\[0.6ex] (2.55); 0\%} & \makecell{91.53 \\[0.6ex] (2.75); 0\%} & \makecell{95.07 \\[0.6ex] (2.33); 0\%} & \makecell{88.84 \\[0.6ex] (3.08); 0\%} \\ 
  \makecell{Gaussian WLS} & \makecell{38.40 \\[0.6ex] (1.72); 0\%} & \makecell{21.53 \\[0.6ex] (2.13); 0\%} & \makecell{35.20 \\[0.6ex] (2.04); 0\%} & \makecell{96.80 \\[0.6ex] (69.55); 0\%} \\ 
  \makecell{Bspl WLS} & \makecell{340.89 \\[0.6ex] (134.06); 0\%} & \makecell{46.56 \\[0.6ex] (17.77); 0\%} & \makecell{67.84 \\[0.6ex] (49.16); 0\%} & \makecell{422.67 \\[0.6ex] (116.04); 0\%} \\ 
  \makecell{Bspl WLS Nugget} & \makecell{135.35 \\[0.6ex] (97.84); 0\%} & \makecell{24.13 \\[0.6ex] (13.01); 0\%} & \textbf{\makecell{24.13 \\[0.6ex] (0.64); 0\%}} & \makecell{173.43 \\[0.6ex] (78.71); 0\%} \\ 
  \makecell{FGM WLS} & \textbf{\makecell{26.42 \\[0.6ex] (0.57); 0\%}} & \textbf{\makecell{10.16 \\[0.6ex] (0.37); 0\%}} & \makecell{25.67 \\[0.6ex] (0.69); 0\%} & \textbf{\makecell{16.06 \\[0.6ex] (0.34); 0\%}} \\ 
   \hline
\end{tabular}
\par\vspace{0.8em}
\noindent\text{(b) When $\sigma_e = 1$}\par\smallskip
\begin{tabular}{ccccc}
  \hline
Method & Exponential & Matérn & Linear Matérn & Spherical \\ 
  \hline
\makecell{Exponential WLS} & \makecell{136.37 \\[0.6ex] (2.29); 0\%} & \makecell{126.92 \\[0.6ex] (2.18); 0\%} & \textbf{\makecell{142.67 \\[0.6ex] (2.61); 0\%}} & \makecell{127.50 \\[0.6ex] (2.11); 0\%} \\ 
  \makecell{Matérn WLS} & \makecell{140.06 \\[0.6ex] (2.41); 1\%} & \makecell{130.20 \\[0.6ex] (2.52); 1\%} & \makecell{147.52 \\[0.6ex] (3.00); 3\%} & \makecell{128.12 \\[0.6ex] (2.31); 4\%} \\ 
  \makecell{Matérn 1.5 WLS} & \makecell{140.17 \\[0.6ex] (2.52); 0\%} & \makecell{130.62 \\[0.6ex] (2.67); 0\%} & \makecell{146.74 \\[0.6ex] (2.93); 0\%} & \makecell{129.39 \\[0.6ex] (2.43); 0\%} \\ 
  \makecell{Spherical WLS} & \makecell{184.99 \\[0.6ex] (3.75); 0\%} & \makecell{190.62 \\[0.6ex] (4.02); 0\%} & \makecell{194.19 \\[0.6ex] (3.73); 0\%} & \makecell{187.73 \\[0.6ex] (4.19); 0\%} \\ 
  \makecell{Gaussian WLS} & \makecell{142.40 \\[0.6ex] (2.58); 0\%} & \makecell{131.67 \\[0.6ex] (2.75); 0\%} & \makecell{148.09 \\[0.6ex] (2.95); 0\%} & \makecell{130.73 \\[0.6ex] (2.52); 0\%} \\ 
  \makecell{Bspl WLS} & \makecell{179.40 \\[0.6ex] (3.08); 0\%} & \makecell{164.19 \\[0.6ex] (2.87); 0\%} & \makecell{174.73 \\[0.6ex] (3.07); 0\%} & \makecell{172.00 \\[0.6ex] (2.93); 0\%} \\ 
  \makecell{Bspl WLS Nugget} & \textbf{\makecell{135.60 \\[0.6ex] (2.27); 0\%}} & \textbf{\makecell{126.64 \\[0.6ex] (2.40); 0\%}} & \makecell{145.47 \\[0.6ex] (2.80); 0\%} & \textbf{\makecell{125.55 \\[0.6ex] (2.07); 0\%}} \\ 
  \makecell{FGM WLS} & \makecell{138.93 \\[0.6ex] (2.44); 0\%} & \makecell{130.65 \\[0.6ex] (3.02); 0\%} & \makecell{154.80 \\[0.6ex] (3.56); 0\%} & \makecell{125.64 \\[0.6ex] (2.07); 0\%} \\ 
   \hline
\end{tabular}
\par\vspace{0.8em}
\noindent\text{(c) When $\sigma_e = 2$}\par\smallskip
\begin{tabular}{ccccc}
  \hline
Method & Exponential & Matérn & Linear Matérn & Spherical \\ 
  \hline
\makecell{Exponential WLS} & \textbf{\makecell{447.90 \\[0.6ex] (7.50); 0\%}} & \textbf{\makecell{442.97 \\[0.6ex] (7.47); 0\%}} & \textbf{\makecell{458.97 \\[0.6ex] (7.80); 0\%}} & \makecell{438.62 \\[0.6ex] (7.28); 0\%} \\ 
  \makecell{Matérn WLS} & \makecell{455.05 \\[0.6ex] (7.54); 3\%} & \makecell{453.32 \\[0.6ex] (7.61); 1\%} & \makecell{468.97 \\[0.6ex] (8.17); 2\%} & \makecell{445.68 \\[0.6ex] (7.53); 4\%} \\ 
  \makecell{Matérn 1.5 WLS} & \makecell{455.72 \\[0.6ex] (7.48); 0\%} & \makecell{454.19 \\[0.6ex] (7.60); 0\%} & \makecell{468.89 \\[0.6ex] (8.13); 0\%} & \makecell{447.17 \\[0.6ex] (7.45); 0\%} \\ 
  \makecell{Spherical WLS} & \makecell{482.51 \\[0.6ex] (8.23); 0\%} & \makecell{488.35 \\[0.6ex] (8.48); 0\%} & \makecell{492.03 \\[0.6ex] (8.34); 0\%} & \makecell{485.16 \\[0.6ex] (8.51); 0\%} \\ 
  \makecell{Gaussian WLS} & \makecell{463.50 \\[0.6ex] (7.98); 0\%} & \makecell{459.06 \\[0.6ex] (8.03); 0\%} & \makecell{468.58 \\[0.6ex] (8.01); 0\%} & \makecell{454.76 \\[0.6ex] (7.85); 0\%} \\ 
  \makecell{Bspl WLS} & \makecell{631.53 \\[0.6ex] (10.80); 0\%} & \makecell{616.40 \\[0.6ex] (10.59); 0\%} & \makecell{626.12 \\[0.6ex] (10.75); 0\%} & \makecell{624.37 \\[0.6ex] (10.67); 0\%} \\ 
  \makecell{Bspl WLS Nugget} & \makecell{448.97 \\[0.6ex] (7.44); 0\%} & \makecell{444.88 \\[0.6ex] (7.45); 0\%} & \makecell{459.92 \\[0.6ex] (7.66); 0\%} & \textbf{\makecell{438.44 \\[0.6ex] (7.19); 0\%}} \\ 
  \makecell{FGM WLS} & \makecell{456.85 \\[0.6ex] (7.88); 0\%} & \makecell{459.36 \\[0.6ex] (8.30); 0\%} & \makecell{479.42 \\[0.6ex] (8.80); 0\%} & \makecell{440.42 \\[0.6ex] (7.30); 0\%} \\ 
   \hline
\end{tabular}
\end{table}

\begin{table}[ht]
\centering
\caption{Mean of MSPEs, its standard error (in parentheses) (unit: $10^{-2}$) and percentage of failures of four parametric methods (exponential, Mat\~ern, spherical and Gaussian) with initial value for $\sigma_e = 10^{-1}$ and nonparametric method (finite Gaussian mixtures method with number of support points $s=100$) using REML estimation}
\label{tab:MSPE_REML}
\vspace{-0.5em}
\scriptsize
\par\vspace{0.8em}
\noindent\text{(a) When $\sigma_e = 0$}\par\smallskip
\begin{tabular}{ccccc}
  \hline
Method & Exponential & Matérn & Linear Matérn & Spherical \\ 
  \hline
\makecell{Exponential REML} & \textbf{\makecell{21.99 \\[0.6ex] (0.38); 0\%}} & \makecell{7.11 \\[0.6ex] (0.15); 1\%} & \makecell{19.51 \\[0.6ex] (0.37); 0\%} & \textbf{\makecell{13.53 \\[0.6ex] (0.24); 0\%}} \\ 
  \makecell{Matérn REML} & \makecell{22.13 \\[0.6ex] (0.38); 1\%} & \makecell{5.56 \\[0.6ex] (0.12); 0\%} & \makecell{17.74 \\[0.6ex] (0.35); 0\%} & \makecell{13.60 \\[0.6ex] (0.24); 1\%} \\ 
  \makecell{Matérn 1.5 REML} & \makecell{22.11 \\[0.6ex] (0.38); 0\%} & \textbf{\makecell{5.53 \\[0.6ex] (0.12); 0\%}} & \textbf{\makecell{17.70 \\[0.6ex] (0.35); 0\%}} & \makecell{13.62 \\[0.6ex] (0.24); 0\%} \\ 
  \makecell{Spherical REML} & \makecell{86.02 \\[0.6ex] (2.55); 0\%} & \makecell{91.53 \\[0.6ex] (2.75); 0\%} & \makecell{95.07 \\[0.6ex] (2.33); 0\%} & \makecell{88.84 \\[0.6ex] (3.08); 0\%} \\ 
  \makecell{Gaussian REML} & \makecell{22.95 \\[0.6ex] (0.38); 0\%} & \makecell{6.04 \\[0.6ex] (0.13); 0\%} & \makecell{18.50 \\[0.6ex] (0.38); 0\%} & \makecell{14.39 \\[0.6ex] (0.25); 0\%} \\ 
  \makecell{FGM REML} & \makecell{22.15 \\[0.6ex] (0.39); 0\%} & \makecell{5.61 \\[0.6ex] (0.12); 0\%} & \makecell{17.85 \\[0.6ex] (0.35); 0\%} & \makecell{13.61 \\[0.6ex] (0.24); 0\%} \\ 
   \hline
\end{tabular}
\par\vspace{0.8em}
\noindent\text{(b) When $\sigma_e = 1$}\par\smallskip
\begin{tabular}{ccccc}
  \hline
Method & Exponential & Matérn & Linear Matérn & Spherical \\ 
  \hline
\makecell{Exponential REML} & \makecell{134.12 \\[0.6ex] (2.20); 0\%} & \makecell{123.40 \\[0.6ex] (2.04); 0\%} & \makecell{137.82 \\[0.6ex] (2.29); 0\%} & \textbf{\makecell{124.61 \\[0.6ex] (2.04); 0\%}} \\ 
  \makecell{Matérn REML} & \textbf{\makecell{133.76 \\[0.6ex] (2.17); 20\%}} & \textbf{\makecell{122.71 \\[0.6ex] (2.15); 26\%}} & \textbf{\makecell{137.17 \\[0.6ex] (2.31); 23\%}} & \makecell{125.29 \\[0.6ex] (2.36); 26\%} \\ 
  \makecell{Matérn 1.5 REML} & \makecell{134.35 \\[0.6ex] (2.20); 0\%} & \makecell{123.20 \\[0.6ex] (2.02); 0\%} & \makecell{137.87 \\[0.6ex] (2.27); 0\%} & \makecell{124.62 \\[0.6ex] (2.03); 0\%} \\ 
  \makecell{Spherical REML} & \makecell{184.99 \\[0.6ex] (3.75); 0\%} & \makecell{190.62 \\[0.6ex] (4.02); 0\%} & \makecell{194.19 \\[0.6ex] (3.73); 0\%} & \makecell{187.73 \\[0.6ex] (4.19); 0\%} \\ 
  \makecell{Gaussian REML} & \makecell{135.48 \\[0.6ex] (2.36); 9\%} & \makecell{124.13 \\[0.6ex] (2.09); 4\%} & \makecell{138.73 \\[0.6ex] (2.28); 1\%} & \makecell{126.94 \\[0.6ex] (2.37); 24\%} \\ 
  \makecell{FGM REML} & \makecell{134.61 \\[0.6ex] (2.19); 0\%} & \makecell{123.56 \\[0.6ex] (2.01); 0\%} & \makecell{138.35 \\[0.6ex] (2.29); 0\%} & \makecell{124.82 \\[0.6ex] (2.03); 0\%} \\ 
   \hline
\end{tabular}
\par\vspace{0.8em}
\noindent\text{(c) When $\sigma_e = 2$}\par\smallskip
\begin{tabular}{ccccc}
  \hline
Method & Exponential & Matérn & Linear Matérn & Spherical \\ 
  \hline
\makecell{Exponential REML} & \textbf{\makecell{447.22 \\[0.6ex] (7.32); 0\%}} & \textbf{\makecell{440.23 \\[0.6ex] (7.18); 0\%}} & \textbf{\makecell{455.75 \\[0.6ex] (7.45); 0\%}} & \textbf{\makecell{437.29 \\[0.6ex] (7.12); 0\%}} \\ 
  \makecell{Matérn REML} & \makecell{455.50 \\[0.6ex] (9.62); 33\%} & \makecell{443.41 \\[0.6ex] (9.65); 40\%} & \makecell{456.29 \\[0.6ex] (8.37); 37\%} & \makecell{451.17 \\[0.6ex] (9.75); 43\%} \\ 
  \makecell{Matérn 1.5 REML} & \makecell{450.95 \\[0.6ex] (7.61); 0\%} & \makecell{444.09 \\[0.6ex] (7.64); 0\%} & \makecell{458.65 \\[0.6ex] (7.64); 0\%} & \makecell{441.41 \\[0.6ex] (7.51); 0\%} \\ 
  \makecell{Spherical REML} & \makecell{482.51 \\[0.6ex] (8.23); 0\%} & \makecell{488.35 \\[0.6ex] (8.48); 0\%} & \makecell{492.03 \\[0.6ex] (8.34); 0\%} & \makecell{485.16 \\[0.6ex] (8.51); 0\%} \\ 
  \makecell{Gaussian REML} & \makecell{448.95 \\[0.6ex] (8.38); 17\%} & \makecell{442.73 \\[0.6ex] (7.72); 9\%} & \makecell{456.73 \\[0.6ex] (7.80); 4\%} & \makecell{439.25 \\[0.6ex] (8.18); 16\%} \\ 
  \makecell{FGM REML} & \makecell{447.38 \\[0.6ex] (7.28); 0\%} & \makecell{440.68 \\[0.6ex] (7.18); 0\%} & \makecell{456.32 \\[0.6ex] (7.45); 0\%} & \makecell{437.38 \\[0.6ex] (7.10); 0\%} \\ 
   \hline
\end{tabular}
\end{table}

\subsection{Additional FSW Data Analysis Results}
\label{sub: appendix_real_data}

Figures~\ref{fig:gaussianity}-~\ref{fig:grid_wls} and Tables~\ref{tab:coef_comparison_methods_without}-~\ref{tab:computation_times} provide additional results related to the FSW data analysis.

\begin{figure}[htbp!]
  \centering
  \subfigure[Histogram of logit-proportions]{%
    \includegraphics[width=0.48\textwidth]{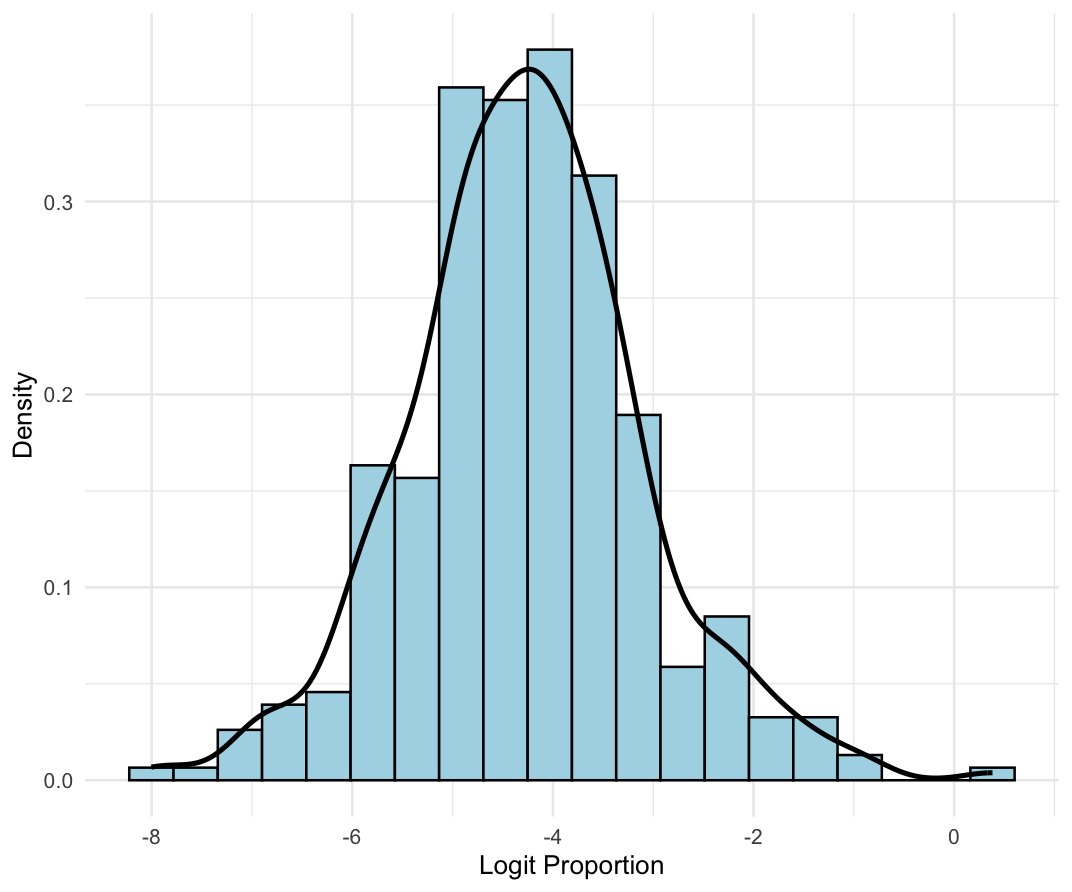}%
    \label{fig:histogram_logit_proportions}%
  }
  \hfill
  \subfigure[Normal Q-Q plot of logit-proportions]{%
    \includegraphics[width=0.48\textwidth]{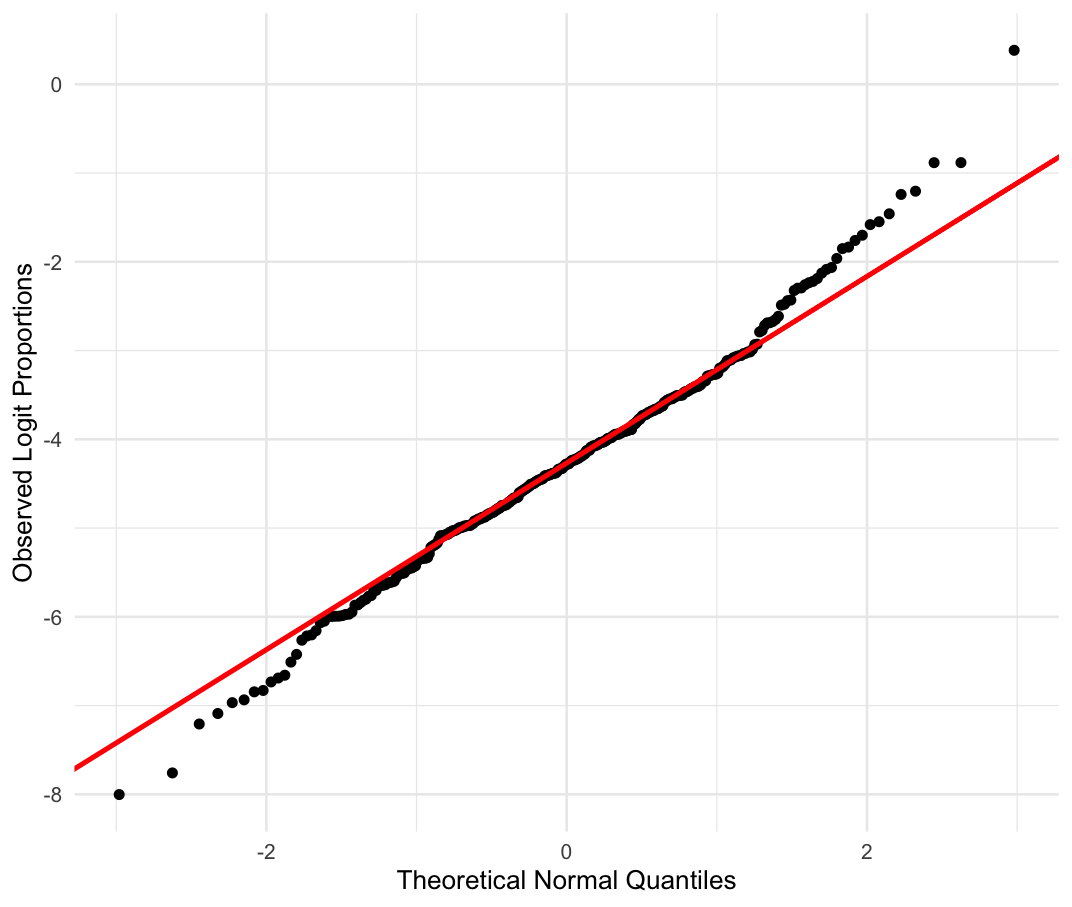}%
    \label{fig:Qqplot_logit_proportions}%
  }
  \caption{Assessing the Gaussianity of logit-transformed proportions using the: 
    (a) histogram and
    (b) normal Q-Q plot.}
  \label{fig:gaussianity}
\end{figure}

\begin{table}[ht]
\centering
\caption{Coefficient estimates and their standard errors (in parentheses) across methods based on data without outliers}
\label{tab:coef_comparison_methods_without}
\begin{tabular}{lcccc}
\hline
& OLS & ML & CAR-INLA \\
\hline
Intercept 
& -4.137 (0.054) 
& -4.070 (0.163) 
& -4.091 (0.053) \\

Urban Population Percentage 
& 0.229 (0.061) 
& 0.216 (0.066) 
& 0.127 (0.079) \\

Reference Population
& -0.583 (0.045) 
& -0.538 (0.044) 
& -0.564 (0.053) \\
\hline
\end{tabular}
\end{table}

% latex table generated in R 4.6.0 by xtable 1.8-8 package
% Sun Jul 26 17:32:51 2026
\begin{table}[ht]
\centering
\caption{Mean of MSPEs and its standard error (in parentheses) for the four parametric methods (exponential, Mat\'{e}rn, spherical and Gaussian) with initial value for $\sigma_e = 10^{-1}$ and the two nonparametric methods (B-splines method with order $p=3$ and number of knots $m=20$ and finite Gaussian mixtures method with number of grid points $s=100$) using WLS estimation from five-fold cross-validation. Values on the logit scale are reported in units of $10^{-2}$, whereas values on the probability scale are reported in units of $10^{-4}$.} 
\label{tab:mspe_realdata_WLS1}
\begin{tabular}{lcccc}
  \multicolumn{1}{c}{Method} & \multicolumn{2}{c}{Logit Scale} & \multicolumn{2}{c}{Probability Scale} \\

\cmidrule(lr){2-3} \cmidrule(lr){4-5}

 & With Outliers & Without Outliers & With Outliers & Without Outliers \\
 \hline
  \hline
Exponential WLS & \textbf{96.91 (10.76)} & 94.37 (5.33) & 20.13 (6.89) & 20.08 (6.75) \\ 
  Mat\'{e}rn WLS & 101.06 (10.94) & 101.03 (6.66) & 20.41 (6.76) & 20.85 (6.49) \\ 
  Mat\'{e}rn 1.5 WLS & 102.30 (11.68) & 101.50 (6.18) & 20.52 (6.76) & 20.83 (6.48) \\ 
  Spherical WLS & 97.46 (10.92) & 93.54 (5.29) & \textbf{20.07 (6.92)} & 19.97 (6.83) \\ 
  Gaussian WLS & 119.46 (13.21) & 116.24 (7.76) & 22.28 (7.38) & 22.33 (7.30) \\ 
  Bspl WLS & 343847.28 (342629.93) & 9081.37 (7853.55) & 409.02 (194.56) & 305.22 (107.86) \\ 
  Bspl WLS Nugget & 93180.80 (93091.02) & \textbf{86.84 (5.18)} & 181.13 (158.79) & \textbf{19.84 (6.92)} \\ 
  FGM WLS & 784874.54 (684999.26) & 488684.56 (488572.02) & 534.36 (344.11) & 369.74 (347.24) \\ 
   \hline
\end{tabular}
\end{table}

% latex table generated in R 4.6.0 by xtable 1.8-8 package
% Sun Jul 26 17:33:17 2026
\begin{table}[ht]
\centering
\caption{Mean of MSPEs and its standard error (in parentheses) for the four parametric methods (exponential, Mat\'{e}rn, spherical and Gaussian) with initial value for $\sigma_e = 10^{-1}$ and nonparametric method (finite Gaussian mixtures method with number of grid points $s=100$) using REML estimation from five-fold cross-validation. Values on the logit scale are reported in units of $10^{-2}$, whereas values on the probability scale are reported in units of $10^{-4}$.} 
\label{tab:mspe_realdata_REML1}
\begin{tabular}{lcccc}
  \multicolumn{1}{c}{Method} & \multicolumn{2}{c}{Logit Scale} & \multicolumn{2}{c}{Probability Scale} \\

\cmidrule(lr){2-3} \cmidrule(lr){4-5}

 & With Outliers & Without Outliers & With Outliers & Without Outliers \\
 \hline
  \hline
Exponential REML & \textbf{86.01 (9.12)} & \textbf{81.66 (4.23)} & 20.32 (7.21) & 20.21 (7.17) \\ 
  Mat\'{e}rn REML & 90.10 (12.43) & 86.22 (7.40) & 20.50 (7.28) & 20.37 (7.20) \\ 
  Mat\'{e}rn 1.5 REML & 90.03 (12.47) & 85.77 (7.39) & 20.49 (7.27) & 20.34 (7.21) \\ 
  Spherical REML & 95.55 (10.38) & 92.32 (4.76) & 20.34 (7.07) & \textbf{20.19 (6.96)} \\ 
  Gaussian REML & 102.25 (11.14) & 94.21 (5.58) & \textbf{13.76 (3.22)} & 21.42 (8.78) \\ 
  FGM REML & 88.15 (9.29) & 82.39 (4.27) & 20.55 (7.27) & 20.31 (7.15) \\ 
   \hline
\end{tabular}
\end{table}

% latex table generated in R 4.6.0 by xtable 1.8-8 package
% Sun Jul 26 17:37:02 2026
\begin{table}[ht]
\centering
\caption{Mean of MSPEs and its standard error (in parentheses) for the CAR-INLA model and the four parametric methods (exponential, Mat\'{e}rn, spherical and Gaussian) with initial value for $\sigma_e = 0$ and the two nonparametric methods (B-splines method with order $p=3$ and number of knots $m=20$ and finite Gaussian mixtures method with number of grid points $s=100$) using ML estimation from five-fold cross-validation. Values on the logit scale are reported in units of $10^{-2}$, whereas values on the probability scale are reported in units of $10^{-4}$.} 
\label{tab:mspe_realdata_ML2}
\begin{tabular}{lcccc}
  \multicolumn{1}{c}{Method} & \multicolumn{2}{c}{Logit Scale} & \multicolumn{2}{c}{Probability Scale} \\

\cmidrule(lr){2-3} \cmidrule(lr){4-5}

 & With Outliers & Without Outliers & With Outliers & Without Outliers \\
 \hline
  \hline
CAR INLA & 97.63 (10.43) & 94.06 (5.00) & 20.63 (7.35) & 20.66 (7.35) \\ 
  Exponential ML & 89.00 (11.82) & 85.77 (7.40) & 20.45 (7.21) & 20.32 (7.18) \\ 
  Mat\'{e}rn ML & 89.95 (11.99) & 86.14 (9.45) & 20.54 (7.28) & 23.11 (8.61) \\ 
  Mat\'{e}rn 1.5 ML & 89.49 (11.90) & 85.80 (7.36) & 20.52 (7.28) & 20.36 (7.22) \\ 
  Spherical ML & 96.81 (10.16) & 92.29 (4.76) & 20.33 (7.06) & \textbf{20.19 (6.96)} \\ 
  Gaussian ML & 101.36 (12.30) & 94.94 (6.26) & \textbf{13.66 (3.24)} & 21.39 (8.78) \\ 
  Bspl ML & 1027.99 (439.13) & 1048.64 (452.92) & 189.89 (88.30) & 190.79 (89.70) \\ 
  FGM ML & \textbf{88.17 (9.38)} & \textbf{82.42 (4.26)} & 20.56 (7.28) & 20.32 (7.16) \\ 
   \hline
\end{tabular}
\end{table}

% latex table generated in R 4.6.0 by xtable 1.8-8 package
% Sun Jul 26 17:37:45 2026
\begin{table}[ht]
\centering
\caption{Mean of MSPEs and its standard error (in parentheses) for the four parametric methods (exponential, Mat\'{e}rn, spherical and Gaussian) with initial value for $\sigma_e = 0$ and the two nonparametric methods (B-splines method with order $p=3$ and number of knots $m=20$ and finite Gaussian mixtures method with number of grid points $s=100$) using WLS estimation from five-fold cross-validation. Values on the logit scale are reported in units of $10^{-2}$, whereas values on the probability scale are reported in units of $10^{-4}$.} 
\label{tab:mspe_realdata_WLS2}
\begin{tabular}{lcccc}
  \multicolumn{1}{c}{Method} & \multicolumn{2}{c}{Logit Scale} & \multicolumn{2}{c}{Probability Scale} \\

\cmidrule(lr){2-3} \cmidrule(lr){4-5}

 & With Outliers & Without Outliers & With Outliers & Without Outliers \\
 \hline
  \hline
Exponential WLS & 102.37 (12.06) & 93.82 (4.85) & \textbf{13.69 (3.20)} & 20.07 (6.75) \\ 
  Mat\'{e}rn WLS & 103.15 (10.43) & 103.47 (6.63) & 20.59 (6.69) & 21.13 (6.34) \\ 
  Mat\'{e}rn 1.5 WLS & 104.23 (11.15) & 104.32 (6.50) & 20.67 (6.69) & 21.18 (6.33) \\ 
  Spherical WLS & \textbf{97.56 (10.92)} & 93.62 (5.29) & 20.07 (6.92) & 19.97 (6.83) \\ 
  Gaussian WLS & 1867.38 (706.54) & 5969.37 (3727.18) & 201.41 (66.19) & 300.02 (78.54) \\ 
  Bspl WLS & 343847.28 (342629.93) & 9081.37 (7853.55) & 409.02 (194.56) & 305.22 (107.86) \\ 
  Bspl WLS Nugget & 93180.80 (93091.02) & \textbf{86.84 (5.18)} & 181.13 (158.79) & \textbf{19.84 (6.92)} \\ 
  FGM WLS & 784874.54 (684999.26) & 488684.56 (488572.02) & 534.36 (344.11) & 369.74 (347.24) \\ 
   \hline
\end{tabular}
\end{table}

% latex table generated in R 4.6.0 by xtable 1.8-8 package
% Sun Jul 26 17:38:10 2026
\begin{table}[ht]
\centering
\caption{Mean of MSPEs and its standard error (in parentheses) for the four parametric methods (exponential, Mat\'{e}rn, spherical and Gaussian) with initial value for $\sigma_e = 0$ and nonparametric method (finite Gaussian mixtures method with number of grid points $s=100$) using REML estimation from five-fold cross-validation. Values on the logit scale are reported in units of $10^{-2}$, whereas values on the probability scale are reported in units of $10^{-4}$.} 
\label{tab:mspe_realdata_REML2}
\begin{tabular}{lcccc}
  \multicolumn{1}{c}{Method} & \multicolumn{2}{c}{Logit Scale} & \multicolumn{2}{c}{Probability Scale} \\

\cmidrule(lr){2-3} \cmidrule(lr){4-5}

 & With Outliers & Without Outliers & With Outliers & Without Outliers \\
 \hline
  \hline
Exponential REML & \textbf{86.01 (9.12)} & 85.72 (7.44) & 20.32 (7.21) & 20.32 (7.18) \\ 
  Mat\'{e}rn REML & 86.41 (8.99) & 86.22 (7.40) & 20.41 (7.28) & 20.37 (7.20) \\ 
  Mat\'{e}rn 1.5 REML & 89.43 (11.90) & 85.77 (7.39) & 20.52 (7.27) & 20.34 (7.21) \\ 
  Spherical REML & 96.83 (10.17) & 92.34 (4.76) & 20.34 (7.06) & \textbf{20.21 (6.96)} \\ 
  Gaussian REML & 101.15 (12.10) & 94.95 (6.26) & \textbf{13.73 (3.29)} & 21.39 (8.78) \\ 
  FGM REML & 88.15 (9.29) & \textbf{82.39 (4.27)} & 20.55 (7.27) & 20.31 (7.15) \\ 
   \hline
\end{tabular}
\end{table}

\begin{figure}[h!]
  \centering
  \subfigure[With outliers]{%
    \includegraphics[width=0.45\textwidth]{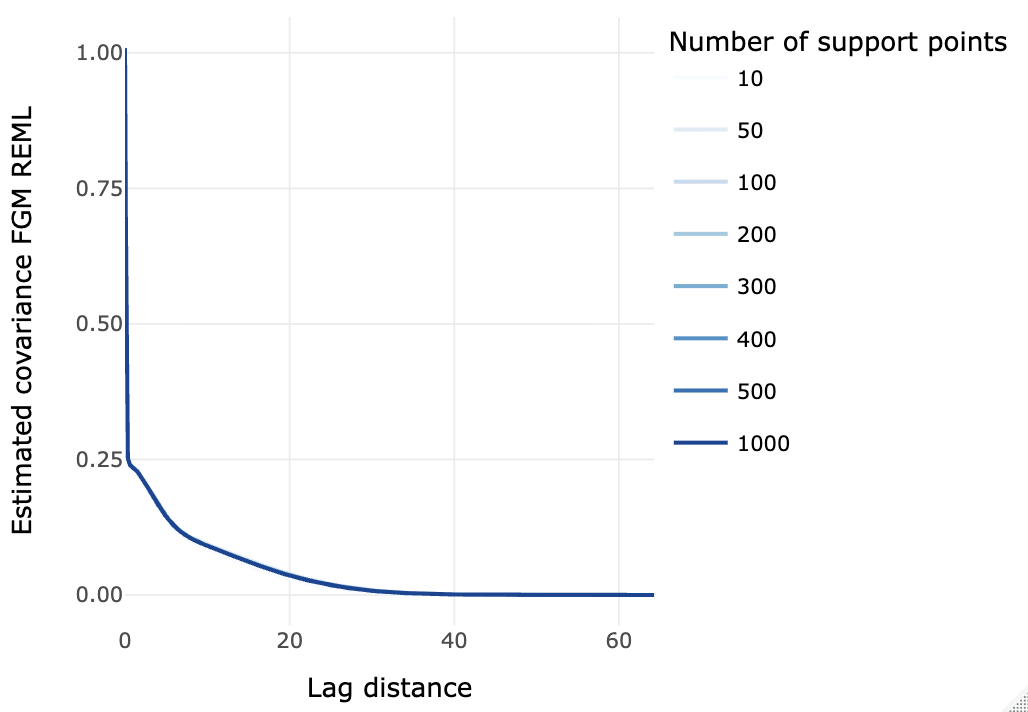}%
    \label{fig:grid_with_reml}%
  }
  \hfill
  \subfigure[Without outliers]{%
    \includegraphics[width=0.45\textwidth]{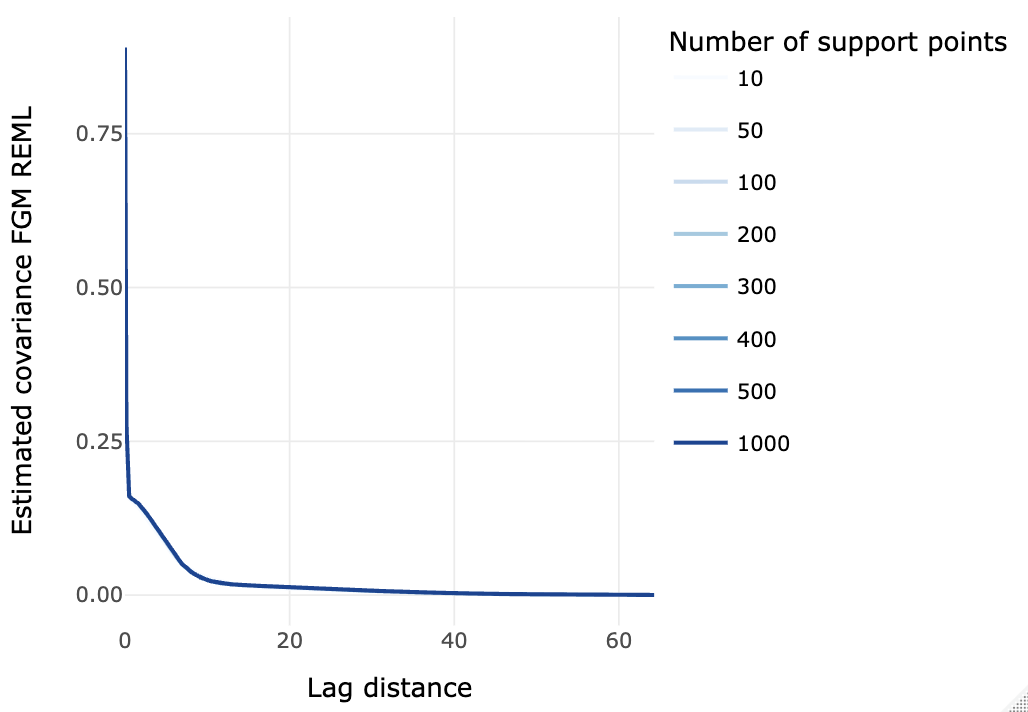}%
    \label{fig:grid_without_reml}%
  }
  \caption{Estimated covariances using the nonparametric method, FGM using REML estimation over different grid sizes: 
    (a) using the dataset with outliers, 
    (b) using the dataset without outliers.}
  \label{fig:grid_reml}
\end{figure}

\begin{figure}[h!]
  \centering
  \subfigure[With outliers]{%
    \includegraphics[width=0.45\textwidth]{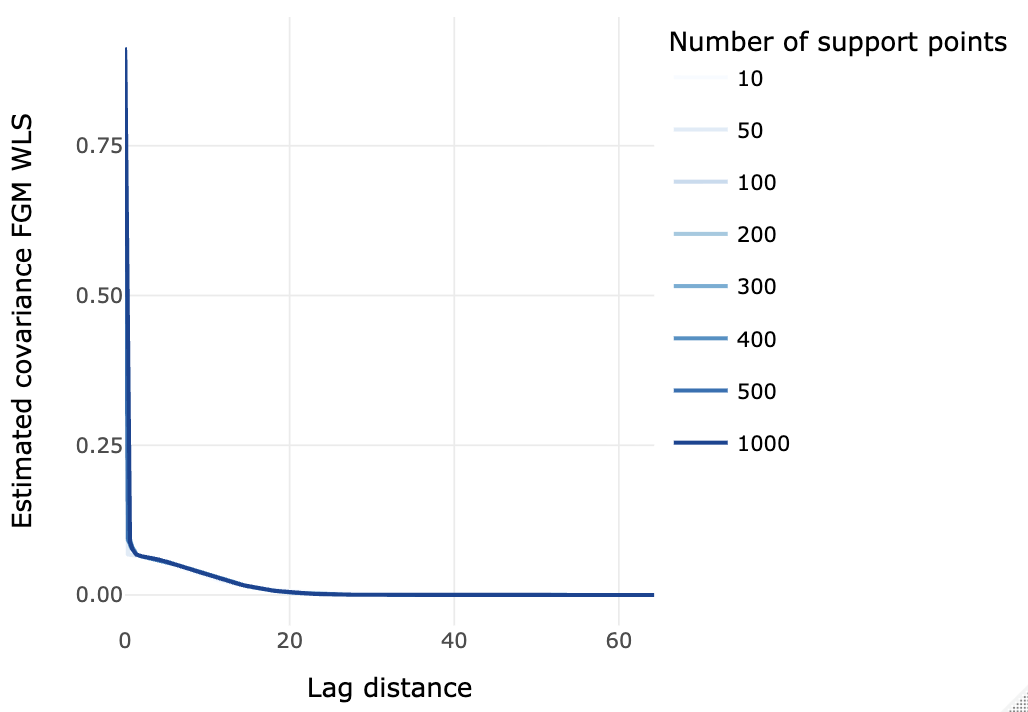}%
    \label{fig:grid_with_wls}%
  }
  \hfill
  \subfigure[Without outliers]{%
    \includegraphics[width=0.45\textwidth]{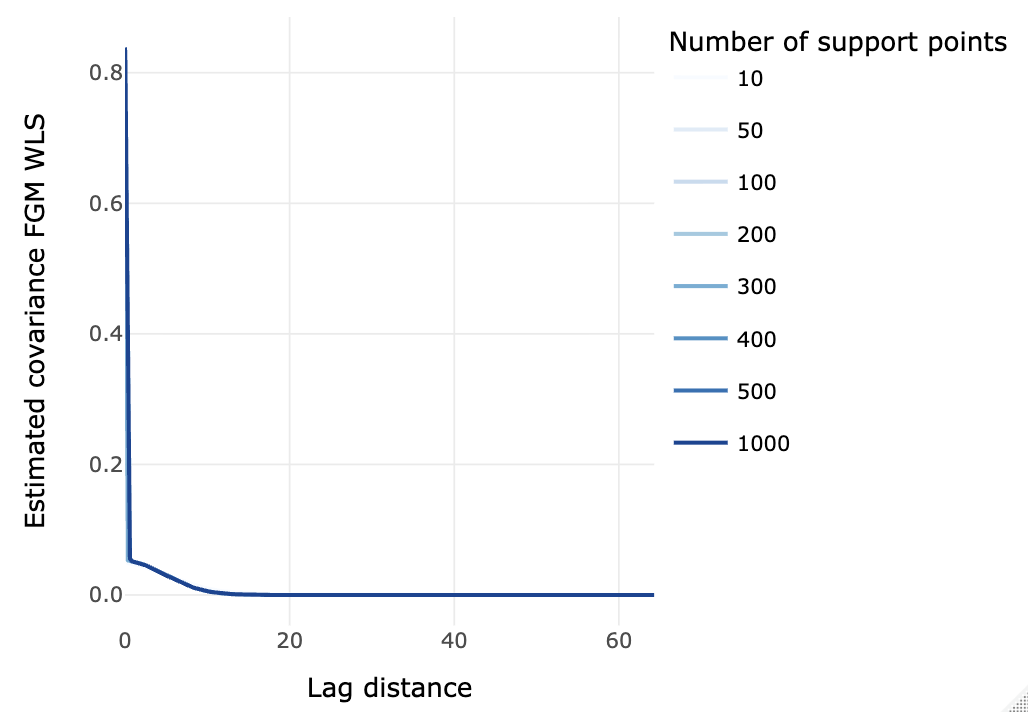}%
    \label{fig:grid_without_wls}%
  }
  \caption{Estimated covariances using the nonparametric method, FGM using WLS estimation over different grid sizes: 
    (a) using the dataset with outliers, 
    (b) using the dataset without outliers.}
  \label{fig:grid_wls}
\end{figure}

\begin{table}[htbp]
\centering
\footnotesize
\caption{MSPEs of the proposed nonparametric method using ML, REML, and WLS estimation methods under different numbers of support points using datasets with and without outliers on the logit (unit: $10^{-2}$) and probability (unit: $10^{-4}$) scales.}
\label{tab:mspe_support_points}
\begin{tabular}{c ccc ccc ccc ccc}
\hline
& \multicolumn{6}{c}{Logit Scale ($10^{-2}$)} 
& \multicolumn{6}{c}{Probability Scale ($10^{-4}$)} \\
\cline{2-7} \cline{8-13}
\text{Support Points}
& \multicolumn{3}{c}{With Outliers}
& \multicolumn{3}{c}{Without Outliers}
& \multicolumn{3}{c}{With Outliers}
& \multicolumn{3}{c}{Without Outliers} \\
\cline{2-4} \cline{5-7} \cline{8-10} \cline{11-13}
& ML & REML & WLS
& ML & REML & WLS
& ML & REML & WLS
& ML & REML & WLS \\
\hline
10
& 100.74 & 100.68 & 121.27
& 104.31 & 103.44 & 112.21
& 47.33 & 47.66 & 46.39
& 20.91 & 20.82 & 21.41 \\

50
& 100.78 & 100.71 & 110.78
& 104.80 & 103.84 & 112.34
& 47.40 & 47.67 & 45.39
& 21.05 & 20.95 & 21.49 \\

100
& 100.74 & 100.69 & 107.64
& 104.85 & 103.84 & 112.33
& 47.35 & 47.66 & 45.12
& 21.06 & 20.95 & 21.49 \\

200
& 100.75 & 100.69 & 121.17
& 104.86 & 103.84 & 201.41
& 47.36 & 47.65 & 46.67
& 21.06 & 20.95 & 161.91 \\

300
& 100.74 & 100.68 & 108.72
& 104.86 & 103.85 & 112.34
& 47.36 & 47.64 & 45.23
& 21.06 & 20.95 & 21.49 \\

400
& 100.74 & 100.68 & 107.73
& 104.85 & 103.85 & 208.43
& 47.36 & 47.64 & 45.14
& 21.06 & 20.95 & 178.53 \\

500
& 100.75 & 100.68 & 108.51
& 104.86 & 103.86 & 198.87
& 47.36 & 47.65 & 45.20
& 21.06 & 20.95 & 155.29 \\

1000
& 100.75 & 100.69 & 120.74
& 104.86 & 103.85 & 201.52
& 47.36 & 47.65 & 46.47
& 21.06 & 20.95 & 162.17 \\
\hline
\end{tabular}
\end{table}

\begin{table}[htbp]
\centering
\caption{Computation times (in seconds) for the proposed nonparametric method using ML, REML, and WLS estimation methods using datasets with and without outliers.}
\label{tab:computation_times}
\begin{tabular}{c ccc ccc}
\hline
& \multicolumn{3}{c}{With Outliers}
& \multicolumn{3}{c}{Without Outliers} \\
\cline{2-4} \cline{5-7}
\text{Support Points}
& ML & REML & WLS
& ML & REML & WLS \\
\hline
10
& 1.52 & 1.34 & 0.001
& 1.16 & 1.11 & 0.001 \\

50
& 4.74 & 5.06 & 0.002
& 3.69 & 3.71 & 0.001 \\

100
& 9.37 & 8.82 & 0.002
& 8.63 & 7.95 & 0.002 \\

200
& 23.44 & 22.27 & 0.002
& 18.74 & 17.18 & 0.002 \\

300
& 39.62 & 35.73 & 0.004
& 24.90 & 27.37 & 0.003 \\

400
& 46.30 & 45.34 & 0.003
& 41.39 & 37.20 & 0.004 \\

500
& 58.42 & 57.58 & 0.004
& 43.46 & 48.25 & 0.004 \\

1000
& 122.39 & 124.87 & 0.007
& 88.79 & 100.45 & 0.007 \\
\hline
\end{tabular}
\end{table}

\end{document}